\documentclass[10pt]{article}
\usepackage[top=4cm, bottom=3cm, left=3.2cm, right=3.2cm]{geometry}
\usepackage{enumerate}
\usepackage{amssymb,amsmath}
\usepackage{graphicx}
\usepackage{subfigure}

\usepackage{amsthm}
\usepackage{amsmath}
\usepackage{amssymb}

\usepackage{extarrows}
\usepackage{epstopdf}
\usepackage{xcolor} 
\usepackage{bm}
\usepackage{geometry} 
\usepackage{graphicx} 

\usepackage{booktabs} 
\usepackage{array} 
\usepackage{paralist} 
\usepackage{verbatim} 
\usepackage{caption} 
\usepackage{cases}

\usepackage{multirow}
\usepackage{graphicx}
\usepackage{float}
\usepackage{subfigure}
\usepackage{algorithmicx,algorithm}
\usepackage{booktabs}
\usepackage{psfrag}
\usepackage[noend]{algpseudocode}

\newtheorem{thm}{Theorem}
\newtheorem{cor}{Corollary}
\newtheorem{lem}{Lemma}
\newtheorem{prop}{Proposition}

\newtheorem{rem}{Remark}
\allowdisplaybreaks

\newcommand{\dint}{\displaystyle\int}

\newcommand{\F}{\mathcal{F}}

\newcommand{\br}{\boldsymbol{r}}

\numberwithin{equation}{section} \numberwithin{lem}{section}
\numberwithin{thm}{section} \numberwithin{prop}{section}
\numberwithin{cor}{section} \numberwithin{rem}{section}

\title{\bf{On the equilibrium of the Poisson-Nernst-Planck-Bikermann model equipping with the steric and correlation effects}}
\begin{document}
\author{Jian-Guo Liu$^{\,1}$ and
	Yijia Tang$^{\,2}$ and Yu Zhao$^{\,3}$\footnote{Corresponding author: Yu Zhao.}}


\maketitle
\begin{center}
{\footnotesize
1-Department of Mathematics and Department of Physics, Duke University, Durham, NC 27708, USA. \\
email: jliu@phy.duke.edu \\
\smallskip
2-School of Mathematical Sciences, Shanghai Jiao Tong University, Shanghai, 200240, P. R. China. \\
email: yijia\underline{~}tang@sjtu.edu.cn \\
\smallskip
3-School of Mathematical Sciences, Peking University, Beijing, 100871, P. R. China. \\  
email: y.zhao@pku.edu.cn 
}
\end{center}

\begin{abstract}
The Poisson-Nernst-Planck-Bikermann (PNPB) model, in which the ions and water molecules are treated as different species with non-uniform sizes and valences with interstitial voids,
can describe the steric and correlation effects in ionic solution neglected by the Poisson-Nernst-Planck and Poisson-Boltzmann theories with point charge assumption.
In the PNPB model, the electric potential is governed by the fourth-order Poisson-Bikermann (4PBik) equation instead of the Poisson equation so that it can describe the correlation effect. What's more, the steric potential is included in the ionic and water fluxes as well as the equilibrium Fermi-like distributions which characterizes the steric effect quantitatively. 

In this work, after doing a nondimensionalization step, we analyze the self-adjointness and the kernel of the fourth-order operator of the 4PBik equation. Also, we show the positivity of the void volume function and the convexity of the free energy. Following these properties, the well-posedness of the PNPB model in equilibrium is given. 
Furthermore, because the PNPB model has an energy dissipated structure, we adopt a finite volume scheme which preserves the energy dissipated property at the semi-discrete level. After that, various numerical investigations are given to show the parameter dependence of the steric effect to the steady state.
\end{abstract}	


\section{Introduction}
\label{sec: intro}

The transport and distribution of ions are crucial in the study of many physical and biological problems, such as ion particles in the electrokinetic fluids, and ion channels in cell membranes. From the perspective of mathematics, some classical models have been proposed to portray the ion transport in solutions dating back to the 19th century \cite{Nernst,Planck}. The Poisson-Boltzmann (PB) equation is proposed by Gouy \cite{Gouy} and Chapman \cite{Chapman} to describe the equilibrium of the ionic fluids, while the Poisson-Nernst-Planck (PNP) equation is to describe the dynamics \cite{EisenbergChen93}. The PNP and PB theories are mean-field descriptions, where the ions are volumeless point charges, water is the dieletric medium, and both are modelled as continuum distributions.
We refer the readers to \cite{Davis1990Electrostatics,FogolariBrigoMolinari02,Baker05,Lu2008Recent} for the reviews of the PB equation and \cite{Eisenberg98,MarkowichRinghoferSchimeiser90,Flavell2014A} for those of the PNP equation. 

At the equilibrium state, when considering a system of $K$ different ionic species, the classical PB theory shows that the $i^{\text{th}}$ ionic concentrations $C_i$ for $i = 1, \cdots, K$, are demonstrated by the Boltzmann distributions, i.e. 
\begin{equation}
   C_i (\br)=C_i^{\mathrm{B}}\exp\left(-\frac{q_i}{k_B T}\phi(\br)\right), \quad i = 1, \cdots, K,
   \label{eq: Boltzmann}
\end{equation}
while the electric potential $\phi$ determined by Gauss's law satisfies the Poisson equation, i.e. 
\begin{equation}
    -\varepsilon_s \Delta \phi(\br) = \rho(\br),
    \label{eq: poisson}
\end{equation}
where $\rho(\br) := \sum_{i = 1}^{K} q_i C_i(\br)$ is called the total charged density.
Here, $\varepsilon_s$ is the dielectric constant of the solvent, $k_B$ is the Boltzmann constant and $T$ is the absolute temperature. $q_i=z_i e$ is the charge of the $i^{\text{th}}$ species with valence $z_i$, $e$ is the proton charge. $C_i^{\mathrm{B}}>0$ is the reference concentration often chosen as the bulk concentration. 

As is well-known, the PB model can only describe the diffusion and electric effect. So many improvements have been made to demonstrate nonideal effects. On one hand, to describe the correlation (or nonlocal screening, polarization) effect, in 2006, Santangelo \cite{Santangelo06} refined the second-order Poisson equation into a fourth-order equation for the equilibrium profile of the pointwise counterions near a charged wall, so that the correlation effect is modelled by introducing a parameter of correlation length $l_c$ in the fourth-order dielectric operator $L:=\varepsilon_s\left(l_c^2\Delta-I\right)\Delta$. Since then, $L$ is used to describe the medium permittivity and the dielectric response of correlated ions \cite{BazantStoreyKornyshev11,LiuEisenberg13,LiuEisenberg14,LiuXieEisenberg17,LiuEisenberg20}. 

On the other hand, due to the limitation of the particles' volumeless assumption, the PB theory fails to describe the steric (or finite size) effect which has been shown to be important in a variety of chemical and biological systems \cite{2001Ion,BAZANT200948,exchanger,nanopore}. The PB equation models the mean-field electrostatic potential where the ions interact through Coulomb force, 
which is independent of volume of the particles. 
A lot of efforts have been made to improve the Boltzmann distribution for a proper description incorporating the steric effect. The first trial for binary ionic liquids with different sizes and voids was made by Bikerman in 1942. 
Then Bazant et al. \cite{BazantStoreyKornyshev11} proposed a significant PB modification in 2011, where the size of ions, which is assumed to be equal, is considered in a Fermi-like concentration distribution in place of the Boltzmann distribution.
Here Fermi-like distributions of all species are derived from the volume exclusion of classical particles. The excluded volume of particles leads to the space competition, which produces a steric energy.
Recently, Liu and Eisenberg \cite{LiuEisenberg13,LiuEisenberg14} gave the microscopic interpretation of the Fermi-like distributions using the configuration entropy model. They derived the Poisson-Fermi (PF) model to describe the equilibrium of aqueous electrolytes with $K$ ionic species of non-uniform size \cite{LiuEisenberg13} and extended it to the case where the water molecules and the interstitial voids are treated as the $K+1$ and $K+2$ species \cite{LiuEisenberg14}. Later, a slight modification was made in the Fermi-like distribution such that the steric energy varies with species size as the electric energy varies with species valence \cite{LiuXieEisenberg17}.
In this work, we focus on this novel model developed by Liu and Eisenberg, called the Poisson-Nernst-Planck-Bikerman (PNPB) model to cover the non-equilibrium state at the same time. 

In the PNPB model, all the ions and water molecules are assumed as hard spheres with volume $v_i$, hence the voids occupy the exclusive volume. 
Then the Fermi-like distributions of concentrations $C_i$ in equilibrium, for $i = 1, \cdots, K+1$ are given by 
\begin{equation}
\label{eq: fermi}
C_{i}(\br)=C_{i}^{\mathrm{B}} \exp \left(-\frac{q_i}{k_BT} \phi(\br)+\frac{v_{i}}{v_{0}}S(\br)\right),
\end{equation}
and the electrostatic field is given by the fourth-order Poisson-Bikerman (4PBik) equation 
\begin{equation}
\label{eq: 4pbik}
L\phi(\br)=\varepsilon_{s}\left(l_{c}^{2} \Delta-I\right) \Delta \phi(\br)=\rho(\br).
\end{equation}
Here 
$v_{0}=\left(\sum_{i=1}^{K+1} v_{i}\right) /(K+1)$ is the average volume. 
$S$ in \eqref{eq: fermi} is called the steric potential, which keeps the ions from getting too close and leads to the charge/space competition \cite{Boda_et_al}. 
The steric potential $S$ is given in the following form
\begin{equation}
\label{eq: steric}
S=\ln{\frac{\Gamma(\br)}{\Gamma^{\mathrm{B}}}},
\end{equation}
where the void volume function 
\begin{equation}
\Gamma(\br) =1-\sum_{i=1}^{K+1} v_{i} C_{i}(\br)
\label{eq: gamma}
\end{equation}
is the fraction left by the excluded volume of ions and water. $\Gamma^{{\mathrm{B}}}>0$ is the reference void volume fraction given by $\Gamma^{{\mathrm{B}}}=1-\sum_{i=1}^{K+1} v_{i} C_{i}^{{\mathrm{B}}}$. \eqref{eq: fermi} indicates the void interacts with the ions and water molecules as a repulsive force and results in the steric energy like the Lennard-Jones in \cite{EnVarA}.
Steric effects are seen in the resulting Fermi-like distribution of finite size ions at large potentials near electrodes and boundaries. 
And the electric potential 
is governed by the 4PBik equation \eqref{eq: 4pbik} instead of the Poisson equation \eqref{eq: poisson}. 
Hence, they generalize the equilibrium model to the non-equilibrium case for ion transport \cite{LiuEisenberg14}, which is the PNPB system. The non-equilibrium PNPB system is coupled by the Nernst-Planck-Bikerman (NPB) equation with the 4PBik equation. The former is a generalization of the classic Nernst-Planck equation in the sense that the motion of ions and water molecules is driven by the steric force besides the electric force and concentration gradient. 
In dynamical case, if the ions are crowded, the steric force $\nabla S$ would spread them out. Hence, the PNPB model are capable to describe the ionic fluids with the steric effect and correlation effect. 
See \cite{LiuEisenberg20} for a comprehensive review of the PNPB theory.
We remark that the energy variational analysis (EnVarA) proposed by Eisenberg, Hyon, and Liu \cite{EnVarA} is another approach to describe the finite size effect and steric potential for ion transport, which combines Hamilton's least action and Rayleigh's dissipation principles to create a variational field theory \cite{EnVarA,EnVarA2,EnVarA3}.


Though providing a unified framework, the PNPB theory is constructed merely based on physical principles, and there are still many mathematical problems requiring to be answered. The first one is the positivity of the volume function of the interstitial voids. 
If $\Gamma$ is not greater than zero, the steric potential $S$ makes no sense. So, it is vital to show the positivity of $\Gamma$ in steady state problem and the positivity-preserving property of $\Gamma$ in the dynamical problem theoretically. 
Another issue is how to prescribe suitable and physical boundary conditions for the PNPB model to describe various actual physical and chemical problems.
What's more, to show the existence, uniqueness and regularity of the solutions to both the equilibrium and non-equilibrium PNPB model remains a major concern.

In this work, we want to give more mathematical understandings and improve the theoretical framework on the PNPB model by some analytical approaches and numerical investigations. 
On the theoretical side, we try to tackle some of the difficulties mentioned before. We mainly focus on the equilibrium problem in this paper. 
First of all, we look into the basic properties of operator $L$ in the 4PBik equation \eqref{eq: 4pbik}, such as the self-adjointness and its kernel structure. Then we show the positivity of $\Gamma$ in equilibrium. 
To prove this rigorously, we first give four equivalent statements for the steady state and then deduce the positivity by contradiction using one of them. As a consequence of this, it is impossible for a volume to be completely filled with ions or water molecules. Hence, the presence of void gives rise to the saturation of ions. 
Thanks to the convexity of the free energy functional, we can demonstrate the well-posedness of the weak solution using standard calculus of variation. 

Besides, though the steric, correlation, saturation and other effects have already been shown in the applications of PNPB model to the biological ion channels in \cite{Liu13,LiuEisenberg15,exchanger,LiuEisenberg20}, we want to investigate these effects from the perspective of parameter dependence. 
After doing nondimensionalization, the PNPB model is tuned by five dimensionless parameters $\eta, \lambda, \nu, z_i$ and $v_i$. 
The parameters can be divided into three classes: the one related to the steric effect: $\eta$, the one determining the electric field: $\lambda, \nu$, and the one that is the intrinsic properties of ions: $z_i, v_i$. 
We aim to gain more understandings on how these parameters influence the concentrations of the ions at the steady state numerically. 
By doing various numerical tests, we demonstrate that $\eta$ dominates the steric repulsion effect, $\eta = 0$ means the steric effect vanishes, thus the concentrations of ionic species are more peaked than the case $\eta > 0$. Also, the steric effect becomes stronger as $\eta$ increases. The dimensionless parameters $\lambda$ and $\nu$ play the role of correlation length and Debye length after scaling, which contribute oppositely to the electric potential. We conclude that the electric potential is higher with smaller correlation length and larger Debye length. Moreover, the valence $z_i$ is reflected on the electrostatic effect while the volume $v_i$ affects the steric effect as expected. In a word, larger $|z_i|, v_i$ lead to stronger effects.

The paper is organized as follows. In section \ref{sec: model}, we briefly introduce the PNPB theory, including the basic setup, the free energy functional, the 4PBik equation, the Fermi-like distribution at equilibrium and the PNPB equation. We also do a nondimensionalization step in section \ref{sec: model}, so that it suffices to analyze the dimensionless equations thereafter. In section \ref{sec: equilibrium}, we deal with the steady state problem and discuss the well-posedness. To be specific, we analyze the self-adjointness and the kernel of the fourth-order operator of the 4PBik equation. Also, we show the positivity of the void volume function and the well-posedness of the 4PBik equation in equilibrium. In section \ref{sec: numer}, we show the energy dissipation relation of both the PNPB equation and its semi-discrete form. Thus, we can deduce a numerical method based on this form. Various numerical tests are given, where we can see the parameter dependence of the model. Concluding remarks are dawn in section \ref{sec: conclusion}.

\section{Preliminaries}
\label{sec: model}

In this section, we give a brief introduction to the PNPB theory, and an in-depth review can be found in \cite{LiuEisenberg20}. We review the generalized free energy functional of the PNPB model and discuss the corresponding equilibrium and non-equilibrium problem. At the end of this section, we do a nondimensionalization step, so that it suffices to analyze the dimensionless equations thereafter.

\subsection{Settings}  

Consider an aqueous electrolytic system with $K$ ionic species in a bounded solvent domain $\Omega$. 
The ions of the $i^{\text{th}}$ species for $i=1,\cdots, K$ are treated as hard spheres with volume $v_i \in \mathbb{R}_{+}$ and valence $z_i \in \mathbb{Z}$. Besides, treat water as the $(K+1)^{\text{th}}$ species and regard it as hard spheres with volume $v_{K+1} \in \mathbb{R}_{+}$ as well, but it is electroneutrality with valence $z_{K+1}=0$. Due to the sizes of ions and water molecules, the void between all the hard spheres is indispensable. It is a whole treated as the $(K+2)^{\text{th}}$ species. 
Let $C_i (\br)$ be the concentration of the $i^{\text{th}}$ species for $i=1,\cdots, K+1$ and $\Gamma (\br)$ be the void volume function.
Then the void volume function $\Gamma (\br)$ satisfies \eqref{eq: gamma}, which is 
\begin{equation*}
\Gamma(\br) =1-\sum_{i=1}^{K+1} v_{i} C_{i}(\br).
\end{equation*}
Further assume the surface of the bounded domain $\Omega$ has a fixed shape and the ions and water molecules can not go through.
The total concentrations of all the $K+2$ species are conserved in $\Omega$.


As mentioned in the introduction, taking the correlation effects into consideration, the correlated electric potential $\phi(\br)$ satisfies the 4PBik equation \eqref{eq: 4pbik}: $L \phi(\br) = \sum_{i = 1}^{K} q_i C_i(\br)$.
The fourth-order operator $$L = \varepsilon_s\left(l_{\mathrm{c}}^{2} \Delta-I\right) \Delta$$ is used to approximate the permittivity of the bulk solvent and the linear response of correlated ions. The basic properties of $L$ including the self-adjointness and the invertibility are given in section \ref{subsec: operator L} under a dimensionless setting. $L$ is derived by a convolution of the classic Poisson dielectric operator with a Yukawa potential kernel \cite{XieLiuEisenberg16}. Therefore, it reduces to the classical Poisson equation if there is no correlation and polarization effect which is modelled by the introduced correlation length $l_c$. Namely, when $l_c = 0$, the 4PBik equation \eqref{eq: 4pbik} reduces to the Poisson equation \eqref{eq: poisson}.

\subsection{Generalized Gibbs free energy functional and the equilibrium}
\label{subsec: 2.2}
Conventionally, the Gibbs free energy functional for an aqueous electrolytic system with $K$ ionic species is given as 
$$
\frac{1}{2} \int_{\Omega} \rho(\br) \phi(\br) \mathrm{d} \br+
k_{B} T \int_{\Omega}\sum_{i=1}^{K} C_{i}(\br)
\left(\ln \frac{C_{i}(\br)}{C_{i}^{\mathrm{B}}}-1\right) \mathrm{d} \br.
$$
The first part is the electrostatic energy where the electric field $\phi$ is created by the charge on different ionic species and generally satisfies the Poisson equation \eqref{eq: poisson}. The second part is the entropy term, which describes the particle Brownian motion of the $K$ ion species.

For the PNPB model, there are two more species: water and void. So, the entropy part of water and void should be included except for ions. And the electric potential $\phi$ satisfies the 4PBik equation \eqref{eq: 4pbik} instead. 
Let $\mathbf{C} = (C_1, \cdots, C_{K+1})'$, the generalized Gibbs free energy \cite{LiuXieEisenberg17} is given in the following form 
\begin{equation}
\mathcal{F}(\mathbf{C}) =\mathcal{F}_{e l}(\mathbf{C})+\mathcal{F}_{e n}(\mathbf{C}),
\label{eq: energy}
\end{equation}
where the electric term 
\begin{equation}
\mathcal{F}_{e l}(\mathbf{C}) 
   =\frac{1}{2} \int_{\Omega} \rho(\br) \phi(\br) \mathrm{d} \br, \text{ s.t. } L \phi(\br) = 
   \rho(\br),
   \label{eq: electric energy}
  \end{equation}
and the entropy term
  \begin{equation}
\mathcal{F}_{e n}(\mathbf{C}) =k_{B} T \int_{\Omega}\left\{\sum_{i=1}^{K+1} C_{i}(\br)\left(\ln \frac{C_{i}(\br)}{C_{i}^{\mathrm{B}}}-1\right)+\frac{\Gamma(\br)}{v_{0}}\left(\ln \frac{\Gamma(\br)}{\Gamma^{\mathrm{B}}}-1\right)\right\} \mathrm{d} \br,
\label{eq: entropy energy}
\end{equation}
subject to the mass conservation constraints of ions and void: for the $i^{\text{th}}$ species, $i=1,\cdots,K+1,$
\begin{equation*}
    \int_{\Omega} C_i(\br)\ \mathrm{d} \br=: m_0^i,
\end{equation*}
and the void
\begin{equation*}
   \int_{\Omega} \Gamma(\br)\ \mathrm{d} \br=V-\sum_{i=1}^{K+1}v_i m_i^0.
\end{equation*}
Here $m_i^0>0$ is the total concentration of the $i^{\text{th}}$ species, $V$ is the volume of domain $\Omega$.


If the operator $L^{-1}$ is self-adjoint, the variational derivative of $\mathcal{F}(\mathbf{C})$ with respect to $C_i$ for $i = 1, \cdots, K+1$, is given by
\begin{equation*}
\frac{\delta \mathcal{F}(\mathbf{C})}{\delta C_{i}} = k_{B} T\left[\ln \frac{C_{i}(\br)}{C_{i}^{\mathrm{B}}}-\frac{v_{i}}{v_{0}} \ln \frac{\Gamma(\br)}{\Gamma^{\mathrm{B}}}\right]+q_{i} \phi(\br).
\end{equation*}
$\frac{\delta \mathcal{F}(\mathbf{C})}{\delta C_{i}}=0$ yields the Fermi-like distribution \eqref{eq: fermi}.
It is straightforward from \eqref{eq: fermi} that the concentration of each species depends on the steric potential which in turn depends on the concentration of all species.
We note that 
the steric potential makes sense only when 
\begin{equation*}
    \Gamma(\br)>0,
\end{equation*} 
which is an important issue we need to claim afterwards.

Thus combined the field equation of the electric potential \eqref{eq: 4pbik} with the Fermi-like distribution \eqref{eq: fermi}, the equilibrium of the PNPB system is summarized
as follows, for $\br \in \Omega$,
\begin{equation}
\label{eq: pnpb}
\begin{aligned}
C_{i}(\br) &= C_{i}^{\mathrm{B}} \exp \left(-\frac{q_i}{k_BT} \phi(\br)+\frac{v_{i}}{v_{0}}S(\br)\right), \quad i = 1, \cdots, K+1, \\
\int_{\Omega} C_i(\br)\ \mathrm{d} \br &= m_0^i,\quad i = 1, \cdots, K+1,\\
L \phi(\br) &= \varepsilon_{s}\left(l_{c}^{2} \Delta-1\right) \Delta \phi(\br)=\rho(\br), \\
S(\br) &= \ln \frac{\Gamma(\br)}{\Gamma^{\mathrm{B}}}, \\
\rho(\br) &= \sum_{i=1}^{K} q_{i} C_{i}(\br).
\end{aligned}
\end{equation}

If we ignore the finite size of ions, i.e., $v_i=0$, the steric effect vanishes. Hence, the Fermi-like distribution \eqref{eq: fermi} will reduce to the Boltzmann distribution \eqref{eq: Boltzmann}.
In the limiting case when we ignore the steric effect and correlation effect, the equilibrium system \eqref{eq: pnpb} will reduce to the classical PB equation with mass conservation. However, the lattice-based modified PB models used to account for steric effects are not consistent with the classical PB as the volume of ions tends to zero \cite{LiuEisenberg14}, this suggests the superiority of the PNPB model.

\subsection{Dynamics: the Poisson-Nernst-Planck-Bikerman equation}

As for the ion transport, it is natural to generalize the gradient flow structure of the classical PNP model to the PNPB model. Assume all the densities are not degenerate, that is,
\begin{equation*}
    C_i(\br,t)>0,~\text{for } i=1,\cdots, K+1, \quad \Gamma(\br,t)>0 ~ a.e. ~ \text{in} ~ \Omega.
\end{equation*}
The chemical potential $\mu_i$ of the $i^{\text{th}}$ ionic species is described by the variational derivative
\begin{equation*}
\mu_i = \dfrac{\delta \mathcal{F}(\mathbf{C})}{\delta C_i}
\end{equation*}
and is referred to in channel biology as the "driving force" for the current of the $i^{\text{th}}$ ionic species \cite{EnVarA}. In the PNPB theory, the chemical potential of species $i$ for $i = 1, \cdots, K+1$ becomes
\begin{equation}
\label{eq: potential}
    \mu_i = k_{B} T\left[\ln \frac{C_{i}(\br, t)}{C_{i}^{\mathrm{B}}}-\frac{v_{i}}{v_{0}} \ln \frac{\Gamma(\br, t)}{\Gamma^{\mathrm{B}}}\right]+q_{i} \phi(\br, t).
\end{equation}
Hence, the non-equilibrium system is
\begin{equation}
\label{eq: dynamicp}
\begin{aligned}
\frac{\partial C_{i}(\br, t)}{\partial t} &
= \nabla \cdot \left( \frac{D_i}{k_B T} C_i \nabla \mu_i \right), \quad \br \in \Omega,
\end{aligned}
\end{equation}
with $D_i$ being the diffusion coefficient.
\eqref{eq: dynamicp} is the Nernst-Planck-Bikerman (NPB) equation and can be rewritten as 
\begin{equation*}
\begin{aligned}
\frac{\partial C_{i}(\br, t)}{\partial t} &
=-\nabla \cdot {J}_{i}(\br, t), \quad \br \in \Omega,
\end{aligned}
\end{equation*}
where the flux density ${J}_{i}$ for each ionic species $i = 1, \cdots, K + 1$ is  
\begin{equation}
\label{eq: flux}
J_{i}(\br, t)=-\frac{D_i}{k_B T} C_i \nabla \mu_i.
\end{equation}
Define the dynamic steric potential
\begin{equation}
\label{eq: stericp}
S(\br, t) := \ln \frac{\Gamma(\br, t)}{\Gamma^{\mathrm{B}}}.
\end{equation}
Then, substituting \eqref{eq: stericp} into \eqref{eq: potential} and \eqref{eq: flux}, for $i = 1, \cdots, K+1$, 
the flux $J_{i}$ becomes
\begin{equation*}
\label{eq: flux2}
J_{i}(\br, t)=-D_{i}\left(\nabla C_{i}(\br, t)+\frac{q_i}{k_BT}C_i(\br, t) \nabla \phi(\br, t) - \frac{v_i}{v_0}C_{i}(\br, t) \nabla S(\br, t)\right).
\end{equation*}
The correlated electric potential $\phi(\br, t)$ satisfies the 4PBik equation
\begin{equation}
\label{eq: 4pbikp}
L \phi(\br, t) = \rho(\br, t) := \sum_{i = 1}^{K} q_i C_i(\br, t), \quad \br \in \Omega.
\end{equation}
Hence the NPB equation \eqref{eq: dynamicp} coupled with the 4PBik equation \eqref{eq: 4pbikp} is the so-called PNPB system.

In addition, to make the problem completed, we impose the initial condition 
\begin{equation}
 \label{eq: initial}
C_i(\br, 0) = C_i^0(\br),
\end{equation}
and the no-flux boundary condition
\begin{equation}
\label{eq: boundary}
J_{i} \cdot \mathbf{n} = 0,
\end{equation}
where $\mathbf{n}$ is the outward unit normal at the boundary $\partial \Omega$.

\begin{rem}[Mass conservation]
Let $C_i(\br, t), i = 1, \cdots, K+1$, be non-negative solutions to the PNPB model (\ref{eq: dynamicp}), (\ref{eq: 4pbikp}) equipped with the initial condition (\ref{eq: initial}) and the boundary condition (\ref{eq: boundary}). It's trivial that the system has the following mass conservation property, 
$$\int_{\Omega} C_i(\br, t)\ \mathrm{d} \br\equiv \int_{\Omega} C_i^0(\br)\ \mathrm{d} \br =: m_0^i$$ 
for each species $i = 1, \cdots, K + 1$. Consequently, 
$$
\int_{\Omega} \Gamma(\br, t)\ \mathrm{d} \br=V-\sum_{i=1}^{K+1}v_im_i^0.
$$
\end{rem}


Moreover, the flux $J_i = 0$ yields the Fermi-like distribution \eqref{eq: fermi}, which implies system \eqref{eq: pnpb} is the equilibrium of the PNPB system.

\subsection{Nondimensionalization}

To get a better understanding of the steric and correlation effects, we non-dimensionalizing the PNPB model and thus it can be dominated by certain dimensionless parameters. Denote $L=V^{1/d}$ as the diameter of the solvent region $\Omega$, $C_*$ and $D$ as the characteristic concentration and a typical diffusion coefficient. Rescale the variables 
$\tilde{\br}=\frac{\br}{L}, \ \tilde{t}=\frac{Dt}{L^2}.$ Let $\tilde{\Omega}=\{\tilde{\br}\in \mathbb{R}^d: \tilde{\br}L\in \Omega\}$ and introduce the dimensionless quantities 
$$
\tilde{\phi}=\frac{e\phi}{k_BT},\quad \tilde{C}_i=\frac{C_i}{C_*},\quad
\tilde{C}_i^{\mathrm{B}}=\frac{C_i^{\mathrm{B}}}{C_*},\quad \tilde{C}_i^0=\frac{C_i^0}{C_*},\quad 
\tilde{D}_i=\frac{D_i}{D},\quad 
\tilde{v}_i=\frac{v_i}{V},\quad 
\tilde{v}_0=\frac{v_0}{V}.
$$
For notation convenience, we drop the tildes from now on.

Define the dimensionless parameter $\eta=V C_*$.
Non-dimensionalizing \eqref{eq: dynamicp}, we obtain
\begin{equation}
\label{eq: no-dim dynamic}
\begin{aligned}
\frac{\partial C_{i}(\br, t)}{\partial t} &
= \nabla \cdot \left(D_i C_i \nabla \mu_i \right), \quad \br \in \Omega,
\end{aligned}
\end{equation}
where the chemical potential  
\begin{equation*}
\mu_i=\ln \frac{C_{i}(\br, t)}{C_{i}^{\mathrm{B}}}+z_i \phi(\br,t)-\frac{v_i}{v_0} S(\br,t).
\end{equation*}
One has the steric potential
$$
S(\br, t) = \ln \frac{\Gamma(\br, t)}{\Gamma^{\mathrm{B}}}
$$
with the void volume function
\begin{equation}
    \Gamma(\br, t)=1-\eta\sum\limits_{i=1}^{K+1}v_i C_i,\quad \Gamma^{\mathrm{B}}=1-\eta\sum\limits_{i=1}^{K+1}v_i C_i^{\mathrm{B}},
    \label{eq: no-dim gamma}
\end{equation} 
and the flux 
\begin{equation*}
    J_i = -D_i\left(\nabla C_i(\br,t)+z_i C_i(\br,t)\nabla\phi(\br,t)-\frac{v_i}{v_0}C_i(\br,t)\nabla S(\br,t)\right), ~~ i = 1, \cdots, K+1.
\end{equation*}
\eqref{eq: no-dim dynamic} is regarded as the dimensionless NPB equation. For all species $i = 1, \cdots, K+1$, the no-flux boundary condition and initial condition become
\begin{equation}
    J_{i} \cdot \mathbf{n} = 0,
    \label{eq: no-dim bc}
\end{equation}
and 
\begin{equation}
C_i(\br, 0) = C_i^0(\br).
\label{eq: no-dim ic}
\end{equation}

Further, introduce the Debye length
$l_D=\sqrt{\dfrac{\varepsilon_s k_B T}{e^2C_*}}$.
Define the dimensionless parameters
$$\nu=\frac{l_D}{L},\quad  \lambda=\frac{l_c}{L},$$
then the dimensionless 4PBik equation reads
\begin{equation}
L\phi(\br, t) = \rho(\br, t),
\label{eq: no-dim 4pbik}
\end{equation}
where the fourth-order operator $L$ becomes
$$
L = \nu^2\left(\lambda^2\Delta - I \right)\Delta,
$$
and the total charged density
$$
\rho(\br, t) = \sum_{i=1}^K z_i C_i(\br, t).
$$
Note that the fourth-order equation \eqref{eq: no-dim 4pbik} reduces to the dimensionless Poisson equation when $\lambda=0$.
Similarly, the dimensionless PNPB system \eqref{eq: no-dim dynamic}, \eqref{eq: no-dim 4pbik} is valid under the assumption that $C_i, \Gamma$ are not degenerate.

\begin{rem}[Mass conservation']
Let $C_i(\br, t), i = 1, \cdots, K+1$, be non-negative solutions to the PNPB model (\ref{eq: no-dim dynamic}), (\ref{eq: no-dim 4pbik}) equipped with the initial condition (\ref{eq: no-dim ic}) and the boundary condition (\ref{eq: no-dim bc}). Let $m_i^0$ be $m_i^0/\eta$ for simplicity, then the dimensionless PNPB system satisfies the mass conservation property, 
$$
\int_{\Omega} C_i(\br, t)\ \mathrm{d} \br\equiv \int_{\Omega} C_i^0(\br)\ \mathrm{d} \br = m_0^i, \quad
i = 1, \cdots, K + 1,$$ 
and consequently,
$$
\int_{\Omega} \Gamma(\br, t)\ \mathrm{d} \br=1-\eta\sum_{i=1}^{K+1}v_i m_i^0.
$$
\end{rem}

Let $J_i = 0$ for all $i = 1, \cdots, K+1$, the equilibrium of the above dimensionless PNPB system becomes 
\begin{equation}
\label{eq: no-dim PNPB}
\begin{aligned}
C_{i}(\br) &= C_{i}^{\mathrm{B}} \exp \left(-z_i \phi(\br)+\frac{v_{i}}{v_{0}}S(\br)\right), \quad i = 1, \cdots, K+1, \\
\int_{\Omega} C_i(\br)\ \mathrm{d} \br &= m_0^i, \quad
i = 1, \cdots, K + 1,\\
L \phi(\br) &=\nu^2 \left(\lambda^2\Delta -1\right)\Delta \phi=\rho(\br), \\
S(\br) &= \ln \frac{\Gamma(\br)}{\Gamma^{\mathrm{B}}}, \\
\rho(\br) &=  \sum_{i=1}^{K} z_{i} C_{i}(\br).
\end{aligned}
\end{equation}

In the rest of the paper, we only take into consideration the dimensionless equations.

\section{Steady state and its well-posedness}
\label{sec: equilibrium}

In this section, we focus on the theoretical side of the equilibrium problem \eqref{eq: no-dim PNPB}. First of all, we show the self-adjointness and the kernel of the fourth-order operator $L$ of the 4PBik equation \eqref{eq: no-dim 4pbik}. The self-adjoint property is significant in the derivation of the first-order variation of the free energy functional. As mentioned in subsection \ref{subsec: 2.2}, the positivity of the void volume function $\Gamma$ is of vital importance. In order to show this property, we give four equivalent characterizations of the steady state of the PNPB system, from which the positivity of $\Gamma$ naturally follows. The existence of the weak solution to the equilibrium problem \eqref{eq: no-dim PNPB} can be proved by the calculus of variations. In addition, we show that the free energy functional is strictly convex. This ensures the uniqueness of the solution.

\subsection{Properties of the fourth-order operator $L$}
\label{subsec: operator L}

\subsubsection{Self-adjointness of the operator $L$}

Define the operator 
\begin{equation*}
\begin{aligned}
L : \mathcal{D}(L) &\to \mathcal{H} \\ 
\phi &\mapsto \rho = L \phi = \nu^2 \left(\lambda^2 \Delta-I \right) \Delta \phi.
\end{aligned}
\end{equation*}
Here we consider $\mathcal{D}(L) = \{ u \in  H^4(\Omega), u|_{\partial \Omega} = \Delta u|_{\partial \Omega} = 0\}$ and $\mathcal{H} = L^2(\Omega)$, then $\overline{\mathcal{D}(L)} = \mathcal{H}$. \\ 

\begin{prop}
	The linear operator $L$ is symmetric and positive definite.
\end{prop}
\begin{proof}
Define an inner product for the Hilbert space $\mathcal{H}$ as $(\mathbf{f}, \mathbf{g})_{\Omega}:= \int_{\Omega} \mathbf{f}(\br) \cdot \mathbf{g}(\br) \,\mathrm{d} \br$.
For all $\phi, \psi \in \mathcal{D}(L)$, via integration by parts and the definition of $\mathcal{D}(L)$, one has
$$
\left( L \phi, \psi \right)_{\Omega}  = \left(\phi, L \psi \right)_{\Omega}.
$$
Hence we can conclude that the operator $L$ is symmetric. 

Besides, for all $\phi \in \mathcal{D}(L)$,
\begin{equation}
\label{pd0}
\begin{aligned}
\left( L \phi, \phi \right)_{\Omega} 
&= \left(\nu^2 \left(\lambda^2 \Delta- I \right) \Delta \phi, \phi \right)_{\Omega} 
=  \nu^2 \lambda^{2} \left( \Delta^2 \phi, \phi \right)_{\Omega} - \nu^2 \left( \Delta \phi, \phi \right)_{\Omega}\\
& = \nu^2 \lambda^{2}  \left( \Delta \phi, \Delta \phi \right)_{\Omega} +\nu^2  \left( \nabla \phi, \nabla \phi \right)_{\Omega}.
\end{aligned}
\end{equation}
Due to Poincar\'{e} inequality, we have
\begin{align}
\int_{\Omega} |\phi(\br)|^2 \,\mathrm{d}\br \leq C(\Omega)\int_{\Omega} |\nabla \phi(\br)|^2 \,\mathrm{d}\br. \label{fri0}
\end{align} 
By (\ref{fri0}), (\ref{pd0}) becomes
\begin{equation}
\begin{aligned}
\left( L \phi, \phi \right)_{\Omega} 
& = \nu^2 \lambda^{2}  \left( \Delta \phi, \Delta \phi \right)_{\Omega} +\nu^2  \left( \nabla \phi, \nabla \phi \right)_{\Omega} \geq C (\phi, \phi)_{\Omega} = C \|\phi\|_{L^2(\Omega)}^2
\end{aligned}
\end{equation}
for some $C> 0$, which means that $L$ is positive definite.
\end{proof}

\begin{prop}
	The linear operator $L$ is self-adjoint.
\end{prop}
\begin{proof}
Consider the whole space $\Omega = \mathbb{R}^d$. Combined the symmetry of the linear operator $L$ with $Ran(L) = \mathcal{H}$, the proposition holds.

Let $\Omega \subset \mathbb{R}^d$ be a bounded open region with smooth boundary. Let the operator $A = \Delta^2$, then $A$ is a closed linear symmetric operator. Denote $\hat{A}$ as the Friedrichs extension of $A$. Let
$$
a(u, v) = \int_{\Omega} \nabla u \overline{\nabla v} \mathrm{d} \br, ~~u, v \in \mathcal{D}(A) = \mathcal{D}(L).
$$
Due to Poincar\'{e} inequality, there exists $\beta > 0$, such that
$$
(Au, u) = \int_{\Omega} |\Delta u|^2 \mathrm{d} \br \geq \beta \int_{\Omega} |u(\br)|^2 \mathrm{d} \br,
$$
hence $A \geq \beta$ and 
$$
\!|\!|\!|u\!|\!|\!| = \|u\|_G, ~~ G = \{ u \in  H^2(\Omega), u|_{\partial \Omega} = \Delta u|_{\partial \Omega} = 0\}.
$$
$G$ is the closure of $\mathcal{D}(L)$ under the norm $\!|\!|\!| \cdot \!|\!|\!| $ in $\mathcal{D}(A)$. Hence $a$ can be expanded to a closed positive definite conjugate bilinear form on $G \times G$. 
Next consider the self-adjoint extension $\hat{A}$. 
$$
\mathcal{D}(\hat{A}) = \left\{ u \in G ~\big|~ \exists ~C_u > 0, \text{such that } \forall v \in G, \left|\int_{\Omega} \nabla u \overline{\nabla v} \mathrm{d} \br\right| \leq C_u \|v\|_{L^2(\Omega)} \right\},
$$
From Riesz theorem, there exists $f \in L^2(\Omega)$, such that
$$
\int_{\Omega} \nabla u \overline{\nabla v} \mathrm{d} \br = \int_{\Omega} f \overline{v} \mathrm{d} \br,
$$
then denote $f = \widetilde{\Delta^2} u$, thus $\hat{A} u = \widetilde{\Delta^2} u$.
Following from the estimates that $\|u\|_{H^4(\Omega)}$ can be dominated by $\|u\|_{L^2(\Omega)}, \|\Delta u\|_{L^2(\Omega)}, \|\widetilde{\Delta^2} u\|_{L^2(\Omega)}$ \cite{evans}, one has 
$$
\mathcal{D}(\hat{A}) = \mathcal{D}(A),
$$
which means $A$ is self-adjoint.
Furthermore, $-\Delta$ is a relatively tight operator of $\Delta^2$ and hence is bounded.
Following from the Kato-Rellich theorem, we obtain that $L$ is self-adjoint.
\end{proof}

Furthermore, the problem
\begin{equation*}
\begin{aligned}
L e(\br) = \nu^2\left(\lambda^2 \Delta-I \right) \Delta e(\br) &= 0, \br \text{ in } \Omega, \\
e &= 0, \br \text{ on } \partial \Omega, \\
\Delta e &= 0, \br \text{ on } \partial \Omega,
\end{aligned}
\end{equation*}
leads to $e(\br) \equiv 0$ in $\Omega$, hence $L$ is an injection. Consider a bounded open region $\Omega \subset \mathbb{R}^d$ with smooth boundary. Then $L$ is invertible, the inverse operator $L^{-1}: \mathcal{H} \to \mathcal{D}(L)$ is also self-adjoint.

\subsubsection{Kernel of the operator $L$}

It is well-known that the fundamental solution to the Laplace equation 
\begin{equation}
\label{eq: Psi}
-\Delta \Psi (\br,\br^{\prime})=\delta (\br-\br^{\prime})
\end{equation}
is
\begin{equation*}
\Psi (\br,\br^{\prime})=\left\{
    \begin{array}{lll}
    -\frac{1}{2}|\br-\br^{\prime}|,\quad & d=1,\\[2mm]
    -\frac{1}{2\pi}\ln|\br-\br^{\prime}|,\quad & d=2,\\[2mm]
    \frac{1}{4\pi|\br-\br^{\prime}|},\quad & d=3.\\
    \end{array}\right.
\end{equation*}
And the fundamental solution to the screened Poisson equation 
\begin{equation}
\label{eq: W}
(I-\lambda^2\Delta)\mathcal{W} (\br,\br^{\prime})=\delta (\br-\br^{\prime})
\end{equation}
is
\begin{equation*}
    \mathcal{W}(\br,\br^{\prime})=\left\{
    \begin{array}{lll}
    \frac{\exp(-|\br-\br^{\prime}|/\lambda)}{2\lambda},\quad & d=1,\\[2mm]
    \frac{1}{2\pi\lambda^2}K_0(|\br-\br^{\prime}|/\lambda),\quad & d=2,\\[2mm]
    \frac{\exp(-|\br-\br^{\prime}|/\lambda)}{4\pi\lambda^2|\br-\br^{\prime}|},\quad & d=3,\\
    \end{array}\right.
\end{equation*}
where $K_0(r)$ is the modified Bessel function of the second kind.

Next, we have the following proposition.
\begin{prop}
Consider the equation
\begin{equation}
  L \mathcal{K}(\br,\br^{\prime})=\delta (\br-\br^{\prime}),
\label{eq: L}
\end{equation}
then the fundamental solution to \eqref{eq: L} is 
\begin{equation}
    \mathcal{K}(\br,\br^{\prime})=\frac{1}{\nu^2}(\Psi-\lambda^2 \mathcal{W})(\br,\br^{\prime})
    =\left\{
    \begin{array}{lll}
    -\frac{1}{2\nu^2}\left(|\br-\br^{\prime}|+\lambda\exp(-|\br-\br^{\prime}|/\lambda)\right),\quad & d=1,\\[2mm]
    -\frac{1}{2\pi\nu^2}\left(\ln|\br-\br^{\prime}| +K_0(|\br-\br^{\prime}|/\lambda)\right),\quad & d=2,\\[2mm]
    \frac{1-\exp(-|\br-\br^{\prime}|/\lambda)}{4\pi\nu^2|\br-\br^{\prime}|},\quad & d=3.\\
    \end{array}\right.
\label{eq: kernel}
\end{equation}
\end{prop}
\begin{proof}
By substituting $\mathcal{K}(\br,\br^{\prime})=\frac{1}{\nu^2}(\Psi-\lambda^2 \mathcal{W})(\br,\br^{\prime})$ into the equation \eqref{eq: L}, one has
$$
\begin{aligned}
L \mathcal{K}(\br,\br^{\prime})=&\nu^2\left(\lambda^2 \Delta-I \right) \Delta \mathcal{K}(\br,\br^{\prime})\\
=&(I-\lambda^2\Delta)(-\Delta)(\Psi-\lambda^2 \mathcal{W})(\br,\br^{\prime})\\
 =&(I-\lambda^2\Delta)(-\Delta)\Psi(\br,\br^{\prime})-\lambda^2(-\Delta)(I-\lambda^2\Delta)\mathcal{W}(\br,\br^{\prime})\\
 =& (I-\lambda^2\Delta)\delta(\br-\br^{\prime})-\lambda^2(-\Delta)\delta(\br-\br^{\prime})\\
 =&\delta(\br-\br^{\prime}).
\end{aligned}
$$
The penultimate equation holds because of the equation \eqref{eq: Psi} and \eqref{eq: W}.
\end{proof}
With the kernel in hand, the solution to the 4PBik equation \eqref{eq: no-dim 4pbik} can be given by $\phi=\mathcal{K}*\rho$.

\subsection{The free energy functional $\mathcal{F}$ }

The dimensionless version of the Gibbs free energy functional $\mathcal{F}$ given in \eqref{eq: energy}-\eqref{eq: entropy energy} is 
\begin{equation}
\label{eq: no-dim energy}
\mathcal{F}(\mathbf{C}) =\mathcal{F}_{e l}(\mathbf{C})+\mathcal{F}_{e n}(\mathbf{C}),
\end{equation}
with  the dimensionless electric term 
\begin{equation}
\mathcal{F}_{e l}(\mathbf{C}) =\frac{1}{2} \int_{\Omega} \rho \phi \mathrm{d} \br=\frac{1}{2} \int_{\Omega} \rho L^{-1} \rho \mathrm{d} \br, \\ 
\end{equation}
and the dimensionless entropy term 
\begin{equation}
\mathcal{F}_{e n}(\mathbf{C}) = \int_{\Omega}\left\{\sum_{i=1}^{K+1} C_{i}\left(\ln \frac{C_{i}}{C_{i}^{\mathrm{B}}}-1\right)+\frac{\Gamma(\br)}{\eta v_{0}}\left(\ln \frac{\Gamma}{\Gamma^{\mathrm{B}}}-1\right)\right\} \mathrm{d} \br.
\end{equation}
Correspondingly, the chemical potential $\mu_i$ of the $i^{\text{th}}$ species of \eqref{eq: no-dim energy} becomes
\begin{equation}
    \mu_i = \frac{\delta \mathcal{F}}{\delta C_i}= \ln \frac{ C_{i}}{C_{i}^{\mathrm{B}}}-\frac{v_{i}}{v_{0}} \ln \frac{\Gamma}{\Gamma^{\mathrm{B}}}+z_{i} \phi.
    \label{eq: chemical}
\end{equation}

In this subsection, we aim at studying the convexity of the free energy $\mathcal{F}$ in order to be prepared to the well-posedness of the equillibrium system \eqref{eq: no-dim PNPB}. To start with, the positivity of $\Gamma$ is needed to make sense of the physical problem. So, we characterize the steady state of PNPB model by four equivalent statements firstly. This allows for a rigorous proof of $\Gamma>0$. Then the convexity of $\mathcal{F}$ can be concluded by simple computation.

\subsubsection{Four equivalent statements for the steady state}

In this part, we give four equivalent statements for the steady state of the PNPB system.
\begin{prop}\label{prop3}
	Assuming that $\bar C_i \in L^1\cap L\log L$ is bounded with $\int_{\Omega}\bar{C}_i\mathrm{d}\br=m_i^0$, $\bar C_i \in C(\Omega)$, $\bar C_i > 0$ in $\Omega$. 
	Let $\bar\mu_i = \ln \frac{\bar C_{i}(\br)}{C_{i}^{B}}-\frac{v_{i}}{v_{0}} \ln \frac{\bar\Gamma(\br)}{\Gamma^{B}}+z_{i} \bar\phi(\br)$, where $ \bar \Gamma(\br) = 1-\eta\sum_{i = 1}^{K+1} v_i \bar C_i(\br) $, $~ \bar \phi(\br)$ satisfies $L \bar \phi(\br) = \bar \rho(\br), ~\bar \rho(\br) = \sum_{i = 1}^{K} z_i \bar  C_i(\br)$. Further assume 
	$n \cdot D_i\bar{C}_{i}\nabla \bar{\mu}_i=0$ on $\partial \Omega$.
	Then the following four statements are equivalent:
	\begin{enumerate}[(i)]
	   
		\item Equilibrium: $\bar\mu_i \in \dot{H}^1(\Omega)$ and $\nabla \cdot (D_i\bar C_i \nabla \bar \mu_i)$ $= 0$ in $H^{-1}(\Omega)$, $\forall ~i = 1, \cdots, K+1$. 
		\item No dissipation: $\sum_{i = 1}^{K+1}\int_{\Omega}D_i \bar C_i|\nabla \bar\mu_i|^2 \,\mathrm{d} \br = 0$.
		\item $\left( \bar C_1, \cdots, \bar C_{K+1} \right)$ is a critical point of $\mathcal{F}$.
		\item $\bar\mu_i$ is a constant, $\forall~ i = 1, \cdots, K+1$.
	\end{enumerate}
	
\end{prop}

\begin{proof}
	At first, we prove (i)$\Rightarrow$(ii). Since $\bar{\mu}_{i} \in H^{1}(\Omega), \nabla \cdot\left(D_i \bar C_{i} \nabla \bar{\mu}_{i} \right)=0$ in $H^{-1}\left(\Omega\right)$, $ C_{0}^{\infty}(\Omega)$ is dense in $\dot{H}^1(\Omega)$ and $\bar C_i$ is bounded, one has
	\begin{equation}
	0 = \int_{\Omega} \bar{\mu}_{i} \nabla \cdot\left( D_i \bar{C}_{i} \nabla \bar{\mu}_{i} \right) \mathrm{d} \br = -\int_{\Omega} D_i\bar{C}_{i}\left| \nabla \bar{\mu}_{i}\right|^{2} \mathrm{d} \br + \int_{\partial \Omega}\bar{\mu}_i n \cdot D_i\bar{C}_{i}\nabla \bar{\mu}_i\mathrm{d}\sigma, \quad \forall ~i = 1, \cdots, K+1.
	\end{equation}
	Hence (ii) holds.
	
	Next we prove (iii)$\Leftrightarrow$(iv). Notice that $\left( \bar C_1, \cdots, \bar C_{K+1}\right) $ is a critical point of $\mathcal{F}$ if and only if 
	\begin{equation}
	\frac{\mathrm{d}}{\mathrm{d} \epsilon}\bigg|_{\epsilon=0} \mathcal{F}\left(\left(\bar C_i, \cdots, \bar C_i + \epsilon V_i, \cdots, \bar C_{K+1} \right)'\right)=0, \quad \forall ~ V_i \in C_{0}^{\infty}(\Omega)
 	\text { with } \int_{\Omega} V_i(\br) \,\mathrm{d} \br = 0.
	\end{equation}
	Therefore,
	\begin{equation*}
	 \int_{\Omega}  z_i V_i(\br)\bar \phi(\br) \mathrm{d} \br + \int_{\Omega}\left\{ V_i(\br) \ln \frac{\bar C_{i}(\br)}{C_{i}^{\mathrm{B}}}-\frac{v_i}{v_{0}} V_i(\br) \ln \frac{\bar \Gamma(\br)}{\Gamma^{\mathrm{B}}}\right\} \mathrm{d} \br=0.
	\end{equation*}
	Equivalently,
	\begin{equation}
	\int_{\Omega} \bar{\mu}_{i} V_i \,\mathrm{d} \br = 0, \quad \forall~ V_i \in C_{0}^{\infty}(\Omega),
	\end{equation}
	which implies $\bar\mu_i $ is a constant, $\forall~ i = 1, \cdots, K+1$.
	
	Then we prove (ii)$\Rightarrow$(iv). Suppose $\sum_{i = 1}^{K+1} \int_{\Omega}D_i\bar C_i|\nabla \bar\mu_i|^2 \,\mathrm{d} \br = 0$. It follows from $\bar C_i > 0$ at any point
	$\br_{0} \in \Omega$ that $\nabla \bar{\mu}_{i}=0$ in $\Omega$ and thus $\bar {\mu}_i$ is a constant for all $i = 1, \cdots, K+1$.
	
	Hence we complete the proof for (ii)$\Rightarrow$(iii) and (iii)$\Rightarrow$(iv).
	
	Finally (i) is a direct consequence of (iv), thus (iv) $\Rightarrow$ (i). 
\end{proof}

\subsubsection{The positivity of the void volume function $\Gamma(\br)$}

With the equivalent characterizations of the steady state in hand, we are ready to show the positivity of $\Gamma$ at equilibrium. Note that this has already been stated formally in \cite{LiuEisenberg20}. 
\begin{prop}\label{prop4}
 If $1-\eta\sum\limits_{i=1}^{K+1}v_i m_i^0>0$, then $\Gamma(\br) =1-\eta\sum\limits_{i=1}^{K+1} v_{i} C_{i}(\br)>0$ at the steady state.
\end{prop}
\begin{proof}
    Consider the set $D=\{\br:\Gamma^*(\br) =1-\eta\sum\limits_{i=1}^{K+1} v_{i} C^*_{i}(\br) = 0\}$.
	Assume that at the steady state, measure $m(D) > 0$, where $C^*_i$ and $\Gamma^*$ are the concentration and void volume function of the steady state respectively.	
	
    First of all, consider the case where $C_i(\br)\equiv m_i^0$, $\Gamma(\br)\equiv 1-\eta\sum\limits_{i=1}^{K+1}v_i m_i^0:=\Gamma_c$. The free energy
	$$\mathcal{F}_c = \frac{1}{2} \int_{\Omega} \rho_c \phi_c \mathrm{d} \br + \int_{\Omega}\left\{\sum_{i=1}^{K+1} m_i^0\left(\ln \frac{m_i^0}{C_i^{\mathrm{B}}}-1\right)+\frac{\Gamma_c}{\eta v_{0}}\left(\ln \frac{\Gamma_c}{\Gamma^{\mathrm{B}}}-1\right)\right\} \mathrm{d} \br,$$
	where $L\phi_c=\rho_c$ with $\rho_c=\sum\limits_{i=1}^{K}z_i m_i^0$. This implies
	$\mathcal{F}^*\leq \mathcal{F}_c < +\infty$.

	According to Proposition \ref{prop3}, $\left(  C^*_1, \cdots, C^*_{K+1} \right)$ is a critical point of $\mathcal{F}$, which means that $\forall~ V_i \in C_{0}^{\infty}(\Omega)$ with $\int_{\Omega} V_i(\br) \,\mathrm{d} \br = 0$,
	\begin{equation}
	\label{eq: test}
	\int_{\Omega} z_i \phi^*(\br)V_i(\br) \mathrm{d} \br + \int_{\Omega}\left\{ V_i(\br) \ln \frac{ C^*_{i}(\br)}{C_{i}^{\mathrm{B}}}-\frac{v_i}{v_{0}} V_i(\br) \ln \frac{\Gamma^*(\br)}
	{\Gamma^{\mathrm{B}}}\right\} \mathrm{d} \br=0, ~~i = 1, \cdots, K+1,
	\end{equation}
	where $L \phi^*(\br) = \rho^*(\br)$ and $\rho^*(\br) = \sum\limits_{i=1}^{K} z_{i} C_{i}^*(\br)$.
	 
	 It follows from the assumption that 
	 $\int_{\Omega}\frac{v_i}{v_{0}} V_i(\br) \ln \frac{ \Gamma^*(\br)}{\Gamma^{\mathrm{B}}} \mathrm{d} \br$ tend to infinity. Since $\phi^*$ is the solution to $L \phi^*(\br) = \rho^*(\br)$, $\phi^*$ is bounded. Hence, equation (\ref{eq: test}) holds only if $C_i^*(\br)=0$ in $D$ for all $i=1,\cdots, K+1$. As a consequence, $\Gamma^*(\br)=1-\eta\sum\limits_{i=1}^{K+1} v_{i} C^*_{i}(\br)=1$ in $D$, which leads to a contradictory. Thus $m(D) = 0$. 
\end{proof}

\begin{cor}
At the equilibrium, $0<C_i<\frac{1}{\eta v_i}$, for each $i=1,\cdots, K+1$.
\end{cor}
\begin{proof}
It follows from the positivity of $\Gamma$ that
$\eta\sum\limits_{i=1}^{K+1}v_i C_i<1$.
Hence, $C_i<\frac{1}{\eta v_i}, ~ i=1,\cdots, K+1$ as $v_i, C_i>0$.
\end{proof} 
This tells that $C_i$ can not reach or exceed the maximum concentration $\frac{1}{\eta v_i}$. This is exactly the saturation phenomenon.
We emphasize that the classical PB theory with point charge assumption fails to describe the concentration saturation as the Boltzmann distribution may produce an infinite concentration when the electric potential tends to infinity. This is a deficiency of PB theory for modeling a system with strong local electric fields or interactions.

\subsubsection{The convexity of free energy functional $\mathcal{F}$}
To show the well-posedness of the equilibrium of the PNPB system, we give the convexity property of the free energy $\F$.

For all $\boldsymbol{V} = (V_1, \cdots, V_{K+1})$, denote $\rho_{\epsilon}(\br) = \sum\limits_{i=1}^{K} z_{i} (C_{i} + \epsilon V_i)(\br)$, $\Gamma_{\epsilon}(\br) = 1-\eta\sum\limits_{i=1}^{K+1} v_{i} (C_{i} + \epsilon V_i)(\br)$, then
\begin{equation*}
\begin{aligned}
&\frac{\mathrm{d}^2}{\mathrm{d} \epsilon^2}  {\mathcal{F}(\boldsymbol{C} + \epsilon \boldsymbol{V})} =\frac{\mathrm{d}^2}{\mathrm{d} \epsilon^2} \left[\mathcal{F}_{e l}(\boldsymbol{C} + \epsilon \boldsymbol{V})+\mathcal{F}_{e n}(\boldsymbol{C} + \epsilon \boldsymbol{V})\right]\\
=~& \frac{\mathrm{d}^2}{\mathrm{d} \epsilon^2} \left[ \frac{1}{2} \int_{\Omega} \rho_{\epsilon} L^{-1} \rho_{\epsilon} \mathrm{d} \br 
+  \int_{\Omega}\left\{\sum_{i=1}^{K+1} (C_{i} + \epsilon V_i)\left(\ln \frac{(C_{i} + \epsilon V_i)}{C_{i}^{\mathrm{B}}}-1\right)+\frac{\Gamma_{\epsilon}}{\eta v_{0}}\left(\ln \frac{\Gamma_{\epsilon}}{\Gamma^{\mathrm{B}}}-1\right)\right\} \mathrm{d} \br \right] \\
=~& \frac{\mathrm{d}}{\mathrm{d} \epsilon} \left[ \int_{\Omega} \left(\sum_{i = 1}^{K} z_i V_i\right) L^{-1} \left(\rho_{\epsilon}\right) \mathrm{d} \br + \int_{\Omega}\left\{\sum_{i=1}^{K+1} V_i \ln \frac{(C_{i} + \epsilon V_i)}{C_{i}^{\mathrm{B}}}-\frac{1}{v_{0}} \left[ \sum_{i=1}^{K+1} v_i V_i\right] \ln \frac{\Gamma_{\epsilon}}{\Gamma^{\mathrm{B}}}\right\} \mathrm{d} \br \right] \\
=~& \int_{\Omega} \left( \sum_{i = 1}^{K} z_i V_i \right) L^{-1} \left(\sum_{i = 1}^{K} z_i V_i\right) \mathrm{d} \br
+ \int_{\Omega}\left\{\sum_{i=1}^{K+1} \frac{V_i^2}{(C_i + \epsilon V_i)} + \left[ \sum_{i=1}^{K+1} v_i V_i \right]^2 \frac{\eta}{v_0 \Gamma_{\epsilon}}\right\} \mathrm{d} \br. 
\end{aligned}
\end{equation*}
Particular attention is paid to the limit case of zero $\epsilon$, then one has
\begin{equation}
\begin{aligned}
\frac{\mathrm{d}^2}{\mathrm{d} \epsilon^2}  \bigg|_{\epsilon = 0} {\mathcal{F}(\boldsymbol{C} + \epsilon \boldsymbol{V})} 
= \int_{\Omega} \left( \sum_{i = 1}^{K} z_i V_i \right) L^{-1} \left(\sum_{i = 1}^{K} z_i V_i\right) +\left\{\sum_{i=1}^{K+1} \frac{V_i^2}{C_i} + \left[ \sum_{i=1}^{K+1} v_i V_i\right]^2 \frac{\eta}{v_0 \Gamma}\right\} \mathrm{d} \br. 
\end{aligned}
\label{eq: 2nd var}
\end{equation}

From \eqref{eq: 2nd var}, we can conclude that if $\Gamma(\br)$ and $C_i(\br), i=1,\cdots,K+1$, are well-defined in $\Omega$, then 
$\frac{\mathrm{d}^2}{\mathrm{d} \epsilon^2} \big|_{\epsilon = 0} {\mathcal{F}(\boldsymbol{C} + \epsilon \boldsymbol{V})}>0$. 
Hence, it follows from Proposition \ref{prop4} that the free energy functional $\mathcal{F}$ is strictly convex at the equilibrium. 

\subsection{Well-posedness}

Motivated by \cite{Li09}, where the well-posedness of the solution to the PB equation is proved.
We show the well-posedness of problem \eqref{eq: no-dim PNPB} by discovering the unique minimizer of the free energy functional $\mathcal{F}$ using the calculus of variations \cite{Haim_Brezis}.

\begin{thm}
If $1-\eta\sum\limits_{i=1}^{K+1}v_i m_i^0>0$, there exists a unique weak solution to the equilibrium of the PNPB system \eqref{eq: no-dim PNPB}.
\end{thm}

\begin{proof}

Let $S_{t}(x)=x(\log x-t)$, then $S_t(x) \geq -e^{t-1}$, $S_{t}(x)$ is convex. 
Then
\begin{equation}
\begin{aligned}
\mathcal{F}(\boldsymbol{C})&= \frac{1}{2}\int_{\Omega}\rho L^{-1}\rho\mathrm{d}\br+ \sum_{i=1}^{K + 1} \int_{\Omega} S_{1+\log C_i^{\mathrm{B}}}(C_i) \mathrm{d} \br +  \int_{\Omega} \frac{1}{\eta v_0} S_{1+\log \Gamma^{\mathrm{B}}}(\Gamma) \mathrm{d} \br \\
&\geq -\left(\sum_{i=1}^{K+1}C_i^{\mathrm{B}}+\frac{\Gamma^{\mathrm{B}}}{\eta v_0}\right).
\end{aligned}
\label{eq: bound}
\end{equation}
Hence, $\mathcal{F}$ is bounded below.

Let $m=\inf_{\boldsymbol{C}} \mathcal{F}(\boldsymbol{C})$, then $m\leq\mathcal{F}^c<\infty$.
Select a minimizing sequence $\{ \boldsymbol{C}^{(k)}\}_{k=1}^{\infty}$, 
that is
$$
\lim_{k\rightarrow \infty} \mathcal{F}(\boldsymbol{C}^{(k)})=m.
$$
We use the direct method for calculus of variation to show there exists $\boldsymbol{C}^*$ such that $\mathcal{F}(\boldsymbol{C}^*)=m$. 
It follows from \eqref{eq: bound} that 
$\int_{\Omega} S_{1+\log C_i^{\mathrm{B}}}(C_i^{(k)}) \mathrm{d} \br$ is bounded. Next, we show that $\{ C_i^{(k)}\}$ is uniformly integrable.
This is because
$$
\begin{aligned}
&\sup_{k}\int_{\{|C_i^{(k)}>M|\}}|C_i^{(k)}| \mathrm{d} \br
=\sup_{k}\int_{\Omega}1_{\{C_i^{(k)}>M\}}C_i^{(k)}\mathrm{d} \br\\
\leq& \frac{1}{\log M - 1-\log C_i^{\mathrm{B}}}\sup_{k}\int_{\Omega}1_{\{C_i^{(k)}>M\}}
C_i^{(k)}\left(\log C_i^{(k)} -1-\log C_i^{\mathrm{B}} \right)\mathrm{d} \br\\
\leq& \frac{1}{\log M - 1-\log C_i^{\mathrm{B}}}\left(\sup_{k}\int_{\Omega} S_{1+\log C_i^{\mathrm{B}}}(C_i^{(k)}) \mathrm{d} \br+C_i^{\mathrm{B}}
\right)\\
\rightarrow & ~0 ~\text{as} ~ M \rightarrow \infty.
\end{aligned}
$$
By Dunford-Pettis theorem (see Theorem \ref{thm: DP}), there exists a subsequence, denote also as $C_i^{(k)}$, such that
\begin{equation}
 C_i^{(k)}\rightharpoonup C_i^* ~\text{in} ~ L^1(\Omega) \quad i=1,\cdots, K+1.
\label{eq: weak conv}  
\end{equation}

By Mazur's lemma (see Lemma \ref{lem: Mazur}), there exist convex combinations $v_i^{(k)}=\sum_{j=1}^k\lambda_{j k}C_i^{(k)}$ with $\lambda_{j k}\geq 0$, $\sum_{j=1}^k\lambda_{j k}=1$ such that $v_i^{(k)}\rightarrow C_i^*$ in $L^1(\Omega)$. Then, there exists a subsequence, denote also as $v_i^{(k)}$, such that $v_i^{(k)}\rightarrow C_i^*$ a.e. in $\Omega$.
Let $A_i=\liminf_{k\rightarrow \infty} \int_{\Omega} S_{1+\log C_i^{\mathrm{B}}}(C_i^{(k)}) \mathrm{d} \br$. Then 
$\forall ~ \epsilon>0, \exists ~ N, \forall ~ k>N$ 
$$\int_{\Omega} S_{1+\log C_i^{\mathrm{B}}}(C_i^{(k)}) \mathrm{d} \br\leq A_i+\epsilon. $$
Hence,
$$
\begin{aligned}
\int_{\Omega} S_{1+\log C_i^{\mathrm{B}}}(C_i^*) \mathrm{d} \br
=\int_{\Omega} \lim_{k\rightarrow \infty}S_{1+\log C_i^{\mathrm{B}}}(v_i^{(k)}) \mathrm{d} \br
\leq \liminf_{k\rightarrow\infty} \int_{\Omega} S_{1+\log C_i^{\mathrm{B}}}(v_i^{(k)}) \mathrm{d} \br\\
= \liminf_{k\rightarrow\infty} \int_{\Omega} S_{1+\log C_i^{\mathrm{B}}}\left(\sum_{j=1}^k\lambda_{jk}C_i^{(k)}\right) \mathrm{d} \br\leq \liminf_{k\rightarrow\infty}\int_{\Omega}\sum_{j=1}^k\lambda_{jk}S_{1+\log C_i^{\mathrm{B}}}(C_i^{(k)}) \mathrm{d} \br
\leq A_i+\epsilon.
\end{aligned}
$$
where Fatou's lemma and Jensen's inequality are applied to the first two inequalities. 
Then, the arbitrariness of $\epsilon$ yields
\begin{equation}
    \int_{\Omega} S_{1+\log C_i^{\mathrm{B}}}(C_i^*) \mathrm{d} \br\leq \liminf_{k\rightarrow \infty} \int_{\Omega} S_{1+\log C_i^{\mathrm{B}}}(C_i^{(k)}) \mathrm{d} \br, \quad i=1,\cdots, K+1.
    \label{eq: wlc 1}
\end{equation}

On the other hand, it follows from \eqref{eq: weak conv} that 
$$
\Gamma^{(k)}=1-\eta\sum_{i=1}^{K+1}v_i C_i^{(k)}\rightharpoonup 1-\eta\sum_{i=1}^{K+1}v_i C_i^*=:\Gamma^* ~\text{in} ~ L^1(\Omega),
$$
$$ \rho^{(k)}=\sum_{i=1}^K z_i C_i^{(k)} \rightharpoonup \sum_{i=1}^K z_i C_i^*=:\rho^* ~\text{in} ~ L^1(\Omega).$$
Using the same argument as $C_i$, one has
\begin{equation}
    \int_{\Omega} S_{1+\log \Gamma^{\mathrm{B}}}(\Gamma^*) \mathrm{d} \br\leq \liminf_{k\rightarrow \infty} \int_{\Omega} S_{1+\log \Gamma^{\mathrm{B}}}(\Gamma^{(k)}) \mathrm{d} \br.
    \label{eq: wlc 2}
\end{equation}

Since $L^{-1}$ is positive definite, 
$\|\rho\|_{L}=\left(\int_{\Omega}\rho L^{-1}\rho\mathrm{d} \br\right)^{1/2}$ is a norm of $\rho$ and $|| \rho||_{H^{-2}(\Omega)}$ is equivalent to $\| \rho\|_{L}$.
It follows from \eqref{eq: bound} that $ || \rho^{(k)} ||_{H^{-2}(\Omega)}$ is bounded. 
Hence, there exists a subsequence, still labeled as $\rho^{(k)}$, that converges weakly to some $\varrho$, i.e.,
$$\rho^{(k)} \rightharpoonup \varrho ~\text{in} ~ H^{-2}(\Omega).$$
Let $\xi \in L^{\infty}(\Omega)\cap H_0^{2}(\Omega)$, one has
$$\varrho (\xi)=\lim_{k\rightarrow \infty} \rho^{(k)}(\xi) =\lim_{k\rightarrow \infty}\int_{\Omega} \rho^{(k)} \xi\mathrm{d} \br =\int_{\Omega} \rho^* \xi\mathrm{d} \br.$$
Since $L^{\infty}(\Omega)\cap H_0^{2}(\Omega)$ is dense in $H_0^2(\Omega)$, one can obtain $\rho^* \in H^{-2}(\Omega)$ and $\varrho=\rho^* ~\text{in} ~ H^{-2}(\Omega)$. Hence, $\rho^{(k)} \rightharpoonup \rho^* ~\text{in} ~ H^{-2}(\Omega)$.
Therefore, 
\begin{equation}
 \|\rho^*\|_{L}\leq \liminf_{k\rightarrow\infty} \|\rho^{(k)}\|_{L}.
\label{eq: wlc 3}   
\end{equation}
Combine \eqref{eq: wlc 1}-\eqref{eq: wlc 3}, one can conclude that  $\mathcal{F}$ is weakly lower semi-continuous. Then, 
$$
m\leq \mathcal{F}(\boldsymbol{C}^*)\leq \liminf_{k\rightarrow \infty}\mathcal{F}(\boldsymbol{C}^{(k)})=m.
$$ 
This implies $\boldsymbol{C}^*$ is a minimizer of $\mathcal{F}$.

Assume $\boldsymbol{C}^*, \boldsymbol{D}^*$ are both minimizers of $\mathcal{F}$.
Since $\mathcal{F}$ is strictly convex, 
$$
m\leq \mathcal{F}\left(\frac{\boldsymbol{C}^*+\boldsymbol{D}^*}{2}\right)\leq \frac{1}{2}\mathcal{F}(\boldsymbol{C}^*)+\frac{1}{2}\mathcal{F}(\boldsymbol{D}^*)=m.
$$
The equality holds only if $\boldsymbol{C}^*=\boldsymbol{D}^*$ a.e.
This implies that the minimizer of $\mathcal{F}$ is unique.

Hence, $\phi^*=L^{-1}\rho^*$ is the unique solution to \eqref{eq: no-dim PNPB}.

\end{proof}

\section{Numerical investigation and discussion}
\label{sec: numer}

In this section, we investigate the parameter dependence of the steady state numerically. Starting from the dynamic dimensionless PNPB system \eqref{eq: no-dim dynamic}-\eqref{eq: no-dim 4pbik}, which has an energy dissipated structure, we adopt a finite volume scheme which preserves the energy dissipated property at the semi-discrete level. In the numerical tests, we run the dynamic for time large enough so that the system reaches the equilibrium. Then we can explore the parameter dependence of the steady state.

This section is organized as follows. At first, we show the energy dissipation relation of the dimensionless PNPB system \eqref{eq: no-dim dynamic}-\eqref{eq: no-dim 4pbik} in subsection \ref{1:energy}. Then we apply a semi-discrete finite volume scheme given in \cite{YongZhaoZhou} and prove the energy dissipation relation at the semi-discrete level in subsection \ref{2:scheme}. At last, various numerical investigations are given to show the parameter dependence of the steady state of the model in subsection \ref{3:results}.

\subsection{Energy dissipation relation} \label{1:energy}
The free energy functional $\mathcal{F}$ for the dynamic system \eqref{eq: no-dim dynamic}-\eqref{eq: no-dim ic} is given in the following
\begin{equation*}
\begin{aligned} 
\mathcal{F}(t) =~& \frac{1}{2} \int_{\Omega} \rho(\br, t) \phi(\br, t) \ \mathrm{d} \br \\ 
&+ \int_{\Omega}\left\{\sum_{i=1}^{K+1} C_{i}(\br, t)\left(\ln \frac{C_{i}(\br, t)}{C_{i}^{\mathrm{B}}}-1\right) + \frac{\Gamma(\br, t)}{\eta v_{0}}\left(\ln \frac{\Gamma(\br, t)}{\Gamma^{\mathrm{B}}}-1\right)\right\} \mathrm{d} \br.
\end{aligned}
\end{equation*}
Then we have the following property.
\begin{prop}\label{prop5}
	The energy functional $\mathcal{F}$ is dissipated along solutions of equation \eqref{eq: no-dim dynamic}:
	\begin{equation}
	\frac{\mathrm{d}}{\mathrm{d} t} \mathcal{F}(t) + D(t) = 0,
	\label{eq: dissipat}
	\end{equation}
	where the dissipation $D = \sum\limits_{i=1}^{K+1} \dint_{\Omega} D_i C_i \left|\nabla \frac{\delta \mathcal{F}}{\delta C_i} \right|^2 \, \mathrm{d} \br \geq 0.$
\end{prop}
\begin{proof}
	\begin{equation*}
	\begin{aligned}
	\frac{\mathrm{d}}{\mathrm{d} t} \mathcal{F}(t) &= \sum_{i=1}^{K+1} \int_{\Omega}  \frac{\delta \mathcal{F}}{\delta C_i} \frac{\partial C_{i}(\br, t)}{\partial t} \, \mathrm{d} \br \\
	&= \sum_{i=1}^{K+1} \int_{\Omega} \frac{\delta \mathcal{F}}{\delta C_i}  \nabla \cdot \left(D_i C_i \nabla \frac{\delta \mathcal{F}}{\delta C_i} \right) \, \mathrm{d} \br \\
	&= - \sum_{i=1}^{K+1} \int_{\Omega} D_i C_i \left|\nabla \frac{\delta \mathcal{F}}{\delta C_i} \right|^2 \, \mathrm{d} \br =: -D(t) \leq 0.
	\end{aligned}
	\end{equation*}
\end{proof}

\subsection{Semi-discrete energy dissipated scheme} \label{2:scheme}

In this subsection, taking the one-dimensional case as an example, we adopt the semi-implicit finite volume scheme for the multi-species model in \cite{YongZhaoZhou} which is of second order in space, first order in time. 
	
First, we rewrite the dimensionless equations \eqref{eq: no-dim dynamic}-\eqref{eq: no-dim gamma} in the following symmetric form:
	\begin{equation}
	\label{sym_eq}
	\left\{
	\begin{aligned}
	&\frac{\partial C_i(x,t)}{\partial t} = \nabla \cdot \left( D_i \exp\left\{-f_i\right\} \nabla \dfrac{C_i}{\exp\left\{-f_i\right\}} \right), \quad i = 1, \cdots, K+1, \quad x \in \Omega, \\
	&f_i(x, t) = z_i \phi -\frac{v_i}{v_0} S, \quad i = 1, \cdots, K+1, \\
	&\phi(x, t) = \mathcal{K}*\rho, \quad \rho(x, t) = \sum_{i=1}^{K} z_i C_i,\\
	& S(x, t) = \ln \frac{\Gamma(x, t)}{\Gamma^{\mathrm{B}}}, \quad \Gamma(x, t) = 1 - \eta
	\sum\limits_{i=1}^{K+1}v_i C_i, \quad \Gamma^{\mathrm{B}} = 1 - \eta\sum\limits_{i=1}^{K+1}v_i C_i^B, \\
	 &\left\langle \mathbf{n}, D_i\exp\left\{-f_i\right\} \nabla \dfrac{C_i}{\exp\left\{-f_i\right\}}\right\rangle=0, \quad i = 1, \cdots, K+1, \quad x \in \partial \Omega, \\
	&C_i(x, 0) = C_i^0(x), \quad i = 1, \cdots, K+1,
	\end{aligned}
	\right.
	\end{equation}
	where the kernel function $\mathcal{K}$ is defined in \eqref{eq: kernel}. Hence we rewrite the free energy and the chemical potential in the symmetric form respectively as
	\begin{align}
	    \mathcal{F} &= \int_{\Omega} \left\{ \sum_{i = 1}^{K+1} C_i \log \dfrac{C_i}{\exp\{ -\frac{1}{2} z_i \phi + \frac{v_i}{v_0} S\} } - C_i  + \frac{1}{v_0 \eta} \left( S - \Gamma \right) \right\} \,\mathrm{d} x, \\
	    \mu_i &= \frac{\delta \mathcal{F}}{\delta C_i} = \log \frac{C_i}{C_i^{\mathrm{B}}\exp\{-f_i\}}, 
	\end{align}
	and the equivalent dissipation relation \eqref{eq: dissipat} becomes
	\begin{align}
	\frac{\mathrm{d}}{\mathrm{d} t} \mathcal{F}(t) + \sum_{i = 1}^{K+1} \int_{\Omega} D_i \frac{\exp\{-2 f_i\}}{C_i} \left| \nabla \frac{C_i}{\exp\{-f_i\}} \right|^2 \mathrm{d}x = 0.
	\end{align}
	
	Next we give a brief introduction of the semi-discrete energy dissipated scheme proposed by \cite{YongZhaoZhou}. 
	\subsubsection{Introduction to the scheme}
	Consider a uniform mesh grid $\mathcal T =\{x_{j}\big|~~x_{j} = -L + (j + N) \Delta x, j = -N, \ldots, N, \Delta x= L/N\}$ on the computational domain $[-L, L]$, the semi-discrete finite volume scheme reads as
	\begin{equation}
	\label{bsys2}
	\frac{\mathrm{d} \bar{C}_{i, j}(t)}{\mathrm{d} t} = - \frac{F_{i, j + \frac{1}{2}}(t) - F_{i, j - \frac{1}{2}}(t)}{\Delta x}, ~~i = 1, \cdots, K+1,
	\end{equation}
where $\bar{C}_{i, j}(t)$ is the average concentration of the $i-$th ionic species on $[x_{j - \frac{1}{2}}, x_{j + \frac{1}{2}}]$ for $j = -N+1, \cdots, N - 1$  and $ [x_{-N}, x_{-N + \frac{1}{2}}]$ or $[x_{N - \frac{1}{2}}, x_{N}]$ for $j = -N$ or $N$. The numerical flux ${F}_{i, j + \frac{1}{2}}(t)$ is defined below
\begin{equation}
\label{bsys3}
F_{i, j + \frac{1}{2}}(t) := -\frac{D_i}{\Delta x} \exp\{-f_{i, j + \frac{1}{2}}(t)\} \left\{ \frac{\bar{C}_{i, j + 1}(t)}{\exp\{-f_{i, j + 1}(t)\}} - \frac{\bar{C}_{i, j }(t)}{\exp\{-f_{i, j}(t)\}}\right\},
\end{equation}
where $\exp\{-f_{i, j + \frac{1}{2}}(t)\}$ takes the harmonic mean of $\exp\{-f_{i, j}(t)\}$ and $\exp\{-f_{i, j + 1}(t)\}$, i.e.,
\begin{equation}
\label{bsys5}
\exp\{-f_{i, j + \frac{1}{2}}(t)\} = \frac{2}{\exp\{ f_{i, j}(t)\} +\exp\{ f_{i, j + 1}(t)\} },
\end{equation} 
where $f_{i,j}(t)$ denotes the numerical approximation of  $f_{i}(x_{j},t)$ at time $t$
and can be computed as follows
\begin{equation}
\label{bsys4}
\begin{aligned}
f_{i,j}(t) &= \int_{-L}^{L}  z_i \mathcal{K} \left(x_{j}-x\right) {\rho}^h(x, t) \,\mathrm{d} x - \frac{v_i}{v_0} S_j^{\mathrm{trc}}, \\
&=  \sum_{m = 1}^K   z_i z_m \int_{-L}^L  \mathcal{K} \left(x_{j}-x\right) {C}_{m}^h(x, t) \,\mathrm{d} x - \frac{v_i}{v_0} S_j^{\mathrm{trc}}, \\
&= \sum_{m = 1}^K   z_i z_m \sum_{p=-N}^{N}  \bar C_{m, p}(t) \int_{-L}^{L} \mathcal{K} (x_{j}-x)  e_p(x) \,\mathrm{d} x - \frac{v_i}{v_0} S_j^{\mathrm{trc}}. 
\end{aligned}
\end{equation}
Here the semi-discrete steric potential and the void function become $$S_j^{\mathrm{trc}}(t) = \ln \frac{ \Gamma_j(t)}{\Gamma^{\mathrm{B}}}, ~~\Gamma_j(t) = 1 - \eta \sum\limits_{i=1}^{K+1}v_i \bar{C}_{i,j}(t)$$ and the approximation of the total charged density 
\begin{equation}
\rho^h(x, t) = \sum_{m = 1}^{K} z_m {C}_{m}^h(x, t),
\end{equation}
where ${C}_{m}^h(x, t)$ is chosen as the piecewise linear interpolation of $C_{m}(x,t)$ using $\bar C_{m,j}$, and is given explicitly 
\begin{equation}\label{chFun}
{C}_{m}^h(x, t) = \sum_{j = -N}^{N} \bar C_{m, j}(t) e_j(x),  \quad ~\forall~x \in [-L, L],
\end{equation}
with $e_j(x)$ being the piecewise linear interpolation function, i.e., the typical hat function. Obviously, $C_{m}^{h}$ is a second order approximation of $C_{m}$. 
Define the convolution tensor as
\begin{equation}
T_{j-p}^{\mathcal{K}}:=\Delta x \int_{-1}^{1}\mathcal{K}\left((j-p-x)\Delta x\right)\hat{e}_0(x)\mathrm{d}x,
\label{eq: tensor}
\end{equation}
where $\hat{e}_0(x)= 1-|x|$ for $x \in [-1, 1]$.
As a result, the convolution field  
$$
f_{i,j}(t)=\sum_{m = 1}^K   z_i z_m \sum_{p=-N}^{N}  \bar C_{m, p}(t) T_{j-p}^{\mathcal{K}} - \frac{v_i}{v_0} S_j^{\mathrm{trc}}. 
$$
We remark that the convolution tensors $T_{j-p}^{\mathcal{K}}$ can be precomputed and $f_{i,j}(t)$ can be evaluated efficiently by FFT. More details related to the scheme can be found in \cite{YongZhaoZhou}.


A fully discrete finite volume scheme by applying the backward Euler method while treating the convolution-type field $f_{i, j + 1}$ explicitly in numerical flux term (\ref{bsys3}) reads as follows
\begin{align}
&\frac{\bar{C}_{i, j}^{n + 1} - \bar{C}_{i, j}^{n}}{\Delta t} = - \frac{{F}_{i, j + \frac{1}{2}}^{n + 1} - {F}_{i, j - \frac{1}{2}}^{n + 1}}{\Delta x},\label{sysodeb} \\
&{F}_{i, j + \frac{1}{2}}^{n + 1} = -\frac{D_i}{\Delta x} \exp\{-f_{i, j + \frac{1}{2}}^{n}\} \left\{ \frac{\bar{C}_{i, j + 1}^{n + 1}}{\exp\{-f_{i, j + 1}^{n }\}} - \frac{\bar{C}_{i, j}^{n + 1}}{\exp\{-f_{i, j}^{n}\}}\right\},  \label{sysfluxb}
\end{align}
for all $i = 1, \cdots, K+1$ and $j = -N,\cdots, N$. 
We emphasize that the fully discrete scheme (\ref{sysodeb})-(\ref{sysfluxb}) is only linearly implicit, and thus it avoids the use of nonlinear solvers. 

\subsubsection{Energy dissipation relation of the semi-discrete scheme}
Next we give the result of the energy dissipation relation of the semi-discrete scheme, where the proof is analogous to that in \cite{YongZhaoZhou}.

\begin{thm}\label{thm2b}(1D semi-discrete free energy dissipation estimate)
Consider the one-dimensional semi-discrete finite volume scheme (\ref{bsys2})-(\ref{bsys4}) of the system \eqref{sym_eq} with initial data $\bar C_{i, j}^0 > 0$ and $\bar C_{i, j}(t) > 0, \, i = 1, \cdots, K+1$. 
Then, for the semi-discrete form of the free energy $\mathcal{F}$ and the dissipation $D$, we have
\begin{equation}
	\dfrac{\mathrm{d}}{\mathrm{d} t} E_{\Delta}(t) = - D_{\Delta}(t) \leqslant 0, \quad \forall \,t \geqslant 0.
\end{equation}
\end{thm}
Here $E_{\Delta}(t)$ and $D_{\Delta}(t)$, 
the semi-discrete free energy and dissipation, are defined explicitly as follows
\begin{equation}
\label{disenergyb}
\begin{aligned}
E_{\Delta}(t) = \Delta x \sum_{i = 1}^{K+1} \sum_{j = -N}^{N} \bar{C}_{i, j}(t) \left(\log  \frac{\bar{C}_{i, j}(t)}{C_i^{\mathrm{B}} \exp\{- g_{i, j}(t)\}} - 1\right)+ \frac{\Delta x}{v_0 \eta} \sum_{j = -N}^{N} \left(S_j^{\mathrm{trc}}(t) - \Gamma_j(t)\right), 
\end{aligned}
\end{equation}
\begin{equation}
\begin{aligned}
D_{\Delta}(t) &= \frac{1}{\Delta x} \sum_{i = 1}^{K+1}\sum_{j = -N}^{N} \exp\{-f_{i, j + \frac{1}{2}}(t)\} \cdot \frac{D_i }{\beta_{i, j}(t)} \left( \frac{\bar{C}_{i, j + 1}(t)}{\exp\{-f_{i, j + 1}(t)\}} - \frac{\bar{C}_{i, j }(t)}{\exp\{- f_{i, j}(t)\}}\right)^2,
\end{aligned}
\end{equation}
where $g_i = \frac{1}{2} z_i \phi - \frac{v_i}{v_0} S$, $g_{i, j}$ is defined similarly to $f_{i,j}$, and $\beta_{i, j}(t)$ sits between $\frac{\bar{C}_{i, j }(t)}{\exp\{-f_{i, j}(t)\}}$ and $\frac{\bar{C}_{i, j + 1}(t)}{\exp\{- f_{i, j + 1}(t)\}}$. 

\begin{proof}
	We only need to prove that 
	\begin{equation*}
	\begin{aligned}
	&\dfrac{\mathrm{d}}{\mathrm{d} t} E_{\Delta}(t) \\
	=~& \dfrac{\mathrm{d}}{\mathrm{d} t} \left[ \Delta x \sum_{i = 1}^{K+1} \sum_{j = -N}^{N} \bar{C}_{i, j}\left( \log  \frac{\bar{C}_{i, j}}{C_i^{\mathrm{B}} \exp\{- g_{i, j}\}}  - 1\right) + \frac{\Delta x}{v_0 \eta} \sum_{j = -N}^{N} (S_j^{\mathrm{trc}} - \Gamma_j)\right] \\
	=~& \Delta x \sum_{i = 1}^{K+1} \sum_{j = -N}^{N} \left[ \log \dfrac{\bar{C}_{i, j}}{C_i^{\mathrm{B}} \exp\{ -g_{i, j} \}} \dfrac{\mathrm{d} }{\mathrm{d} t} \bar{C}_{i, j} + \bar{C}_{i, j} \dfrac{\mathrm{d}}{\mathrm{d} t} g_{i, j} \right] + \frac{\Delta x}{v_0 \eta} \sum_{j = -N}^{N} \dfrac{\mathrm{d} }{\mathrm{d} t} (S_j^{trc} - \Gamma_j) \\
	=~& \Delta x \sum_{i = 1}^{K+1} \sum_{j = -N}^{N}  \log \dfrac{\bar{C}_{i, j}}{C_i^{\mathrm{B}} \exp\{ -f_{i,j} \}} \dfrac{\mathrm{d} }{\mathrm{d} t} \bar{C}_{i, j}  - \Delta x \sum_{i = 1}^{K+1} \sum_{j = -N}^{N} \frac{v_i \bar{C}_{i,j}}{v_0 \Gamma_j} \dfrac{\mathrm{d} }{\mathrm{d} t} \Gamma_{j} + \frac{\Delta x}{v_0 \eta} \sum_{j = -N}^{N} \dfrac{\mathrm{d} }{\mathrm{d} t} (S_j^{\mathrm{trc}} - \Gamma_j) \\
	=~& \Delta x \sum_{i = 1}^{K+1} \sum_{j = -N}^{N} \log \dfrac{\bar{C}_{i, j}}{C_i^{B} \exp\{ -f_{i,j} \}} \dfrac{\mathrm{d} }{\mathrm{d} t} \bar{C}_{i, j} 
	\end{aligned}
	\end{equation*}
	The third equality holds because of $T^{\mathcal{K}}_{p-j} = T^{\mathcal{K}}_{j-p}$.
	
	According to (\ref{bsys2}), we have 
	\begin{equation*}
	\dfrac{\mathrm{d}}{\mathrm{d} t} E_{\Delta}(t) = - \sum_{i = 1}^{K+1} \sum_{j = -N}^{N}  \left( \log \dfrac{\bar{C}_{i, j}}{C_i^{\mathrm{B}}\exp\{ - f_{i, j} \}} \right)\left(F_{i, j + \frac{1}{2}} - F_{i, j - \frac{1}{2}}\right).
	\end{equation*}
	Since we take no-flux boundary conditions, the discrete boundary conditions satisfy 
	$$F_{i, -N - \frac{1}{2}} = F_{i, N+ \frac{1}{2}} = 0, \, i = 1, \cdots, K+1.$$
	Hence, using Abel's summation formula, we obtain
	\begin{equation}
	\dfrac{\mathrm{d}}{\mathrm{d} t} E_{\Delta}(t) 
	=- \sum_{i = 1}^{K+1} \sum_{j = -N}^{N} \left( \log \dfrac{\bar{C}_{i, j}}{C_i^{\mathrm{B}} \exp\{ - f_{i, j} \}} - \log \dfrac{\bar{C}_{i, j + 1}}{C_i^{\mathrm{B}} \exp\{ - f_{i, j + 1} \}} \right) F_{i, j + \frac{1}{2}}.
	\label{eq: dt_dis_E}
	\end{equation}
	Plug \eqref{bsys3} into \eqref{eq: dt_dis_E} and apply the mean-value theorem yields
	\begin{equation*}
 	\begin{aligned}
    	\dfrac{\mathrm{d}}{\mathrm{d} t} E_{\Delta}(t) 
	=&- \frac{1}{\Delta x} \sum_{i = 1}^{K+1} \sum_{j = -N}^{N} \exp\{-f_{i, j + \frac{1}{2}}(t)\} \cdot \frac{D_i }{\beta_{i, j}(t)} \left( \frac{\bar{C}_{i, j + 1}(t)}{\exp\{-f_{i, j + 1}(t)\}} - \frac{\bar{C}_{i, j }(t)}{\exp\{-f_{i, j}(t)\}}\right)^2 \\
	=& - D_{\Delta}(t) \leqslant 0, \quad \forall t \geqslant 0,
	\end{aligned}
	\end{equation*}
	where $\beta_{i, j}(t)>0$ is a point between $\frac{\bar{C}_{i, j }(t)}{\exp\{-f_{i, j}(t)\}}$ and $\frac{\bar{C}_{i, j + 1}(t)}{\exp\{-f_{i, j + 1}(t)\}}$.
\end{proof}

\begin{rem}
The positivity preserving property of the concentrations $C_m$ and $\Gamma$ at both continuous and discrete level remains to be done for future work.
\end{rem}

\subsection{Results and discussion on the one-dimensional case} \label{3:results}

Before talking about the parameter dependence, we slightly modify the dimensionless equations (\ref{eq: no-dim dynamic}), (\ref{eq: no-dim gamma}) by adding an external electric field $V_0$, for $i = 1, \cdots, K+1$, the domain $\Omega = [-1, 1]^d$,
\begin{equation}
\label{model}
\left\{
\begin{aligned}
&\frac{\partial C_i(x,t)}{\partial t} = \nabla\cdot\left(D_i\left(\nabla C_i(x,t) + z_i C_i(x,t)\nabla\phi(x,t)-\frac{v_i}{v_0}C_i(x,t)\nabla S(x,t) + z_i C_i \nabla V_0(x)\right)\right), \quad x \in \Omega, \\
&S(x, t) = \ln \frac{\Gamma(x, t)}{\Gamma^{\mathrm{B}}}, \quad
\Gamma(x, t) = 1 - \eta\sum\limits_{i=1}^{K+1}v_i C_i,\quad \Gamma^{\mathrm{B}} = 1 - \eta\sum\limits_{i=1}^{K+1}v_i C_i^{\mathrm{B}},  \quad x \in \Omega, \\
&L \phi(x) =\nu^2\left(\lambda^2\Delta - 1\right)\Delta \phi=\rho(x), \quad \rho(x) =  \sum_{i=1}^{K} z_{i} C_{i}(x),  \quad x \in \Omega, \\
&\left\langle \mathbf{n}, D_i\left(\nabla C_i(x,t) + z_i C_i(x,t)\nabla\phi(x,t)-\frac{v_i}{v_0}C_i(x,t)\nabla S(x,t) + z_i \nabla V_0\right)\right\rangle = 0, \quad x \in \partial \Omega, \\
&C_i(x, 0) = C_i^0(x), \quad x \in \Omega.
\end{aligned}
\right.
\end{equation}

In this section, we take the one-dimensional case $d = 1$ as an example and for simplicity choose the number of the ionic species except water $K = 2$. Without loss of generality, the first ionic species is assumed to be positive while the second one is assumed to be negative. Let the external field $V_0 = 10 x$ and the constant diffusion coefficients $(D_1, D_2, D_3) = (1, 1, 1)$. In addition, we set the initial conditions in the following form,
\begin{equation}
\label{ini}
C_i^0(x) = 0.5, \quad i = 1, 2, 3.
\end{equation}
Hence, $C_i^{\mathrm{B}}=0.5$.


We recall that one of the statements of the steady state is that the chemical potential $\mu_i$ is constant for $i = 1, \cdots, K +1$, i.e.
\begin{align}
\label{steadystate}
\mu_i = \ln \frac{C_{i}(x)}{C_{i}^{\mathrm{B}}}-\frac{v_{i}}{v_{0}} \ln \frac{\Gamma(x)}{\Gamma^{\mathrm{B}}}+z_{i} \phi(x) + z_i V_0(x) = \text{constant.}
\end{align}
It is seen that the steady state of system \eqref{model} is the result of multiple effects: 
the diffusion of all the ionic and water species caused by concentration gradient $-\nabla C_i$, the electric force $-z_i\nabla \phi$, the steric force $\frac{v_i}{v_0} \nabla S$ as well as the external force $-z_i\nabla V_0$. 

Since we aim to observe those effects, we give some further assumptions on the electric field $\phi$ and the kernel function $\mathcal{K}$:
\begin{itemize}
    \item[(1)] the correlated electric field is considered in the whole space, which means that the 4PBik equation $L \phi(x) = \rho(x)$ is satisfied for all $x \in \mathbb{R}$;
    \item[(2)] the kernel function $\mathcal{K}$ of the electric field is given in the form of Picard solution;
    \item[(3)] the correlated electric field decays to zero at the infinity.
\end{itemize}
Following the above assumptions, the 1D kernel function becomes 
\begin{align}
\label{kernel}
\mathcal{K}_{\nu, \lambda} (x) = \frac{\lambda}{2 \nu^2} \exp\left(-\frac{|x|}{\lambda}\right)
\end{align}
as an example.

It must be emphasized that the external field $V_0 = 10 x$ would push the cations to the left while push the anions to the right. Hence the inclusion of the external field $V_0 = 10 x$ also contributes to the formation of the boundary layer. Besides, the steric force represents the force of void exerted on the ions and water molecules, the sign of steric force suggests that the diffusion of the void is in the opposite direction to the diffusion of other particles.

The degree and manner of aggregation of steady-state concentrations are influenced by \eqref{steadystate}. Combined with \eqref{kernel}, the relevant parameters can be divided into three parts: the parameter $\eta$, which is in the definition of $\Gamma$ and thus related to the steric effect $S$, the parameters $\lambda, \nu$, which are related to the correlated electric field $\phi$, and the parameters $z_i$, $v_i$, which are the intrinsic properties of ions. Next, we will use numerical tests to investigate the effect of all the parameters we mentioned before.

\subsubsection{Parameter related to the steric effect: $\eta$}

We can see from the expression of the steric potential
$$
S(x) = \ln \dfrac{1-\eta\sum_{i=1}^{K+1}v_i C_i(x)}{\Gamma^{\mathrm{B}}}
$$
that $\eta$ is related to the steric potential $S$ directly. Additionally, $\eta = 0$ means the steric effect vanishes. Combined with \eqref{steadystate}, larger $\eta$ makes the steric repulsion stronger and hence the concentrations of ionic species at the boundary layers are less peaked for all the species.
Next we illustrate how $\eta$ has impact on the equilibrium of the PNPB model through two tests.

At first, we pay attention to different $\eta$. Fix the parameters $\lambda = 1, \nu = 1, (z_1, z_2, z_3) = (1, -1, 0), (v_1, v_2, v_3) = (0.01, 0.01, 0.01)$, and hence $v_0 = 0.01$. Let $\eta = 0, 1, 3, 5$, Fig \ref{plot_test1} shows the results of concentrations, total charged density $\rho=z_1 C_1+z_2 C_2$, total volume density $v_1C_1+v_2C_2+v_3C_3$ and voids at time $t = 1$ with different $\eta$. Note that the system is near equilibrium at $t=1$ since the distributions are not changing any more. We observe that the inclusion of the steric potential makes the concentrations of ions $C_1$ and $C_2$ not overly peaked. Larger $\eta$ leads to larger steric repulsion effects and thus the the cations are less gathered around the left boundary while the anions are less gathered on the right. 
As for water, since the water molecules are neutral, they are not subjected to the electric field and the external field.
Hence, when $\eta = 0$, there is no steric effect, the distribution of water $C_3$ remain unchanged. When $\eta>0$, the steric force squeezes water from the boundaries to the middle. Moreover, larger $\eta$ makes the water molecules much gathered in the middle part. 

\begin{figure}[htp]
    \centering
    \includegraphics[width=4.6cm,height=4.8cm]{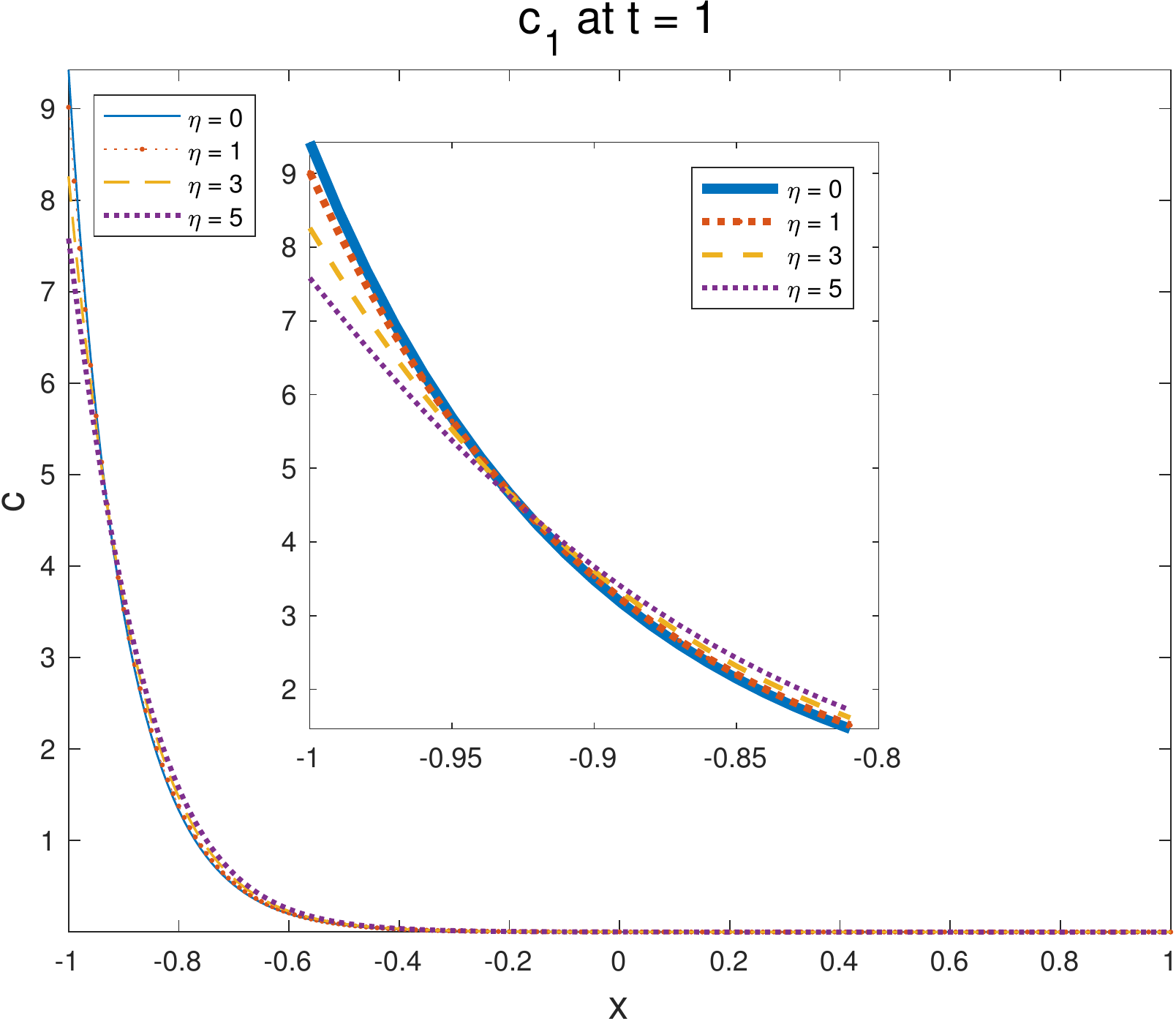}
    \includegraphics[width=4.6cm,height=4.8cm]{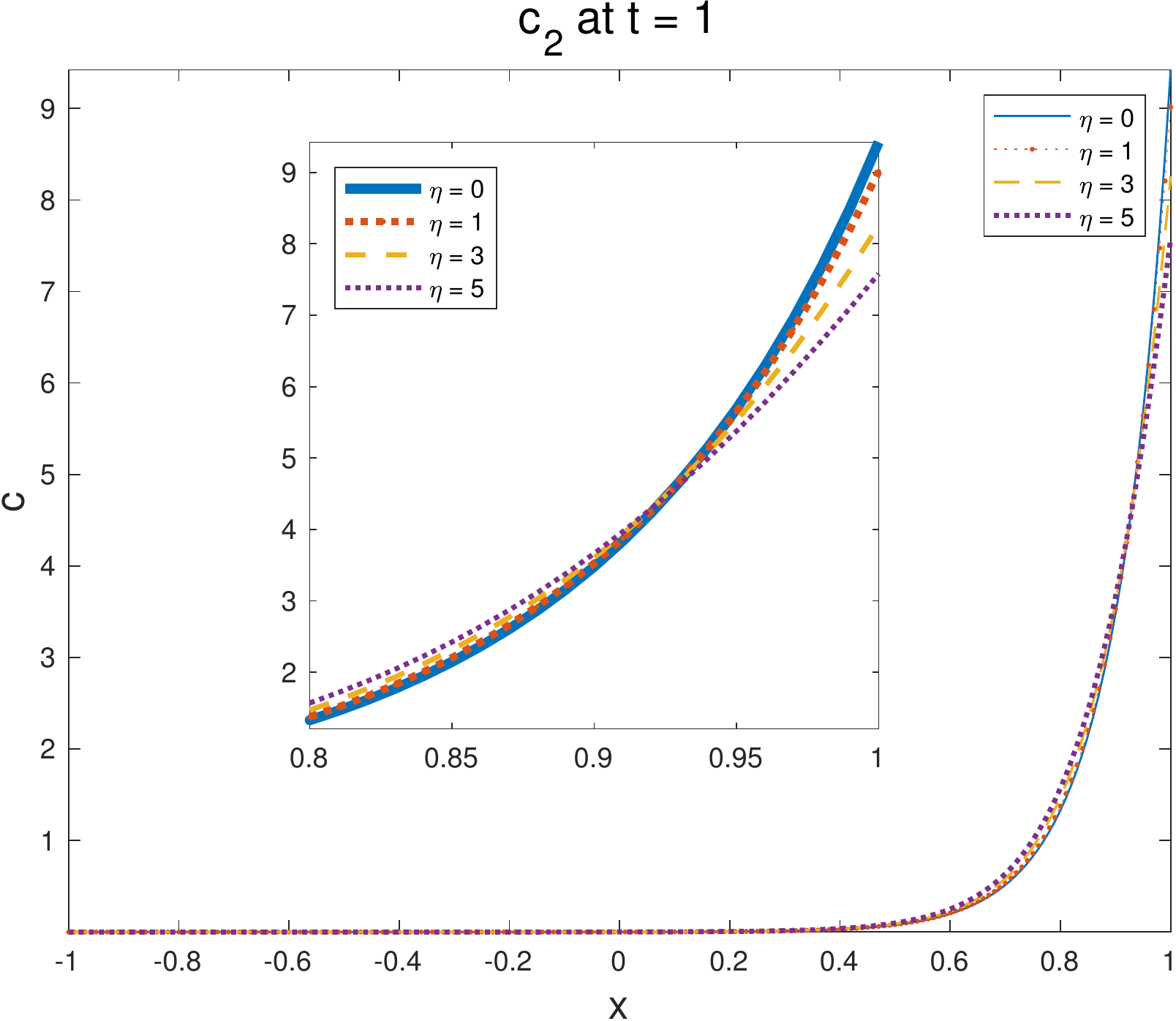}
    \includegraphics[width=4.5cm,height=4.8cm]{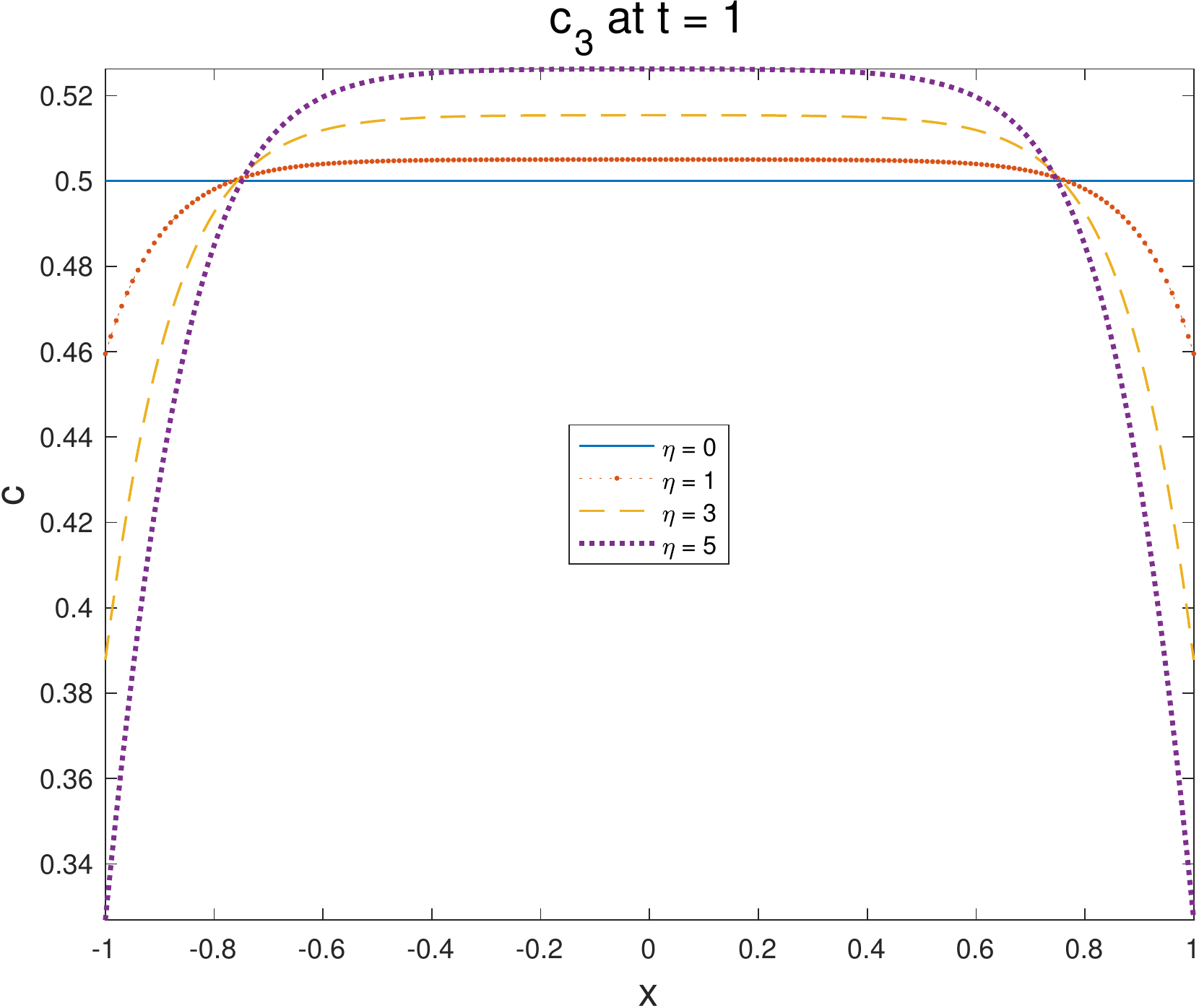}
    
    \includegraphics[width=4.5cm,height=4.8cm]{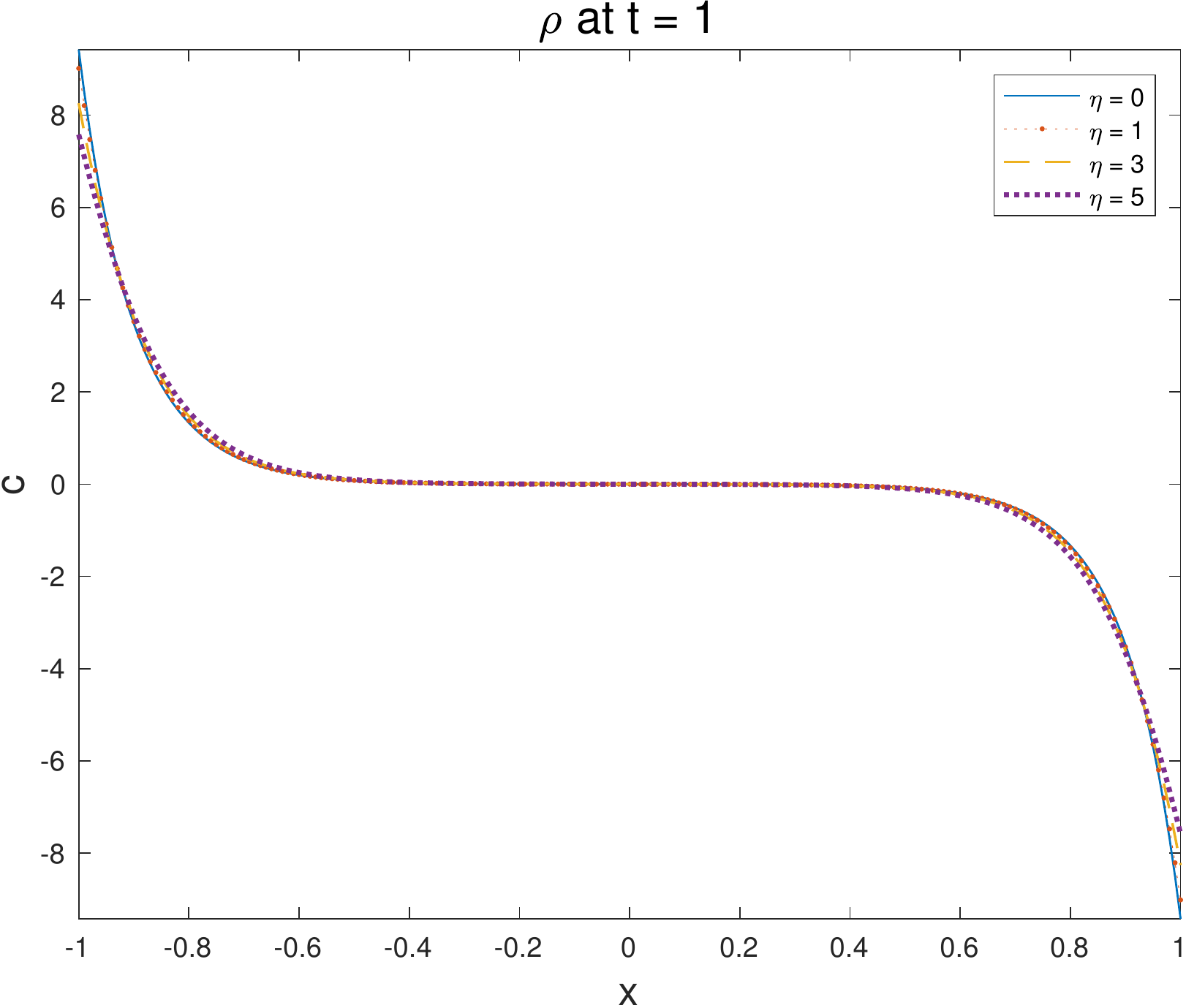}
    \includegraphics[width=4.5cm,height=4.8cm]{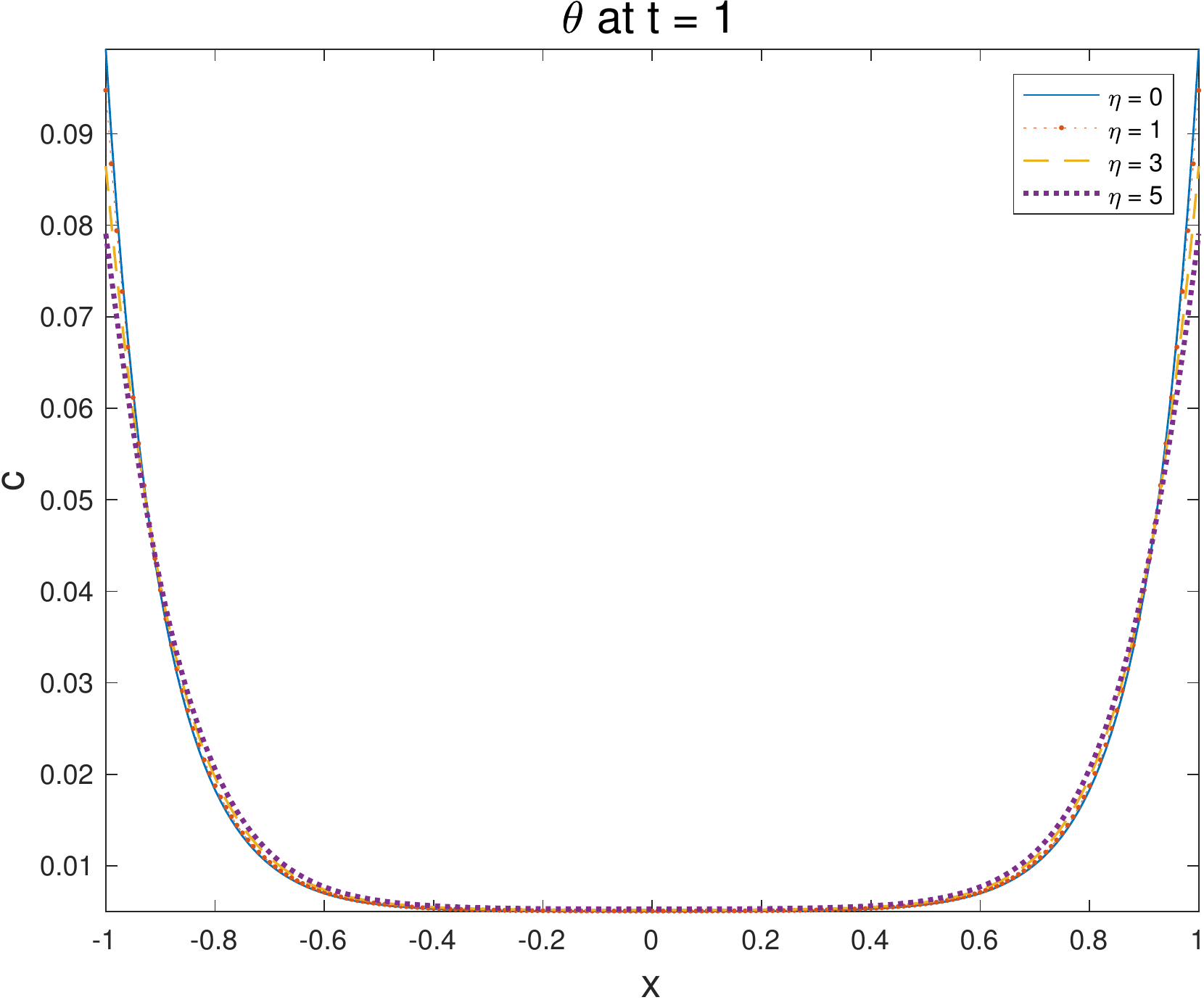}
    \includegraphics[width=4.5cm,height=4.8cm]{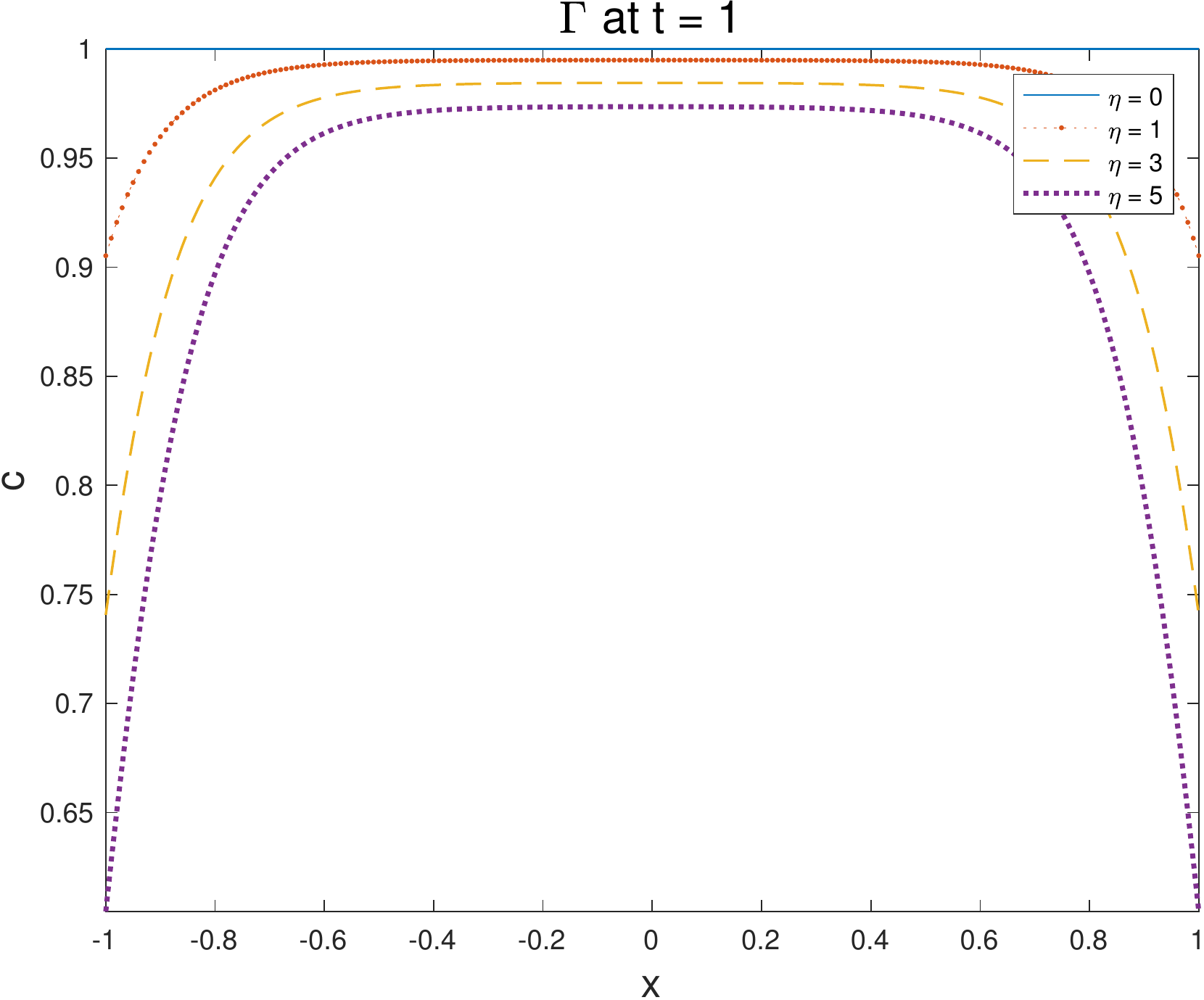}
    \caption{$C_i(x, 1)$ for $i = 1, 2, 3$, $\rho(x, 1), \theta(x, 1)$ and $\Gamma(x, 1)$ with the $\Delta t = 0.005, \Delta x = 0.01$ and $\eta = 0, 1, 3, 5$.}
	\label{plot_test1}
\end{figure}

In the second test, we want to observe the critical value of $\eta$ that makes $\Gamma$ positive. The parameters are set to be $\lambda = 1, \nu = 1, (z_1, z_2, z_3) = (1, -1, 0), (v_1, v_2, v_3) = (0.01, 0.01, 0.01)$. Let $\eta = 8, 8.1, 8.21$, Fig \ref{plot_test2} shows the results of concentrations, total charged and volume densities and voids at time $t = 1$. We can observe that when considering the initial condition \eqref{ini}, taking $\eta = 8.21$ with the mesh size $\Delta t = 0.005, \Delta x = 0.01$ leads to $\Gamma(x_0, 1) \leq 0$ for some $x_0$ in the domain $[-1, 1]$, while smaller $\eta$ can preserve the positivity of $\Gamma$. This indicates that $\eta$ can not be too large, otherwise, it will leads to a meaningless $\Gamma$. To get a better understanding of $\eta$ and the positivity preserving property of $\Gamma$ in a dynamic PNPB model with a certain mesh size, a more detailed discussion will be considered in the future.


\begin{figure}[htp]
    \centering
    \includegraphics[width=4.5cm,height=4.8cm]{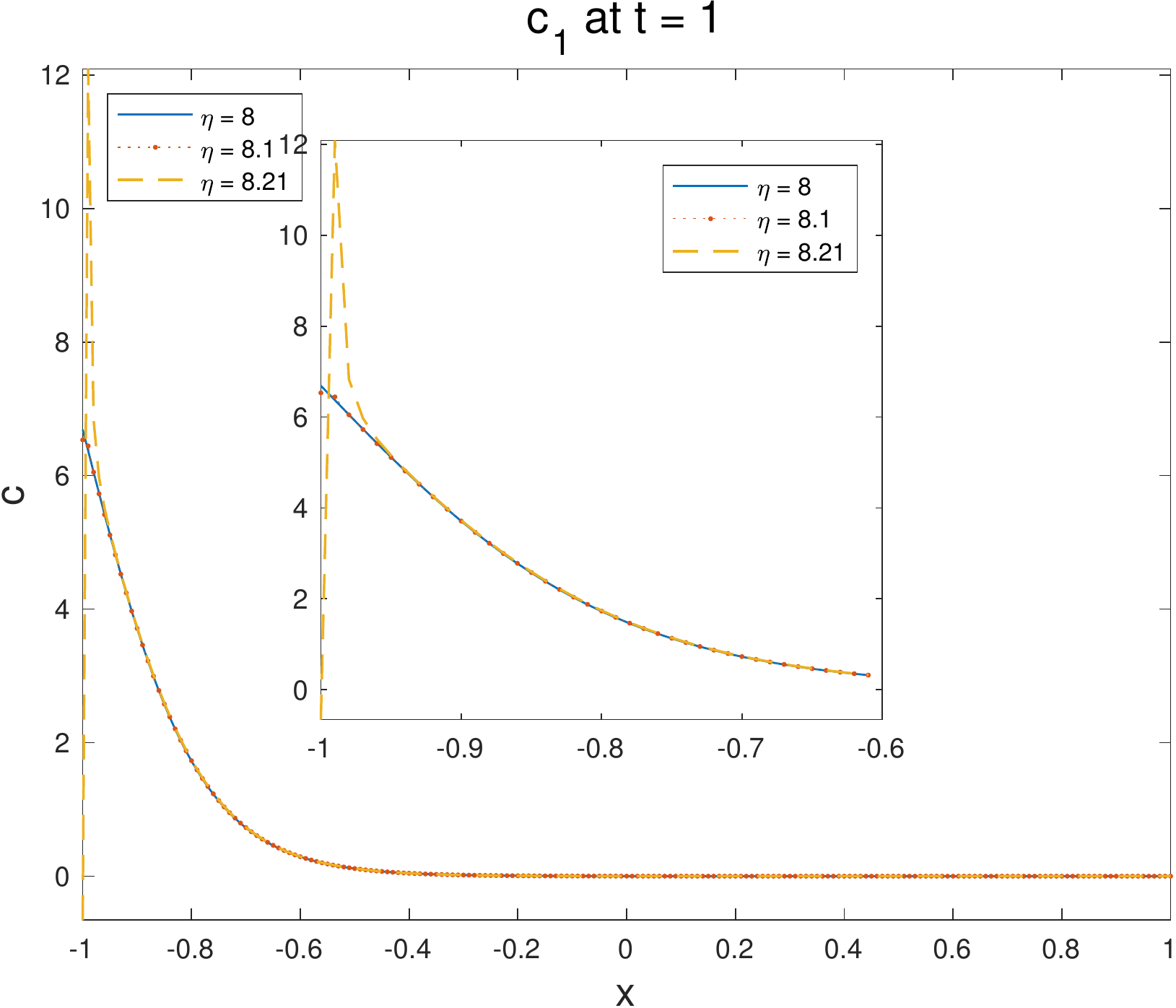}
    \includegraphics[width=4.5cm,height=4.8cm]{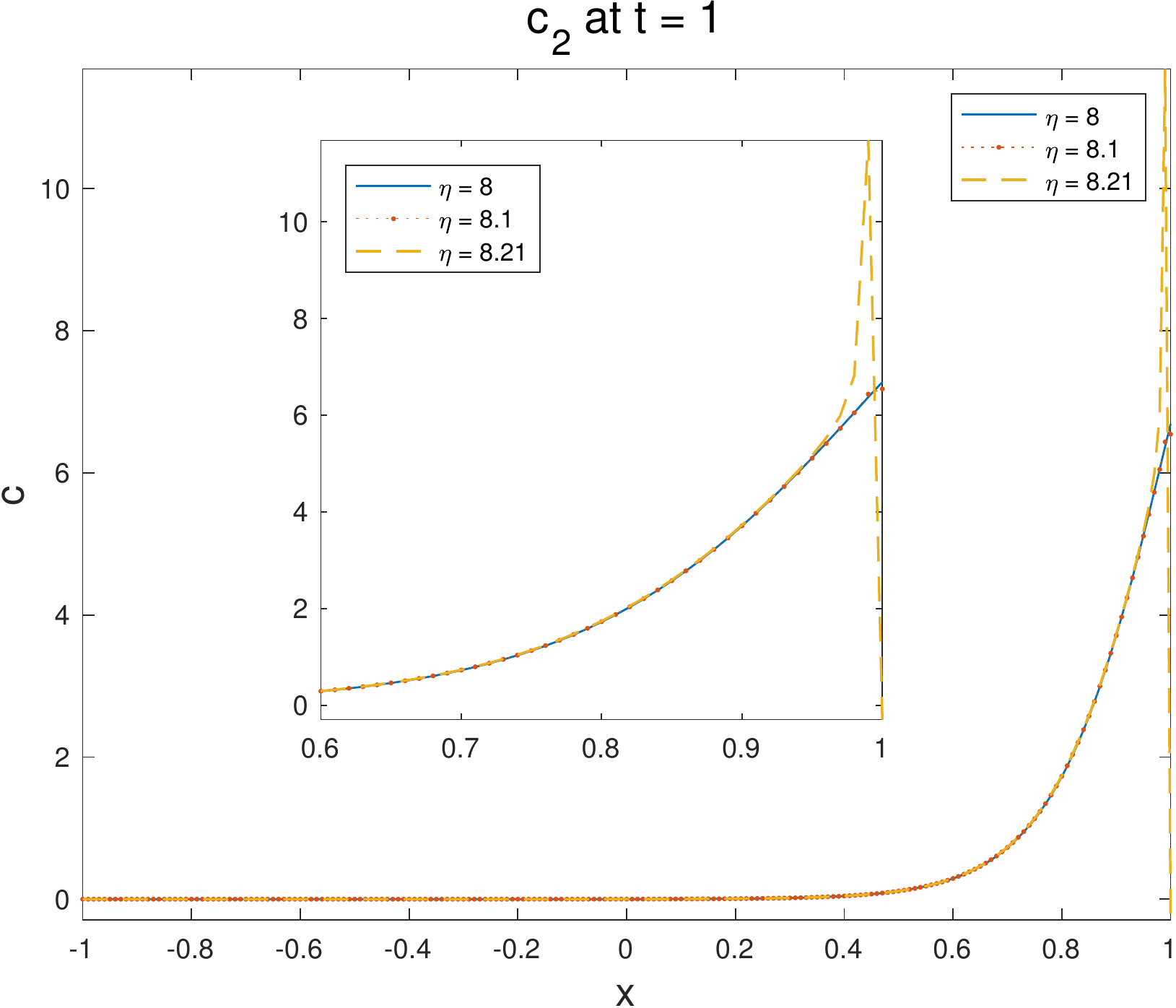}
    \includegraphics[width=4.5cm,height=4.8cm]{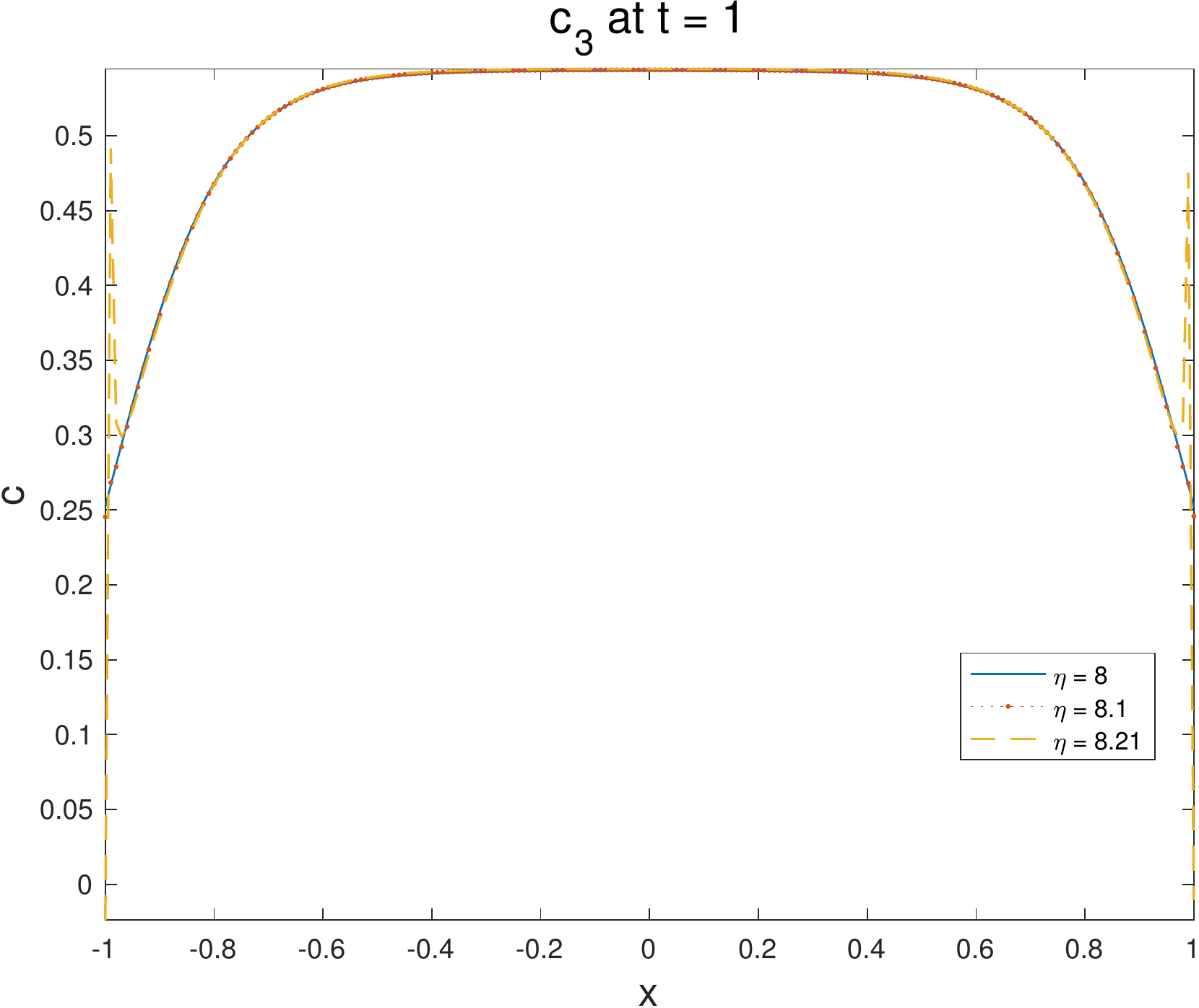}
    
    \includegraphics[width=4.5cm,height=4.8cm]{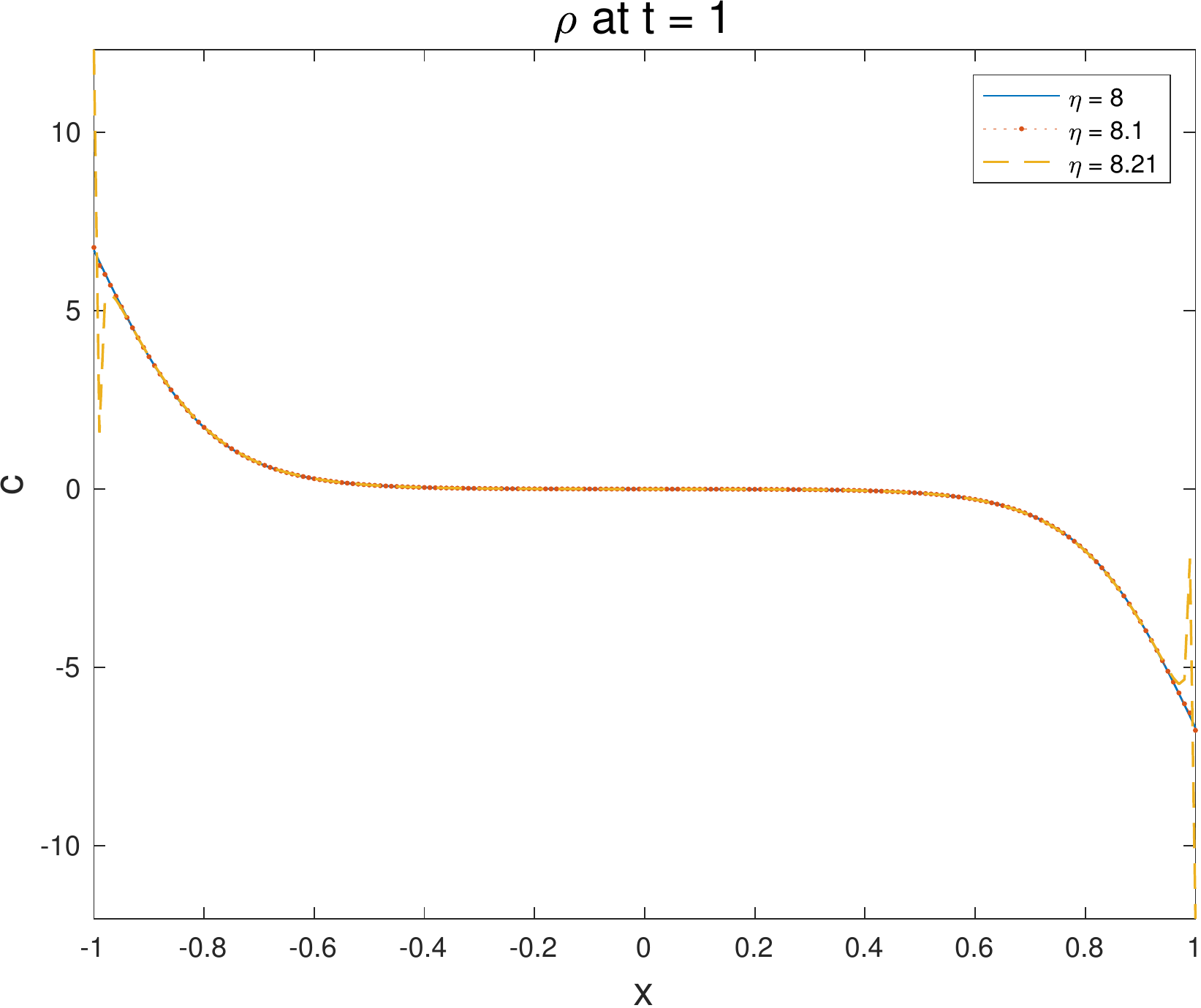}
    \includegraphics[width=4.5cm,height=4.8cm]{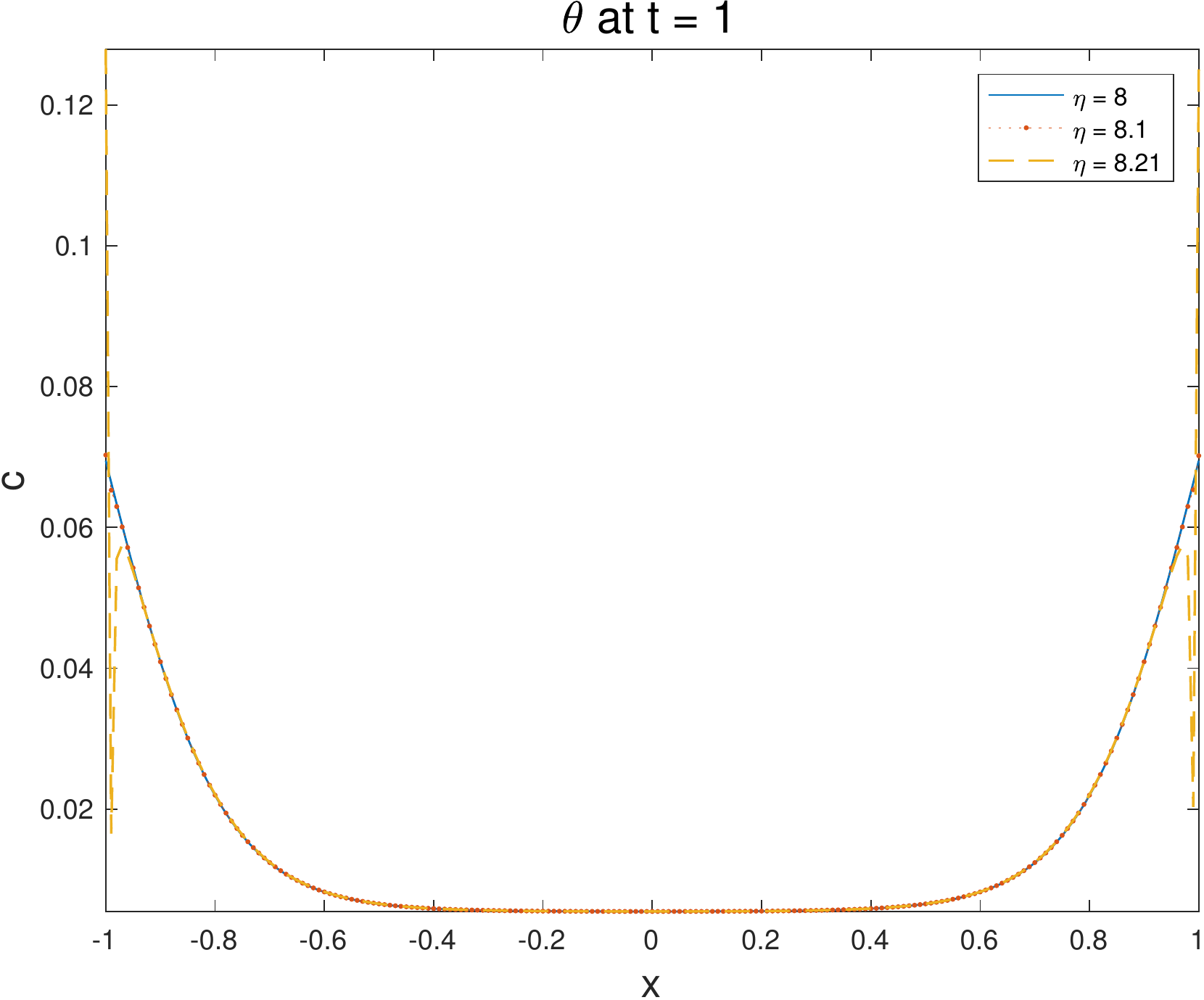}
    \includegraphics[width=4.5cm,height=4.8cm]{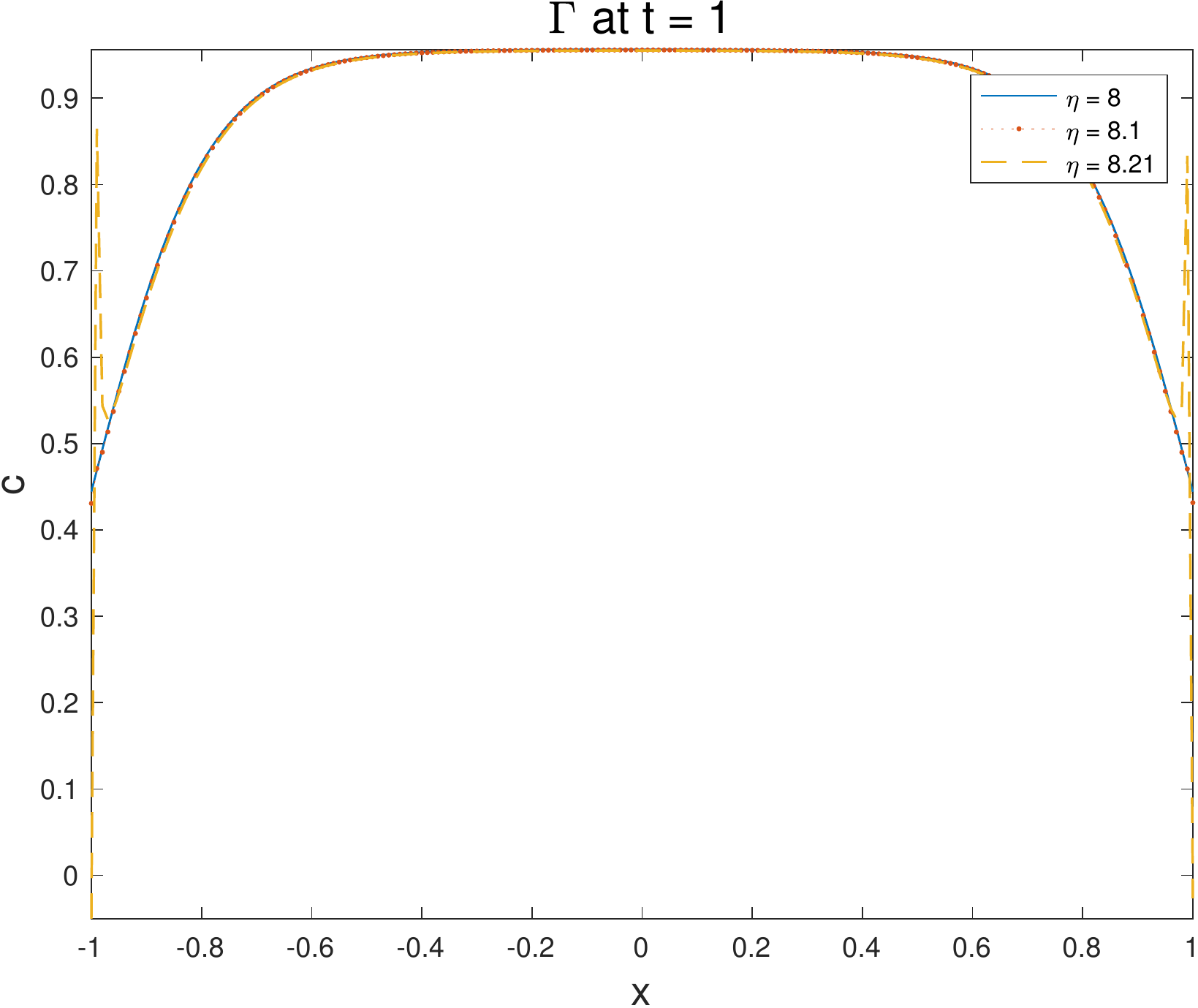}
    \caption{$C_i(x, 1)$ for $i = 1, 2, 3$, $\rho(x, 1), \theta(x, 1)$ and $\Gamma(x, 1)$ with the $\Delta t = 0.005, \Delta x = 0.01$ and $\eta = 8, 8.1, 8.21$.}
	\label{plot_test2}
\end{figure}

\subsubsection{Parameters related to the correlated electric field: $\nu, \lambda$}
Recall that the correlated electric potential $\phi = \mathcal{K}_{\nu, \lambda}*\rho$, where we take the kernel $\mathcal{K}_{\nu, \lambda} (x) = \frac{\lambda}{2 \nu^2} \exp(-|x|/\lambda)$. The dimensionless parameters $\lambda, \nu$ play the role of correlation length and Debye length after scaling and thus have an effect on the correlated electric potential $\phi$, whose gradient influences the dynamic of the concentrations $C_i$, for all $i$.

In the third numerical test, the parameters $\eta = 1, \lambda = 1, (z_1, z_2, z_3) = (1, -1, 0), (v_1, v_2, v_3) = (0.01, 0.01, 0.01)$, and hence $v_0 = 0.01$. Fig \ref{plot_test5_nu} illustrates the results of concentrations, total charged/volume densities and the electric potential at time $t = 1$ with respect to $\nu = 1/3, 1, 3$. The larger the parameter $\nu$ is, the weaker internal correlated electric potential $\phi$ is, hence the stronger the aggregation is on the boundary layer. Clearly, we can see from Fig \ref{plot_test5_nu} that for larger $\nu$, the peaks of concentrations $C_1$ and $C_2$ are higher on the boundary layers. 


\begin{figure}[htp]
    \centering
    \includegraphics[width=4.6cm,height=4.8cm]{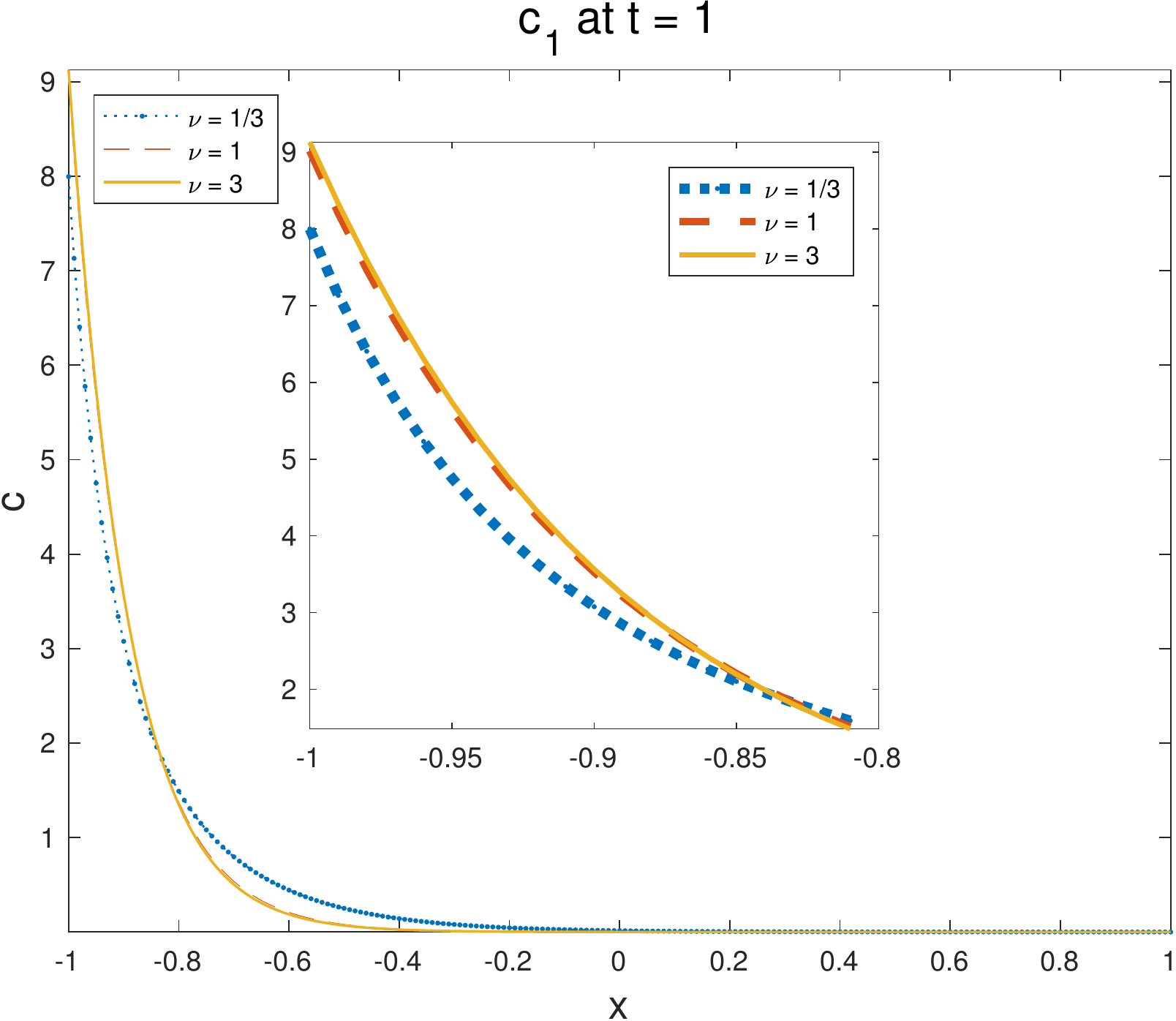}
    \includegraphics[width=4.6cm,height=4.8cm]{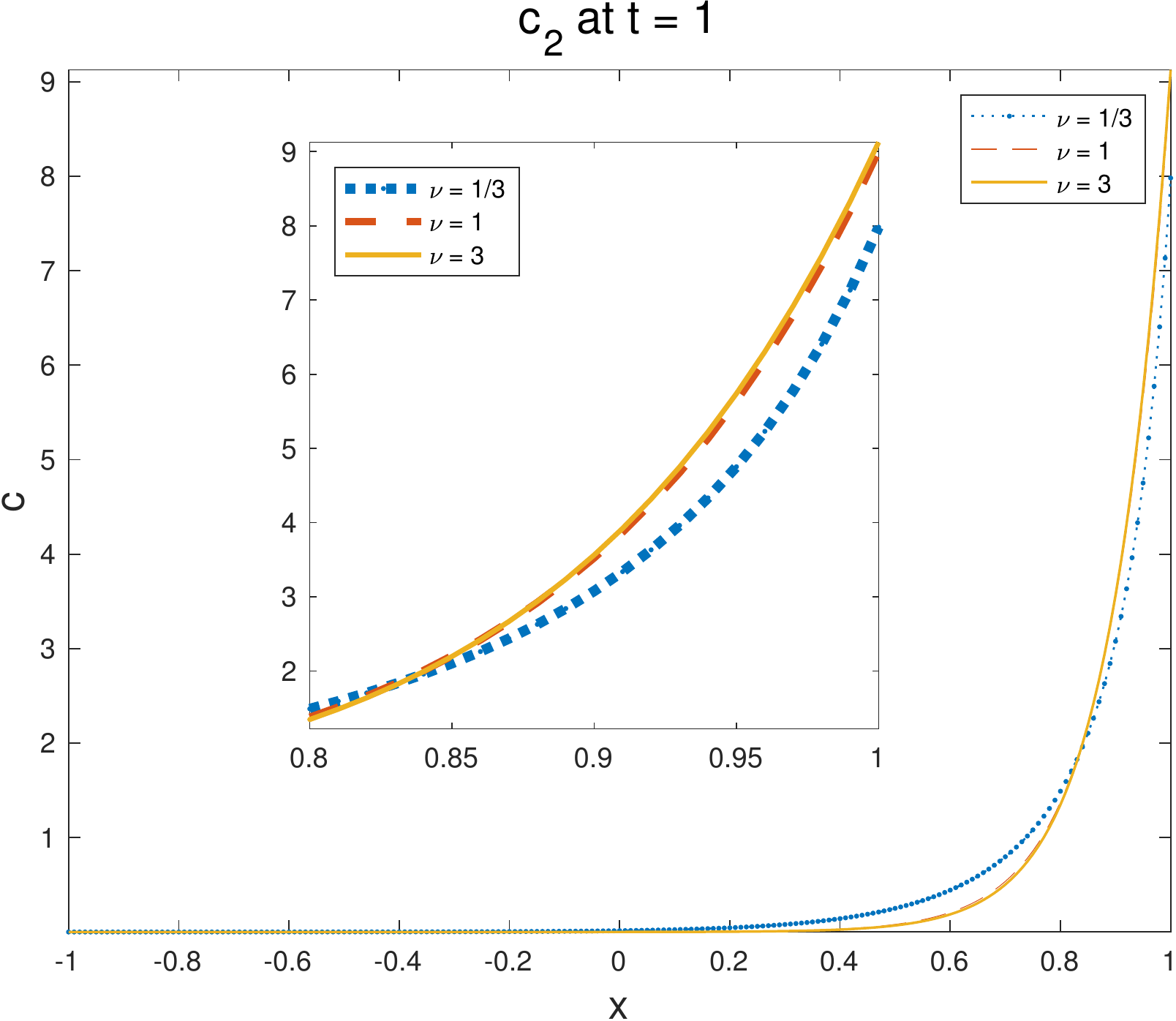}
    \includegraphics[width=4.6cm,height=4.8cm]{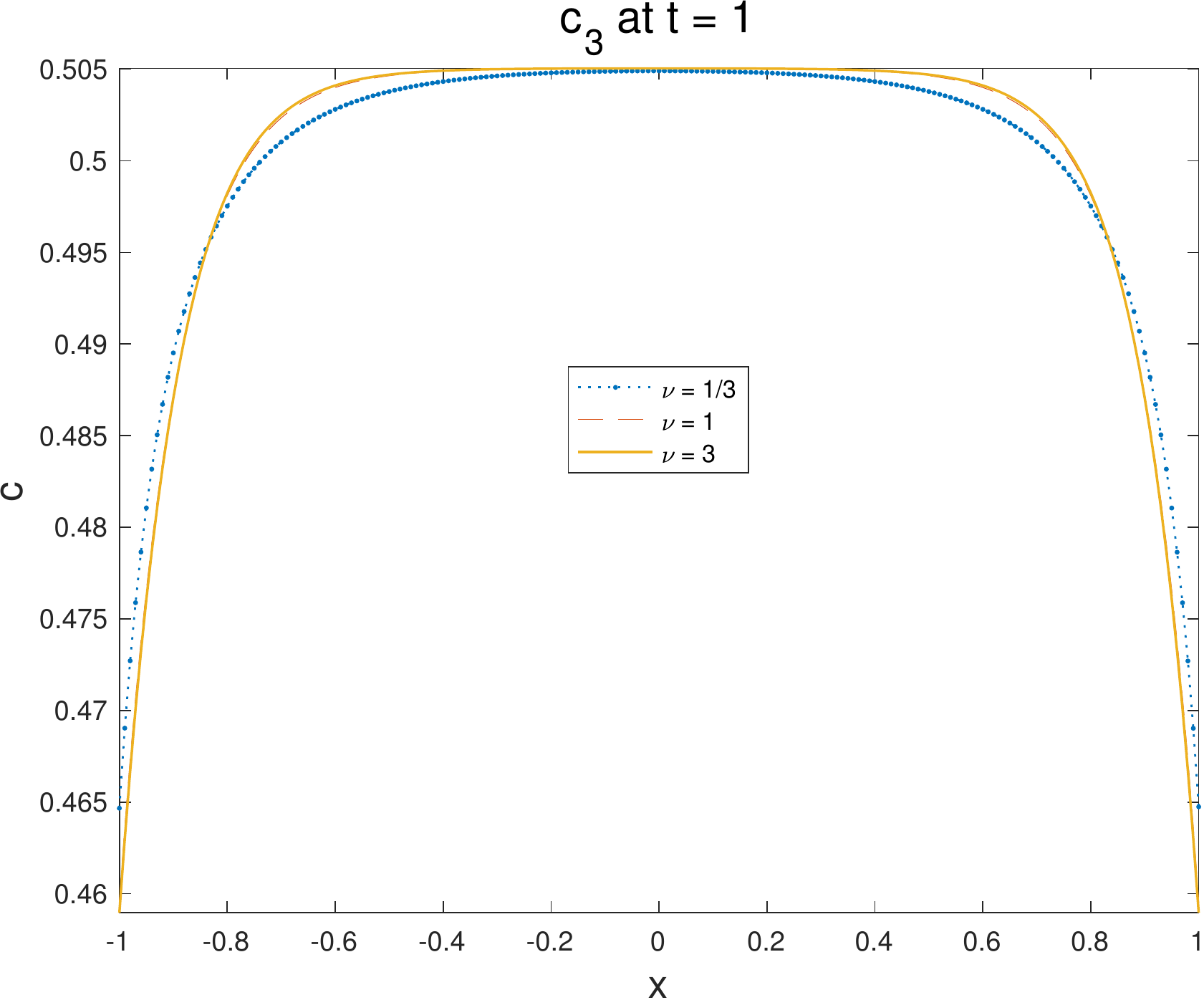}
    
    \includegraphics[width=4.6cm,height=4.8cm]{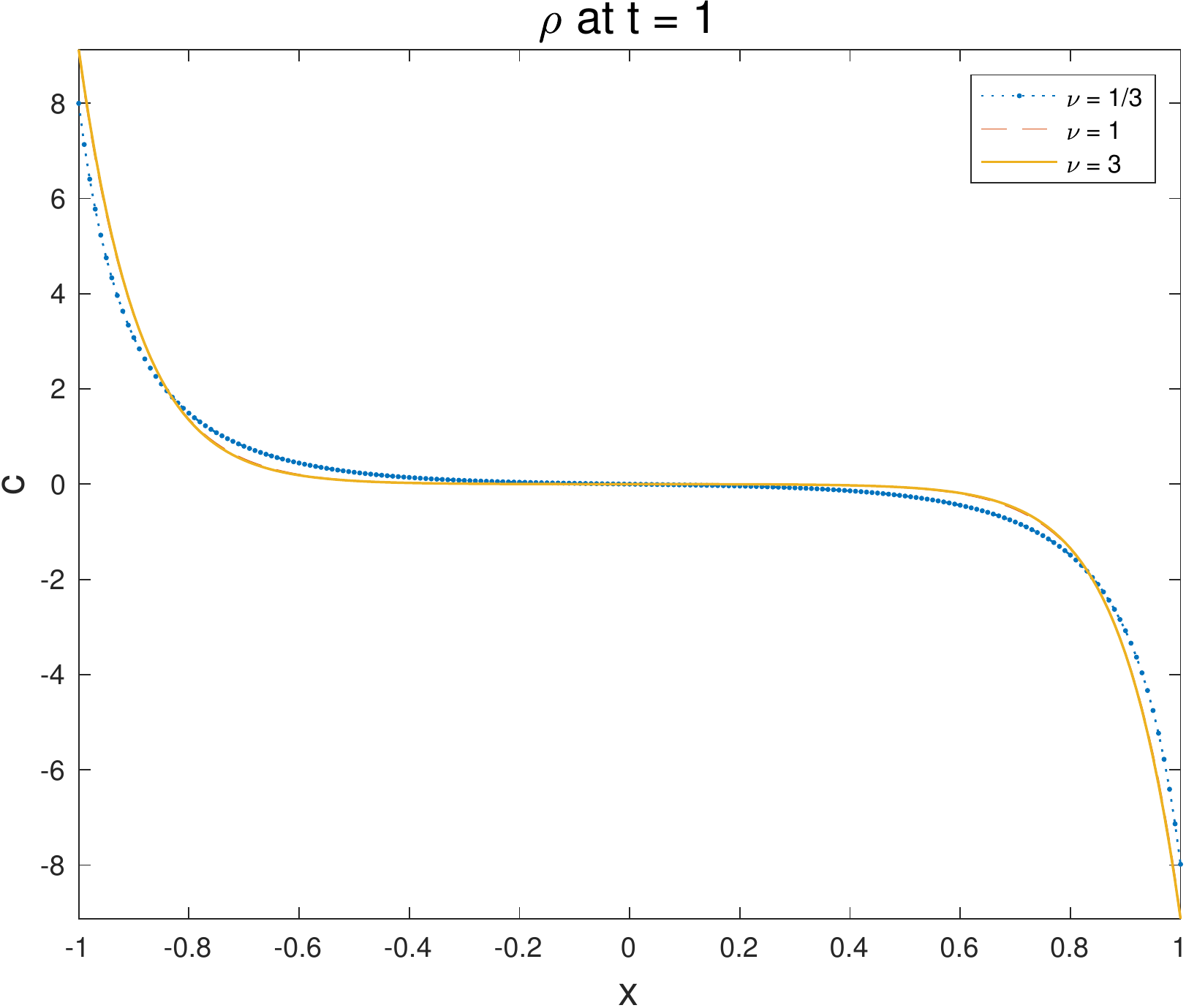}
    \includegraphics[width=4.6cm,height=4.8cm]{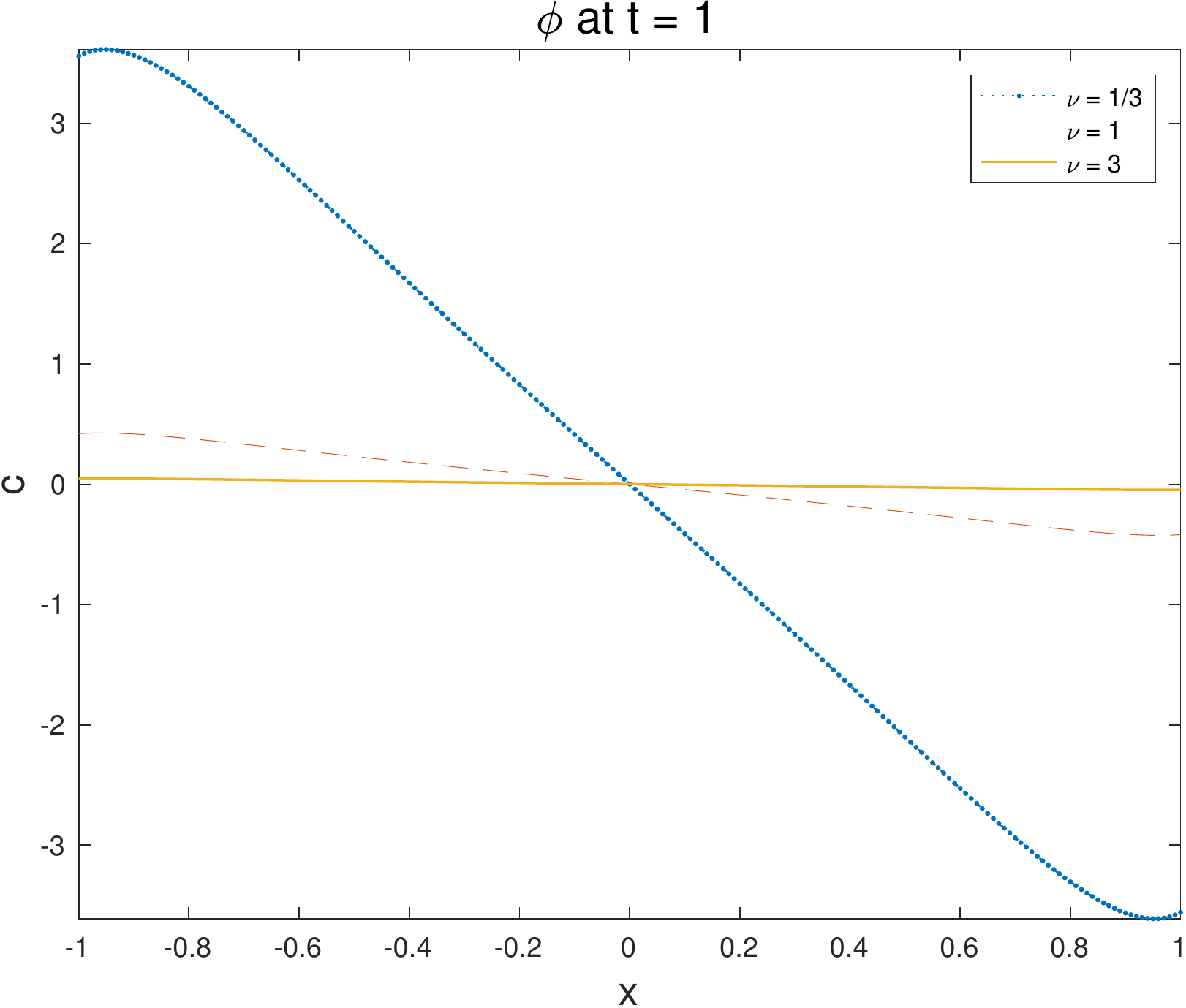}
    \includegraphics[width=4.6cm,height=4.8cm]{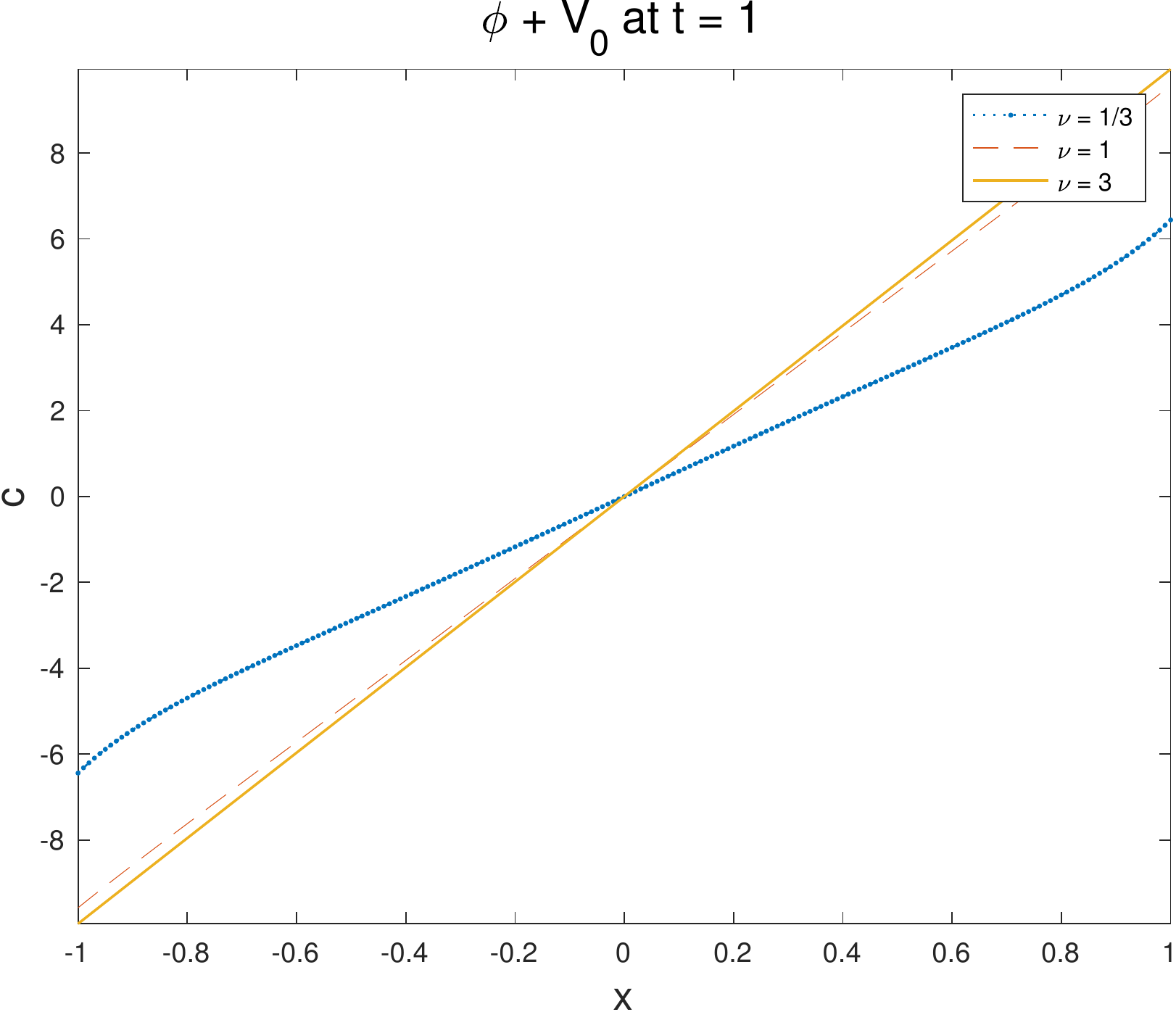}
    \caption{$C_i(x, 1)$ for $i = 1, 2, 3$, $\rho(x, 1), \phi(x, 1)$ and $\phi(x, 1) + V_0(x)$ with $\Delta t = 0.005, \Delta x = 0.01$ and $\nu = 1/3, 1, 3$.}
	\label{plot_test5_nu}
\end{figure}

In the fourth numerical test, the parameters $\eta = 1, \nu = 1, (z_1, z_2, z_3) = (1, -1, 0), (v_1, v_2, v_3) $ $= (0.01, 0.01, 0.01)$. Fig \ref{plot_test5_lambda} shows the corresponding distributions at time $t = 1$ when $\lambda = 1/10, 1, 10$. We can see from Fig \ref{plot_test5_lambda} that the equilibrium is more insensitive to the parameter $\lambda$ than $\nu$. And larger $\lambda$ leads to a stronger internal correlated electric potential $\phi$. 


\begin{figure}[htp]
    \centering
    \includegraphics[width=4.6cm,height=4.8cm]{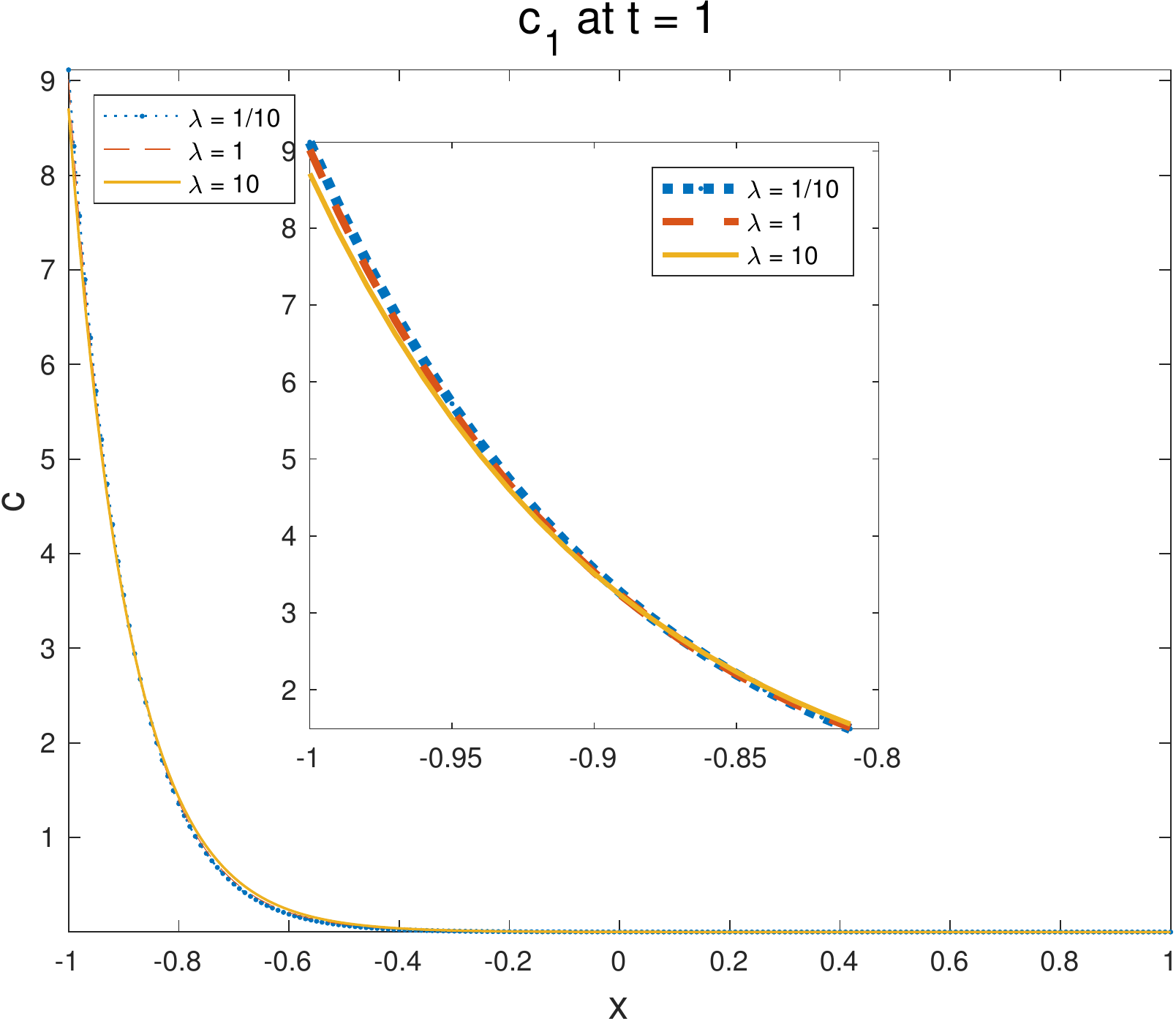}
    \includegraphics[width=4.6cm,height=4.8cm]{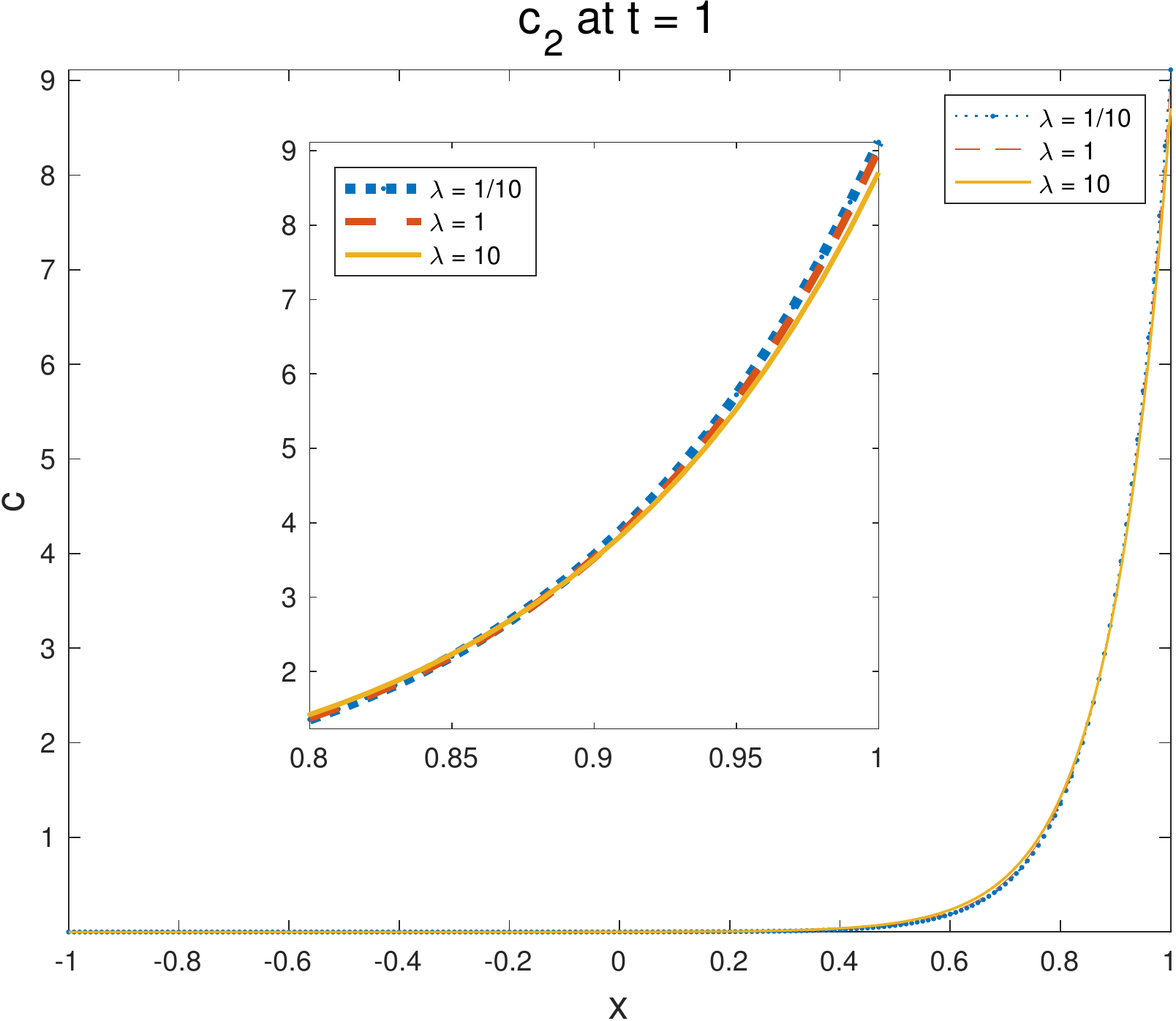}
    \includegraphics[width=4.6cm,height=4.8cm]{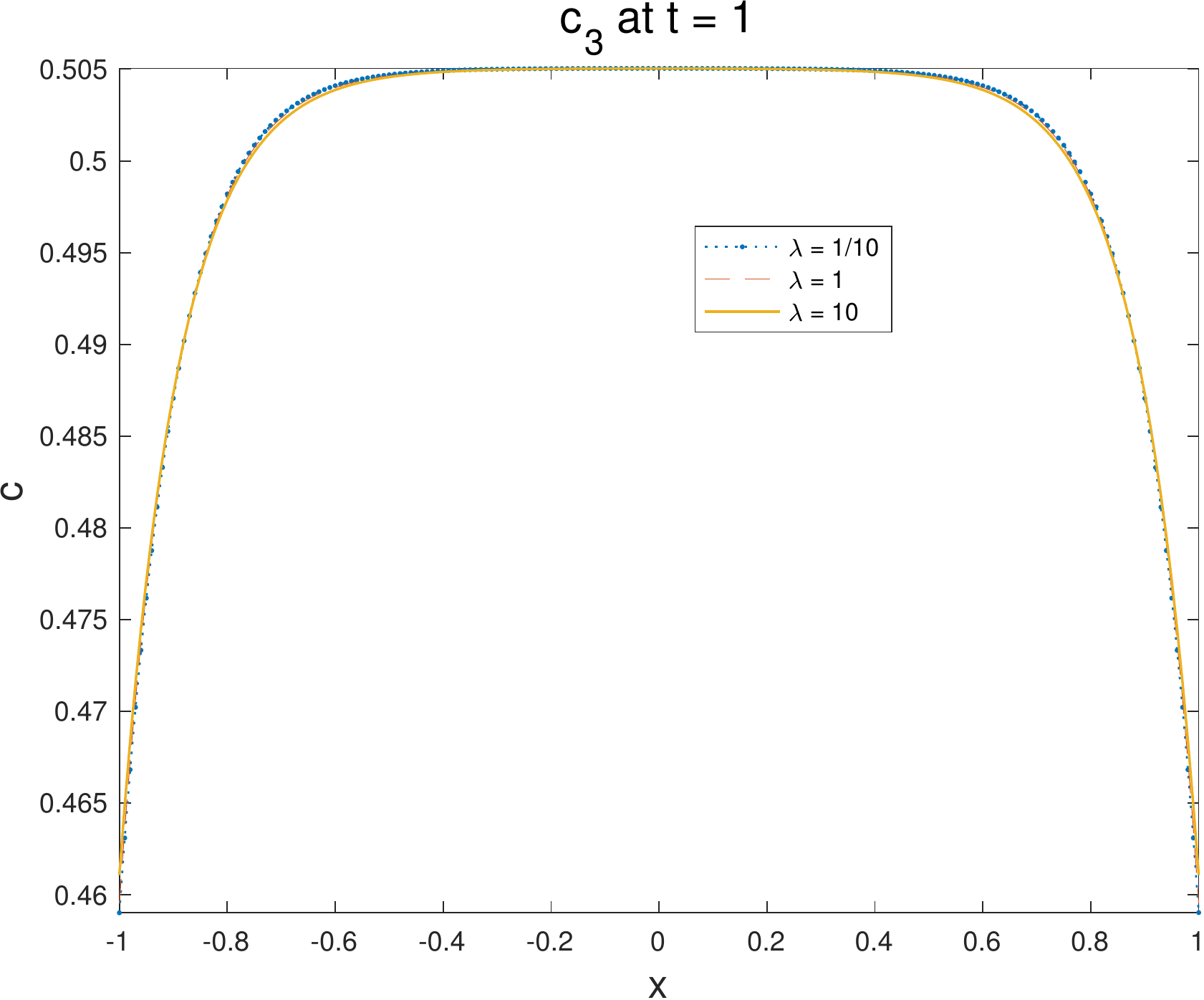}
    
    \includegraphics[width=4.6cm,height=4.8cm]{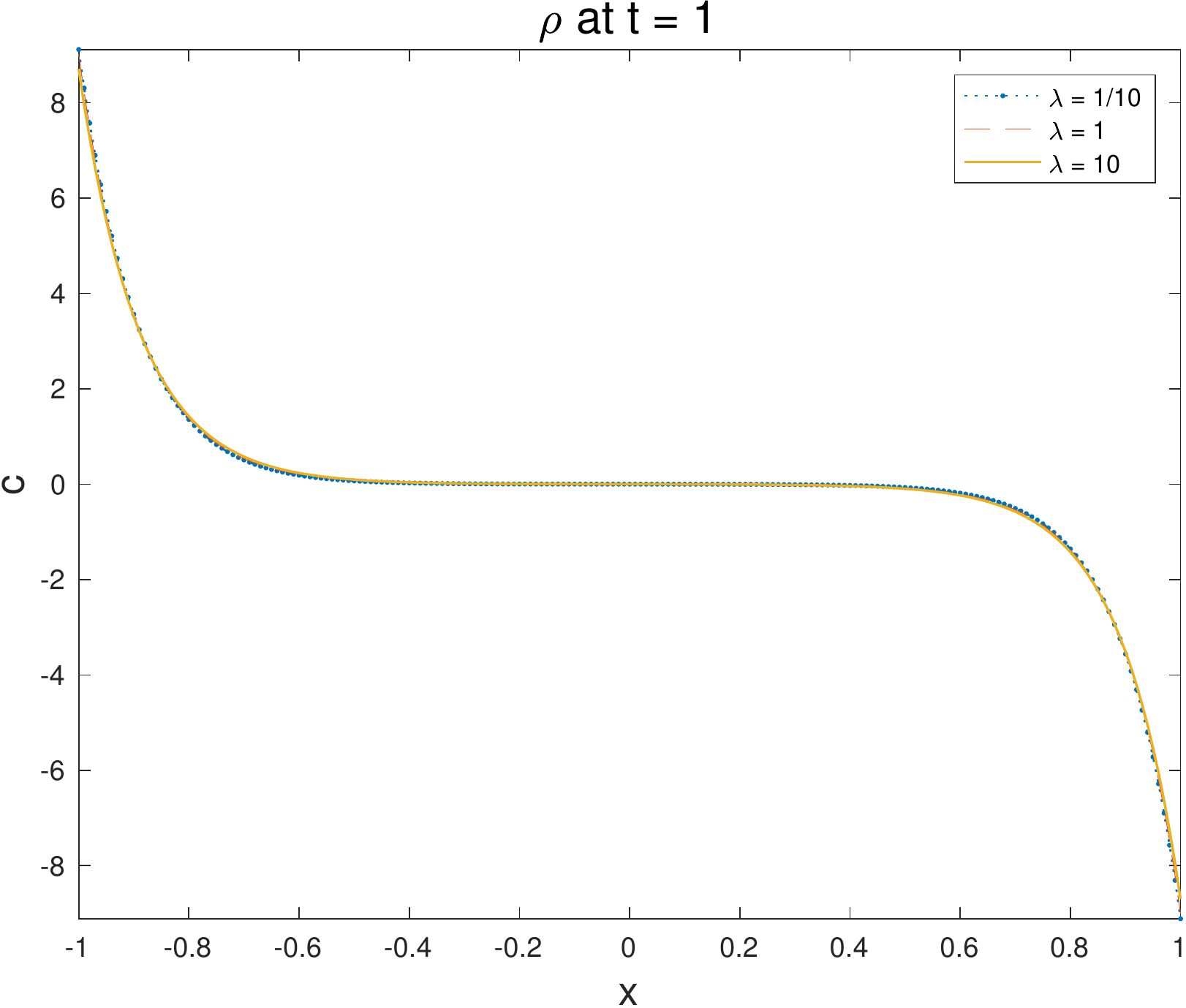}
    \includegraphics[width=4.6cm,height=4.8cm]{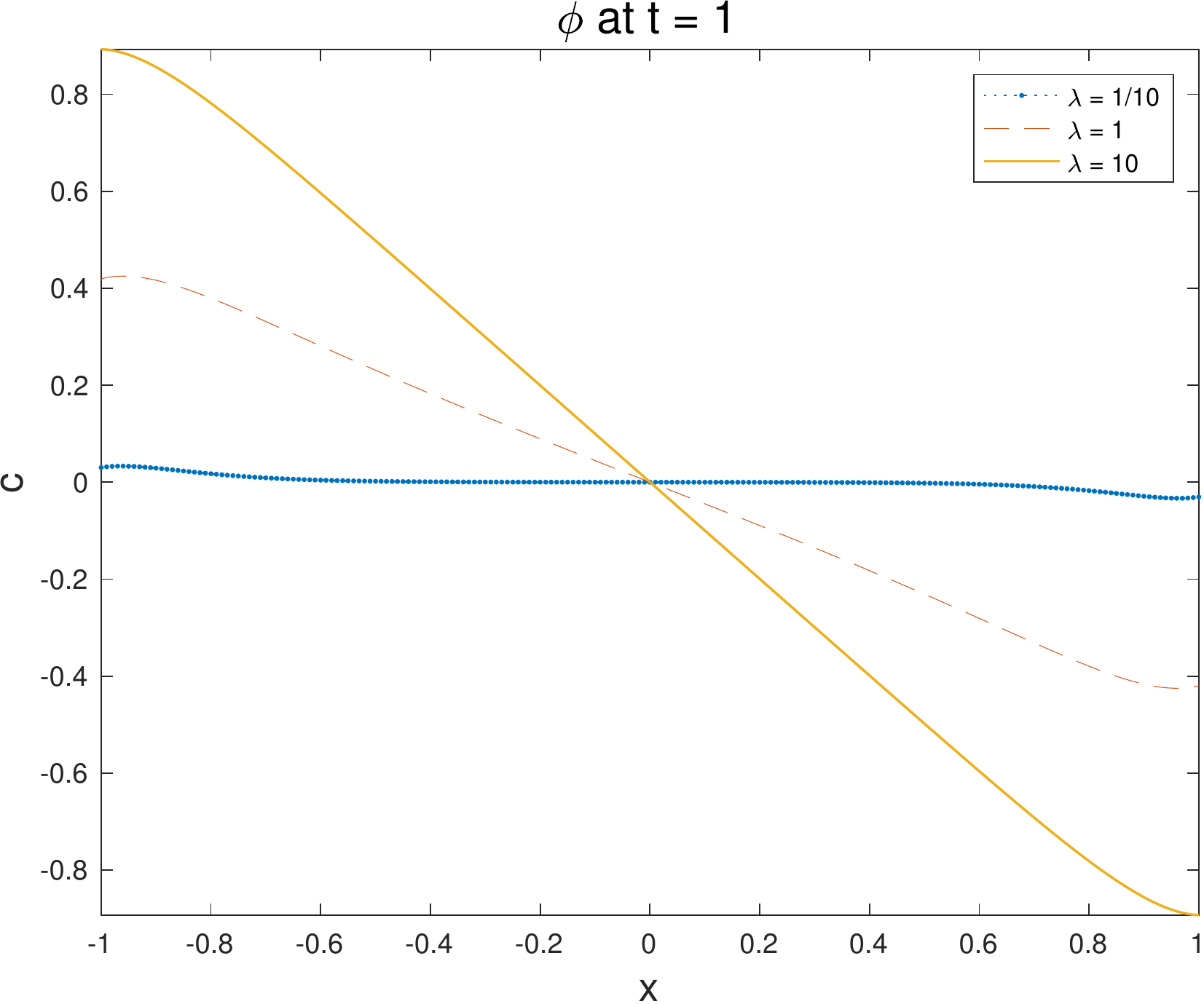}
    \includegraphics[width=4.6cm,height=4.8cm]{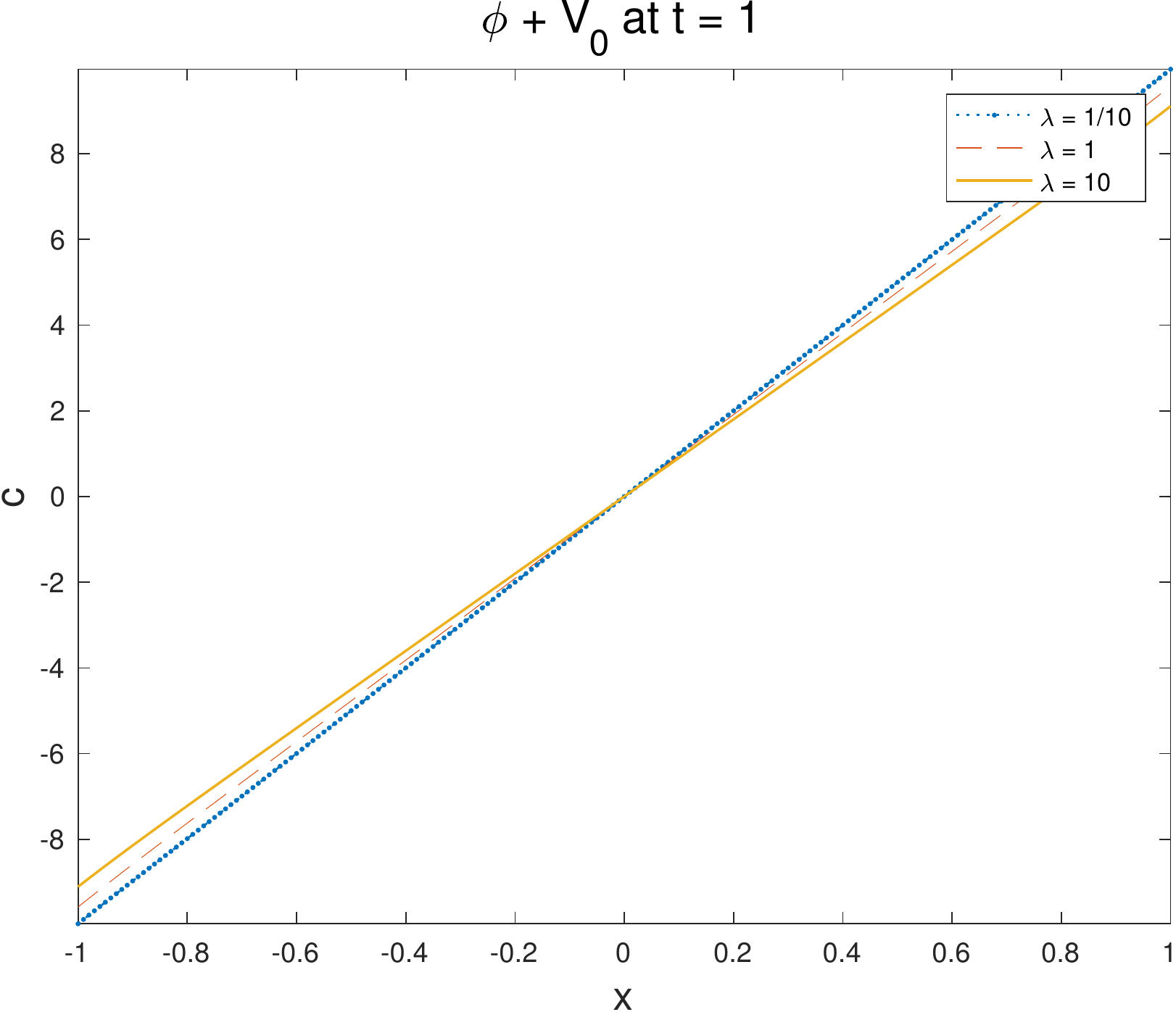}
    \caption{$C_i(x, 1)$ for $i = 1, 2, 3$, $\rho(x, 1), \phi(x, 1)$ and $\phi(x, 1) + V_0(x)$ with $\Delta t = 0.005, \Delta x = 0.01$ and $\lambda = 1/10, 1, 10$.}
	\label{plot_test5_lambda}
\end{figure}

\subsubsection{Parameters related to the intrinsic properties of ions: $z_i, v_i$}
In this part, we consider the valence $z_i$ of the $i^{\text{th}} $ species at first. Due to the Kohlrausch's law, the flux related to the electric part $F_i^{\mathrm{el}}$ is given via the electrostatic potential gradient, i.e. $F_i^{\mathrm{el}} = -D_iz_i C_i \nabla \phi$, which we can see from the system \eqref{model}. Consequently, assuming all other parameters are fixed, 
the electric force felt by the $i^{\text{th}}$ species would be stronger for larger $|z_i|$.

To do numerical investigations, we set the other parameters $\eta = 2, \lambda = 1, \nu = 1, (v_1, v_2, v_3) = (0.01, 0.01, 0.01)$ and the valence of anions $z_2=-1$. The parameter $z_1$ to be considered is set to be $z_1 = 1, 2, 4$. Then Fig \ref{plot_test4_z} depicts the concentrations of the PNPB system and the PNP system (we set $\eta =0$) at time $t = 1$. We observe that, the larger the valence of cations $z_1$ is, more cations are concentrated on the left boundary. The reason is that larger $z_1$ makes the electrostatic effect become larger on the cations, combined with the external field $V_0 = 10 x$, $C_1$ will peak more on left, while it makes little difference to the anions. Besides, when $z_1=1$, the electrolyte is symmetric, the distributions are symmetric, which is also the case in the former four tests. However, when $z_1\neq 1$, i.e., asymmetric electrolyte, no symmetry is guaranteed. 
Similar behaviours can be observed in the PNP system as $z_1$ varies. However, the peaks of $C_1$ in the PNP system are higher compared to those in PNPB and the profiles of $C_3$ remain flat, as there is no steric effect.
All the results consist with the theoretical discussion.


\begin{figure}[htp]
    \centering
    \includegraphics[width=4.6cm,height=4.8cm]{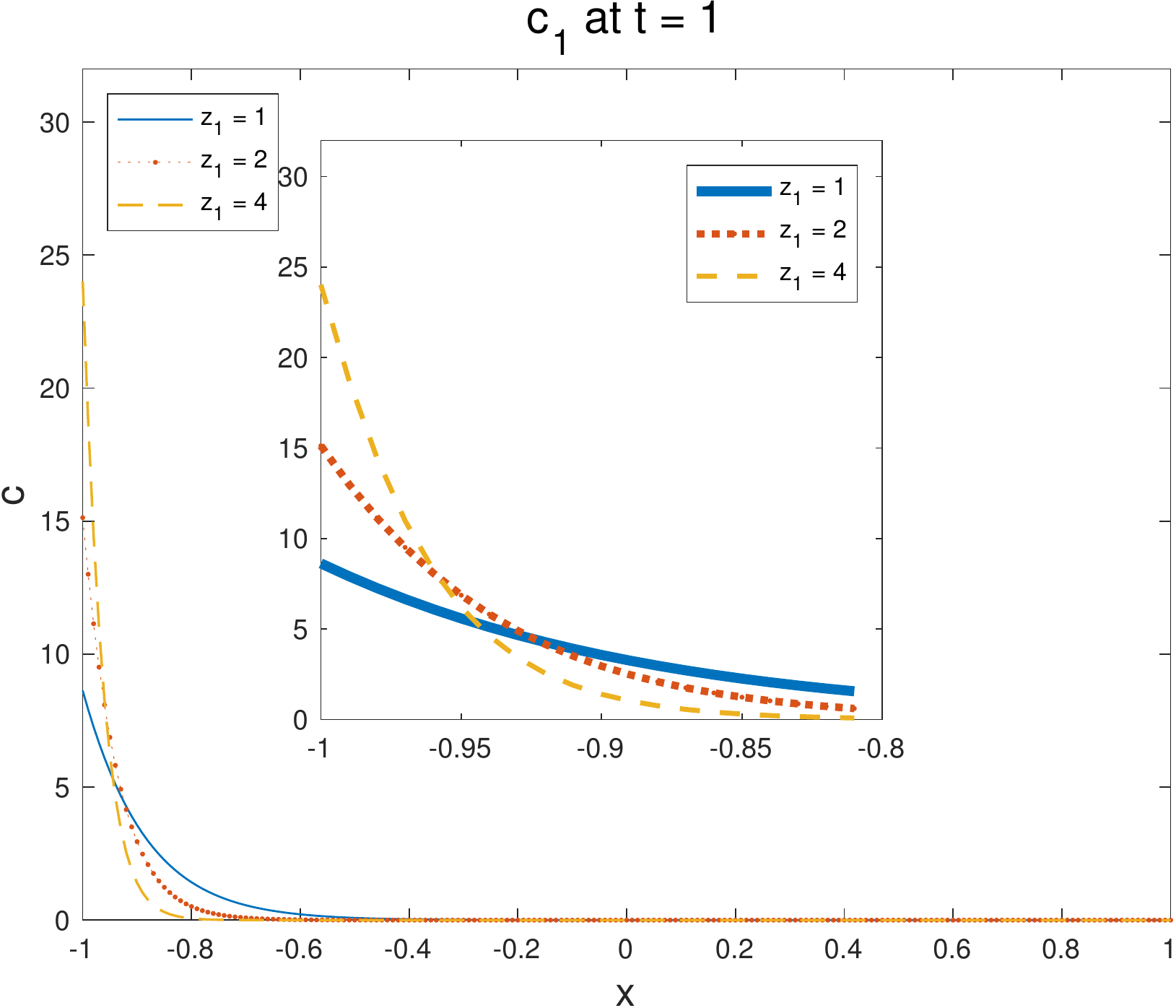}
    \includegraphics[width=4.6cm,height=4.8cm]{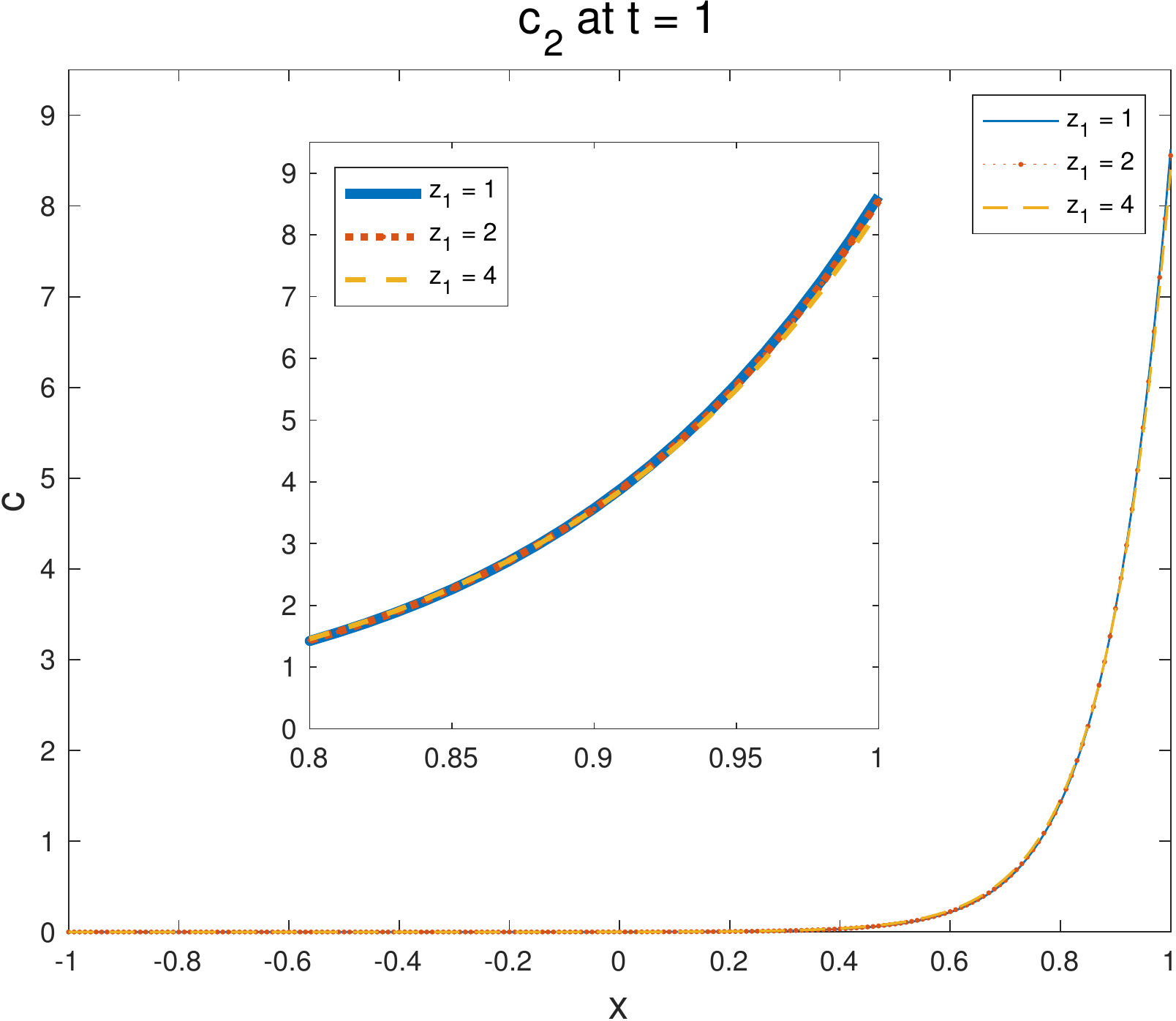}
    \includegraphics[width=4.6cm,height=4.8cm]{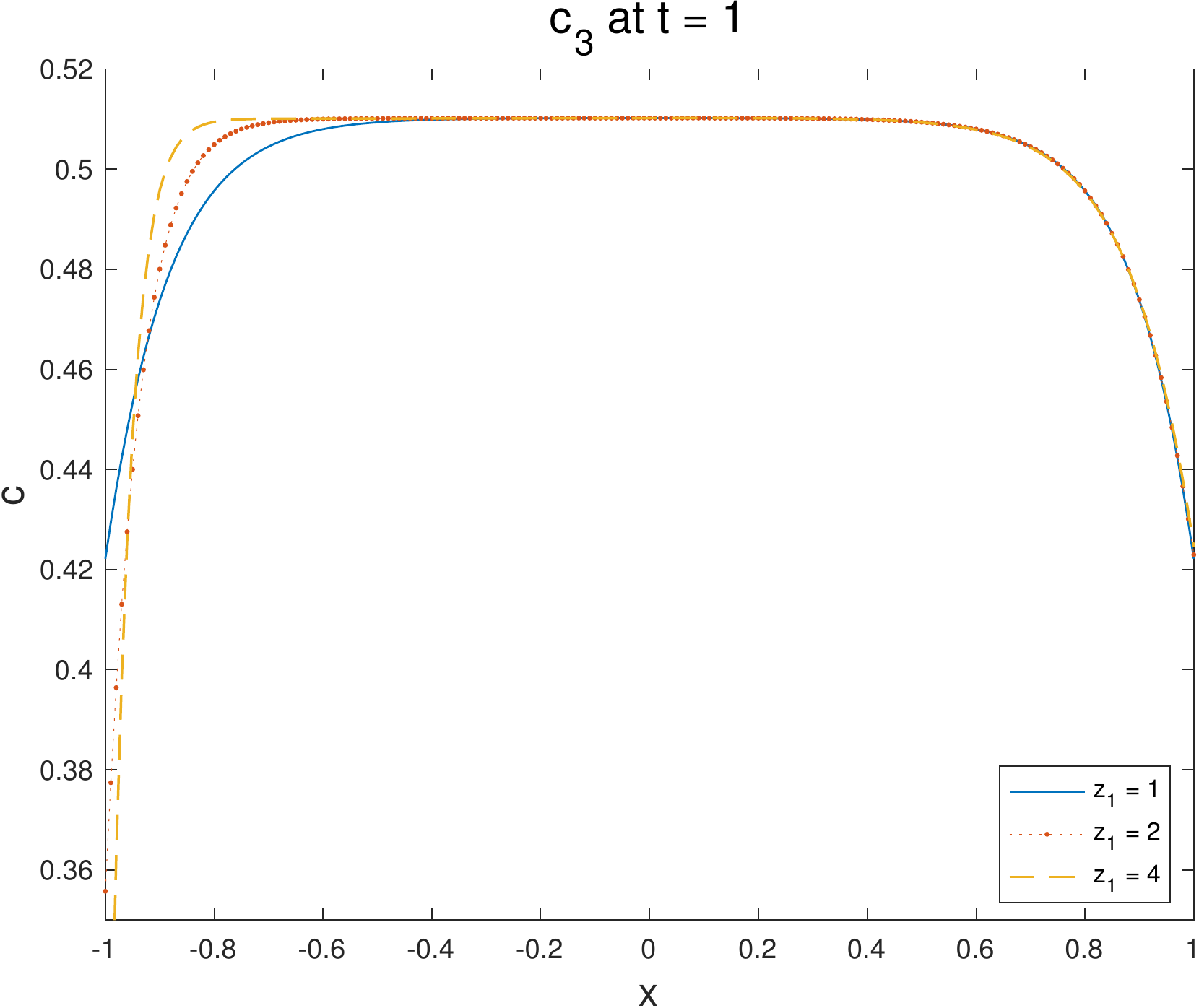}
    
    \includegraphics[width=4.6cm,height=4.8cm]{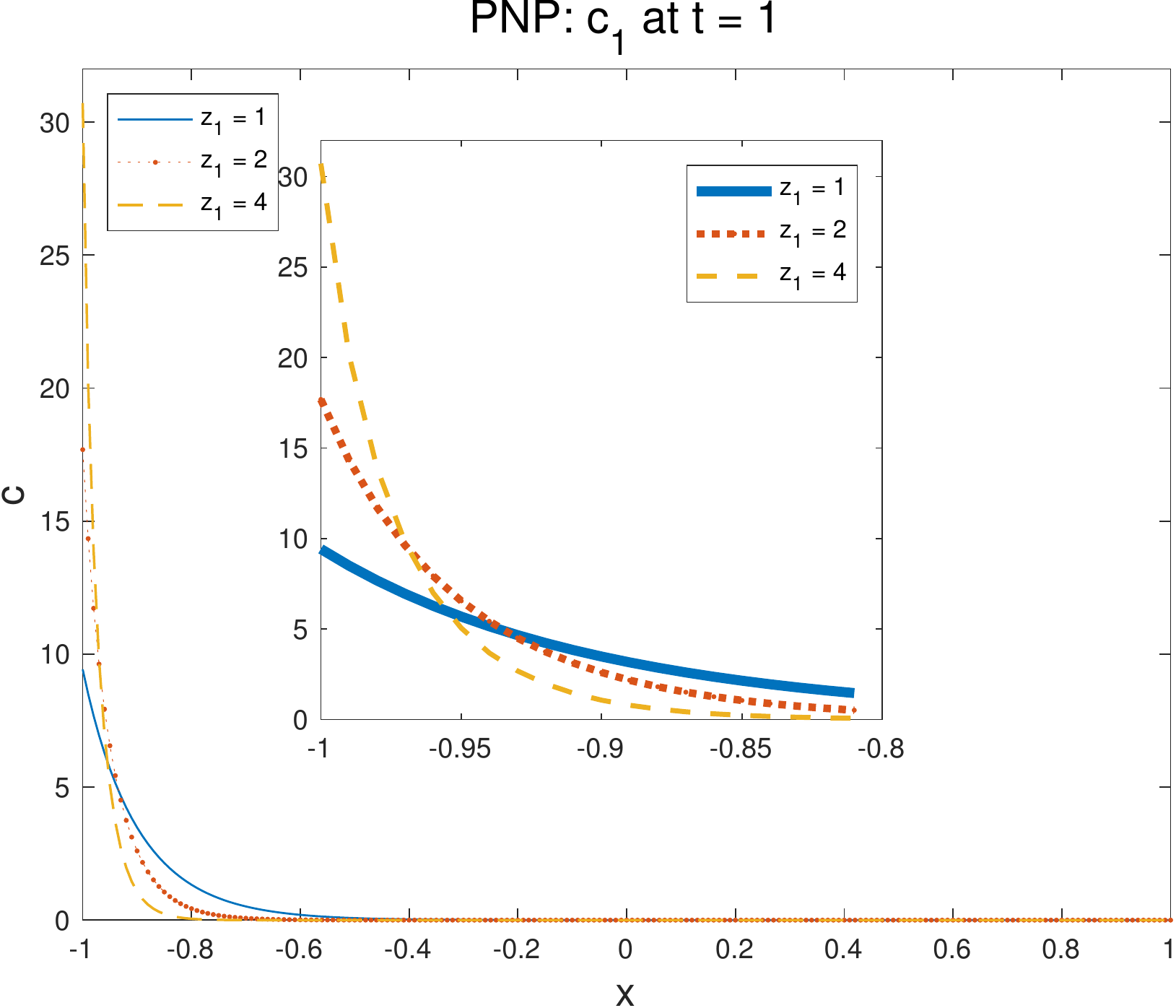}
    \includegraphics[width=4.6cm,height=4.8cm]{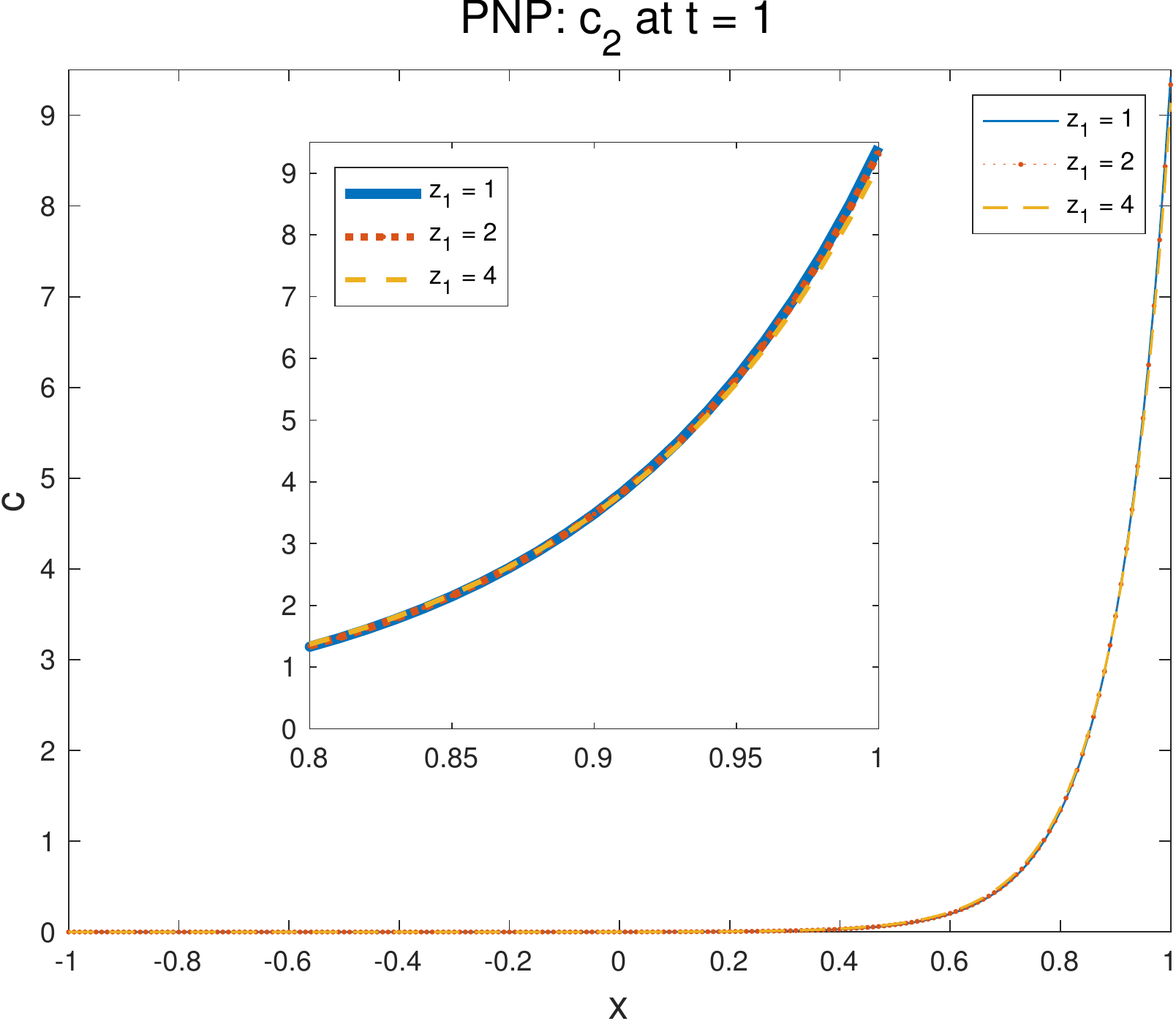}
    \includegraphics[width=4.6cm,height=4.8cm]{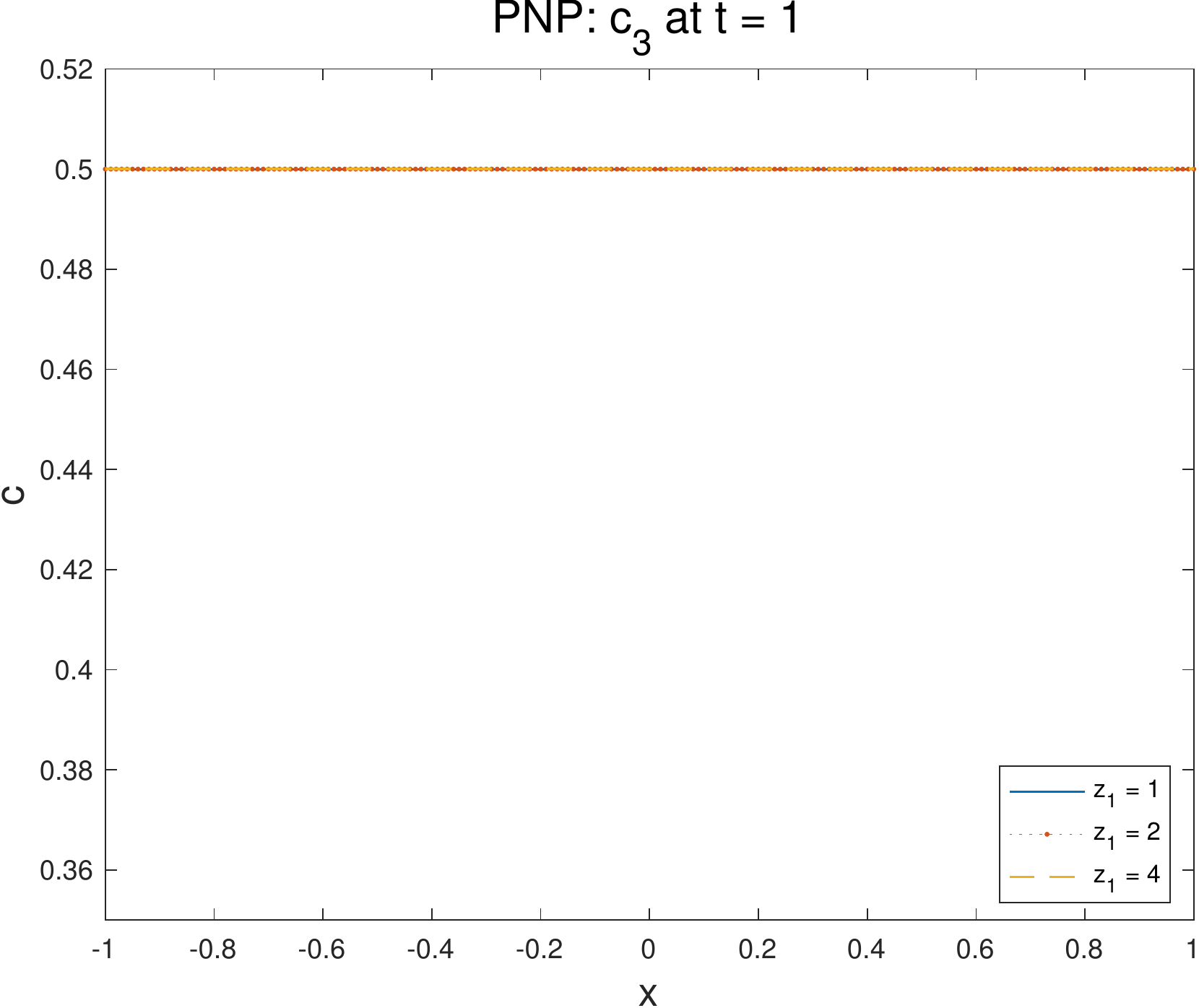}
    \caption{Comparison of the PNPB ($\eta = 2$) and PNP ($\eta = 0$) model: $C_i(x, 1)$ for $i = 1, 2, 3$ with $\Delta t = 0.005, \Delta x = 0.01$ and $z_1 = 1, 2, 4$.}
	\label{plot_test4_z}
\end{figure}

On the other hand, as for the volume $v_i$, the flux related to the steric part $F_i^{\mathrm{trc}}$ is given via the steric potential gradient, i.e. $F_i^{\mathrm{trc}} = D_i\frac{v_i}{v_0}C_i(x,t)\nabla S$. Assuming all other parameters are constant, larger $v_i$ will lead directly to a larger steric effect on the $i^{\text{th}}$ species, hence the larger steric force makes a stronger repulsion, which will smooth the peak of the concentrations.

In this numerical test, let the other parameters $\eta = 1, \lambda = 1, \nu = 1, (z_1, z_2, z_3) = (1, -1, 0)$, the volumes of the anions and water molecules stays $v_2=v_3=0.01$. Test the cases where the volume of cations are $v_1 = 0.01, 0.03, 0.05$, and hence $v_0=0.01, 1/60, 7/300$ respectively. Fig \ref{plot_test3_v} gives the results of the concentrations of the PNPB system and the PNP system (we set $\eta = 0$) at time $t = 1$. We observe that, the larger $v_1$ is, the less $C_1$ peak on the left boundary. This is because larger $v_1$ makes the steric effect become larger on the cations but has little difference on the anions. The steric force is a kind of repulsion force to prevent the particles from getting close to each other, and hence larger $v_1$ makes the concentration $C_1$ become less on the left boundary.
Also, in the case where $v_1, v_2, v_3$ are not the same, the distributions are no longer symmetric.
And the change in volume $v_i$ has no effect on the PNP model as expected.


\begin{figure}[htp]
    \centering
    \includegraphics[width=4.6cm,height=4.8cm]{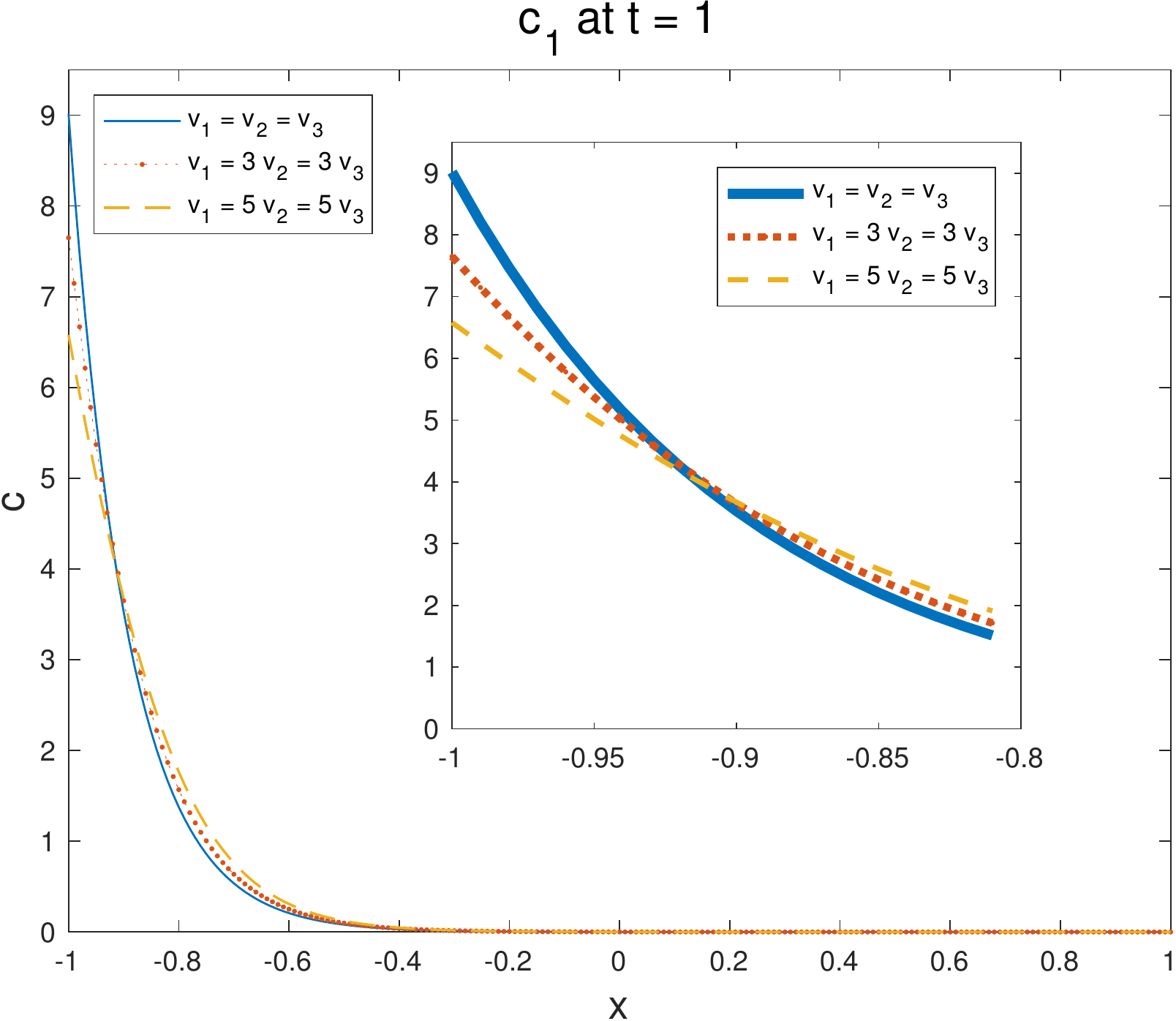}
    \includegraphics[width=4.6cm,height=4.8cm]{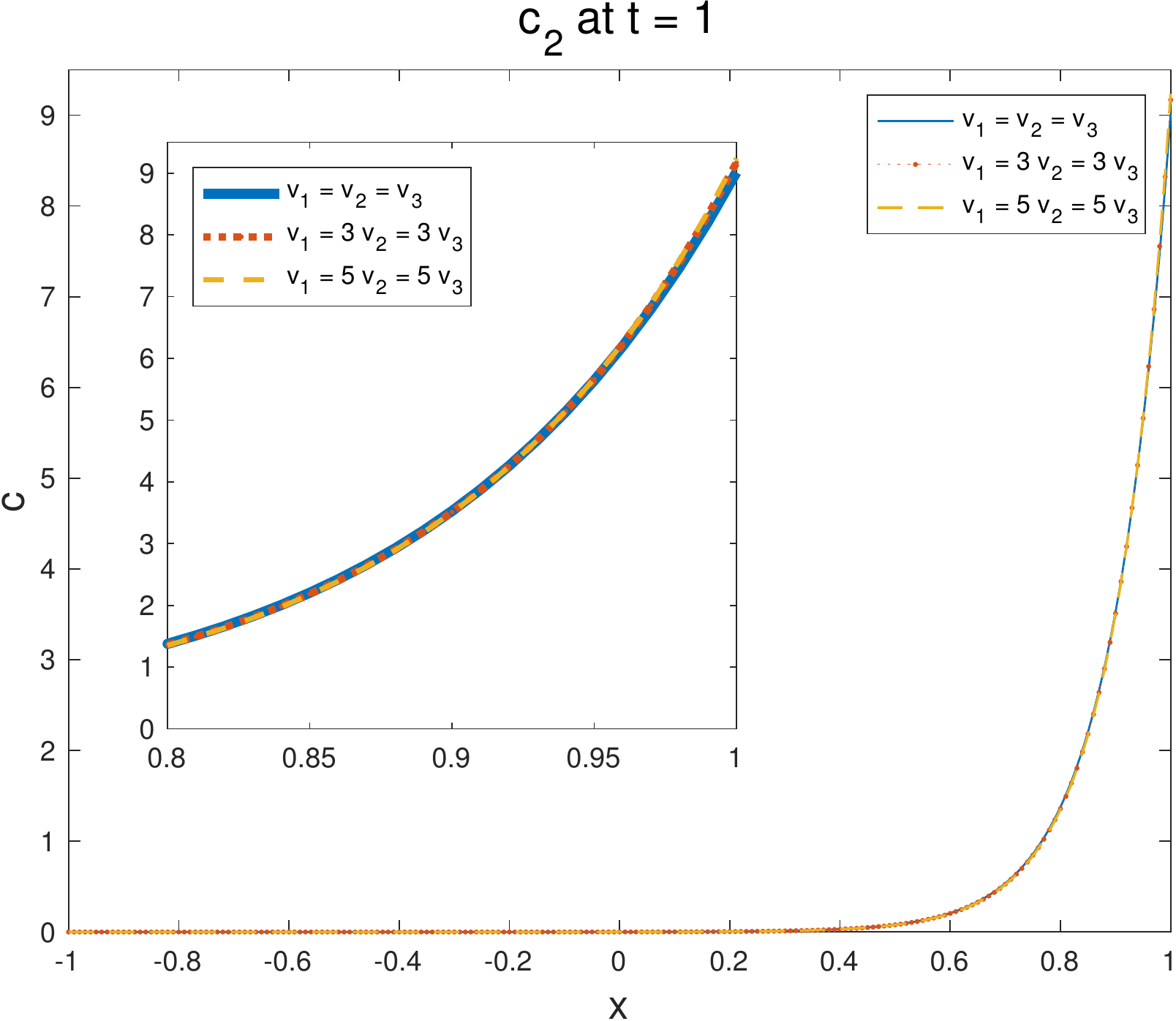}
    \includegraphics[width=4.6cm,height=4.8cm]{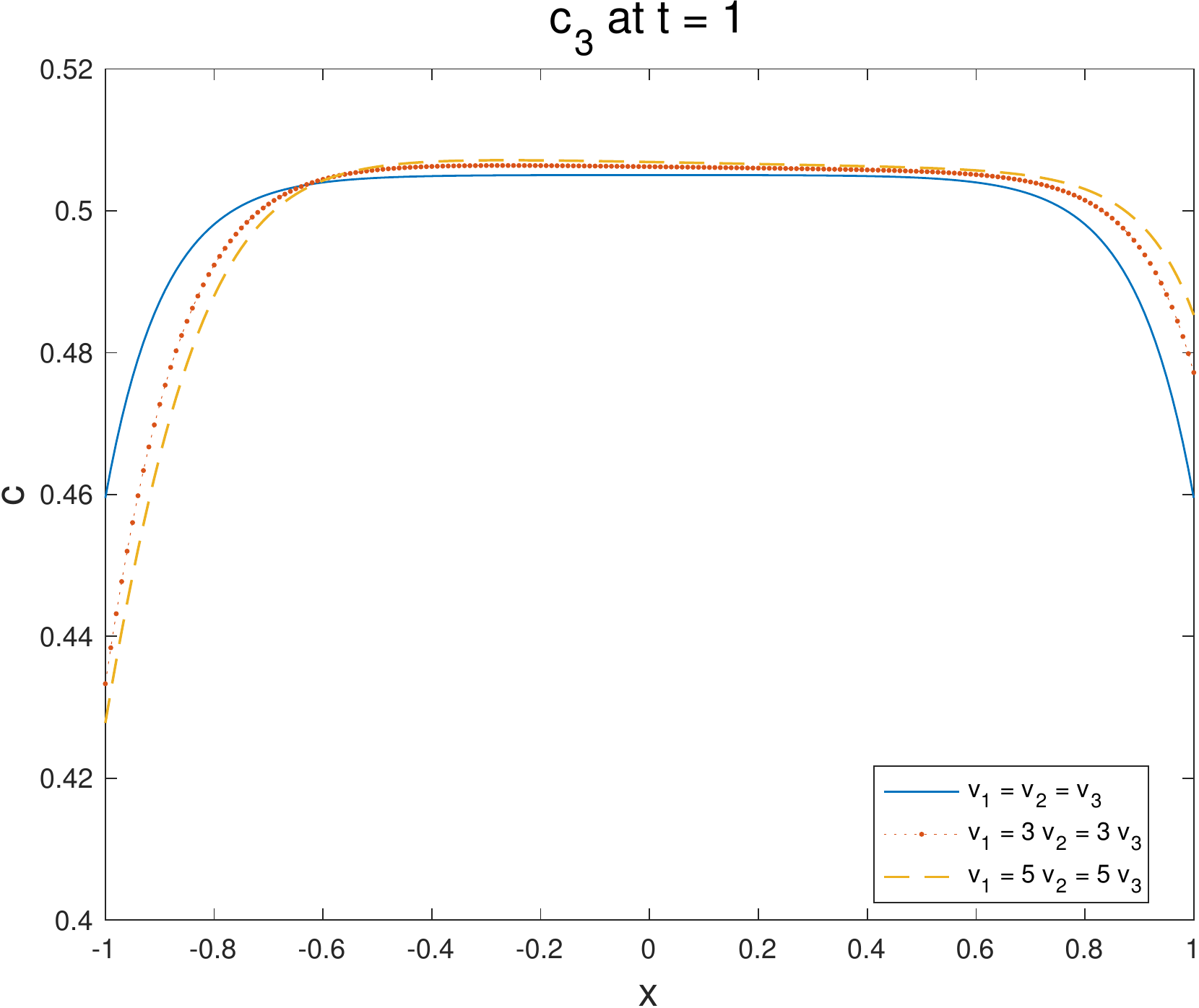}
    
    \includegraphics[width=4.6cm,height=4.8cm]{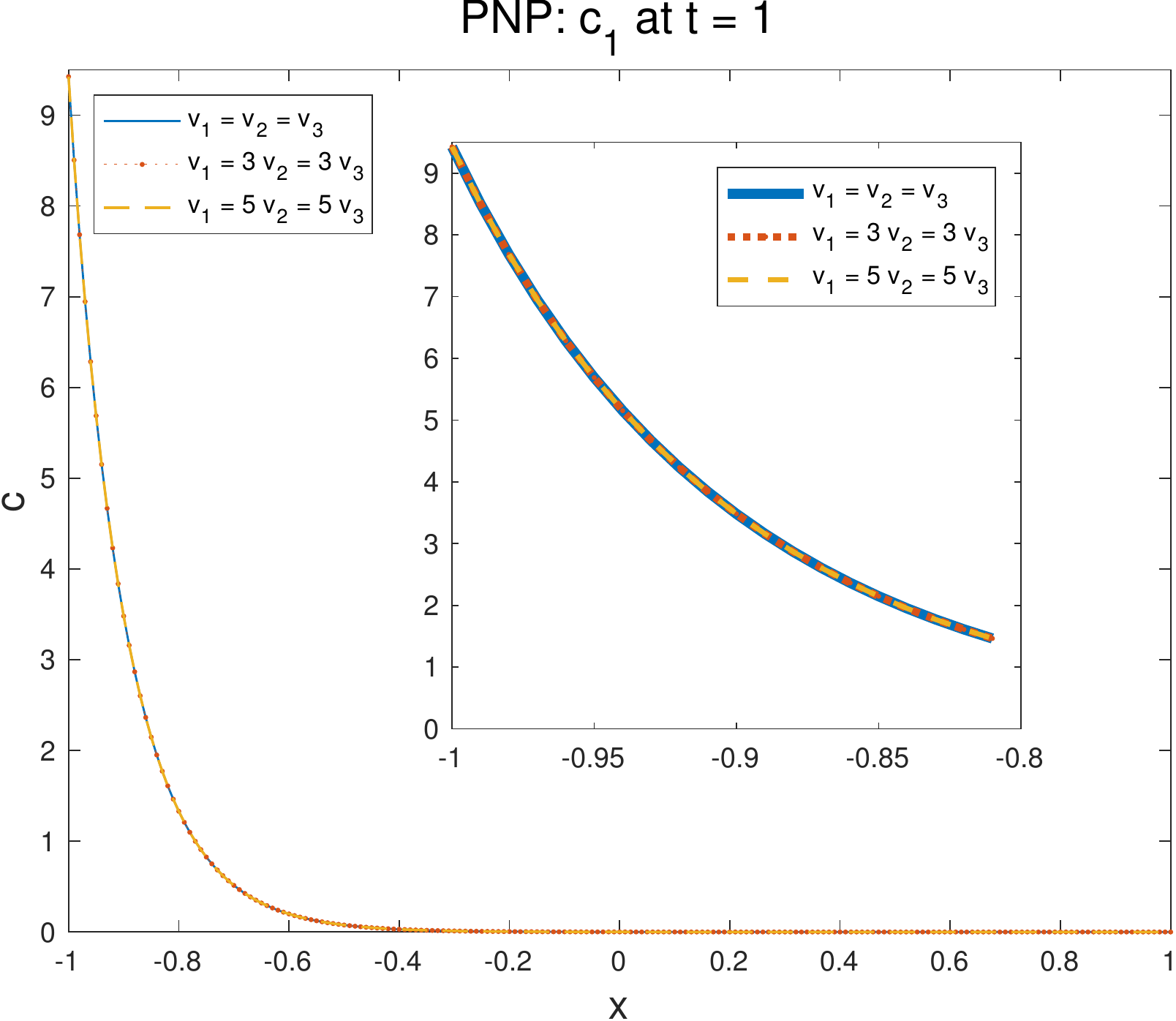}
    \includegraphics[width=4.6cm,height=4.8cm]{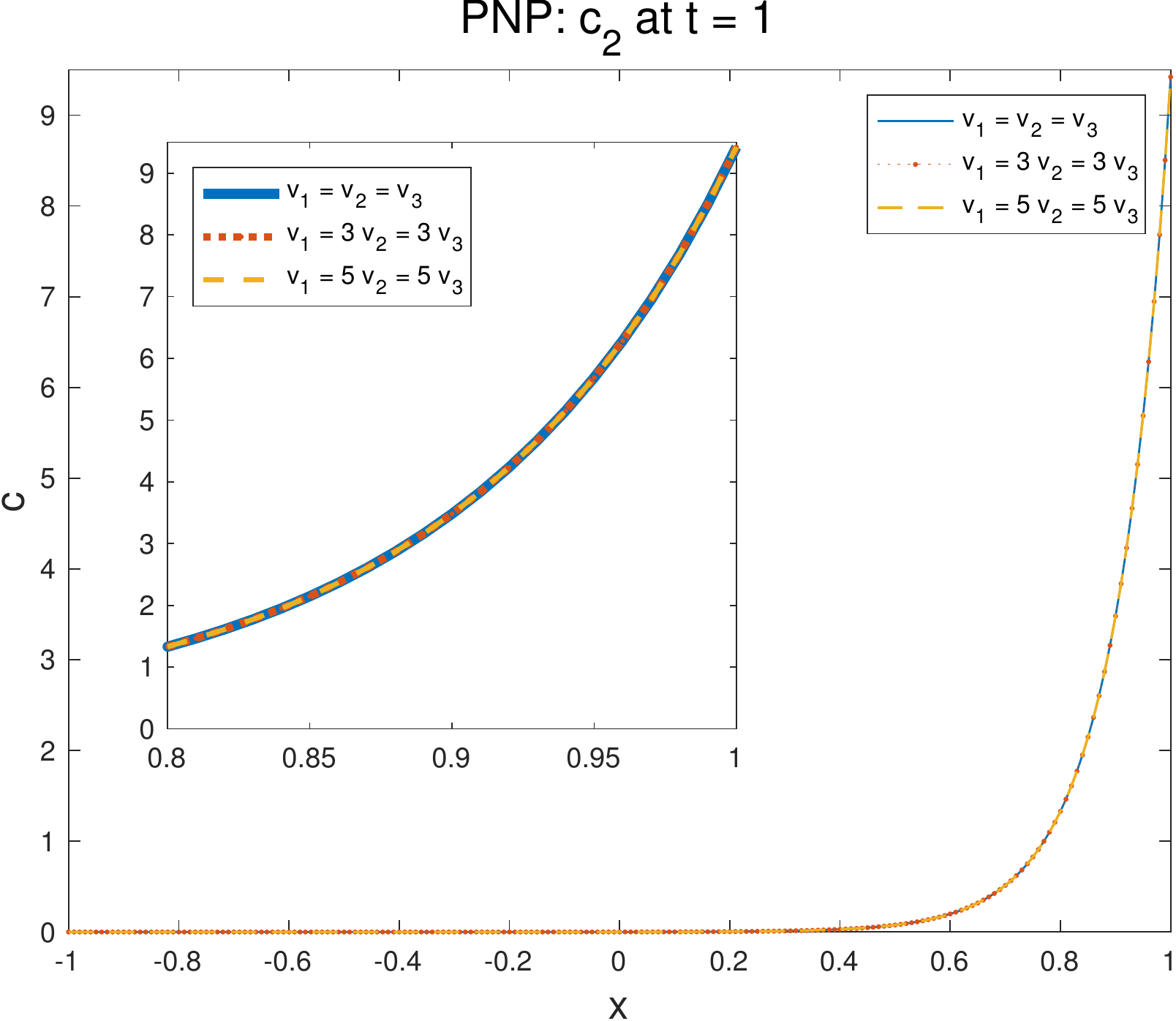}
    \includegraphics[width=4.6cm,height=4.8cm]{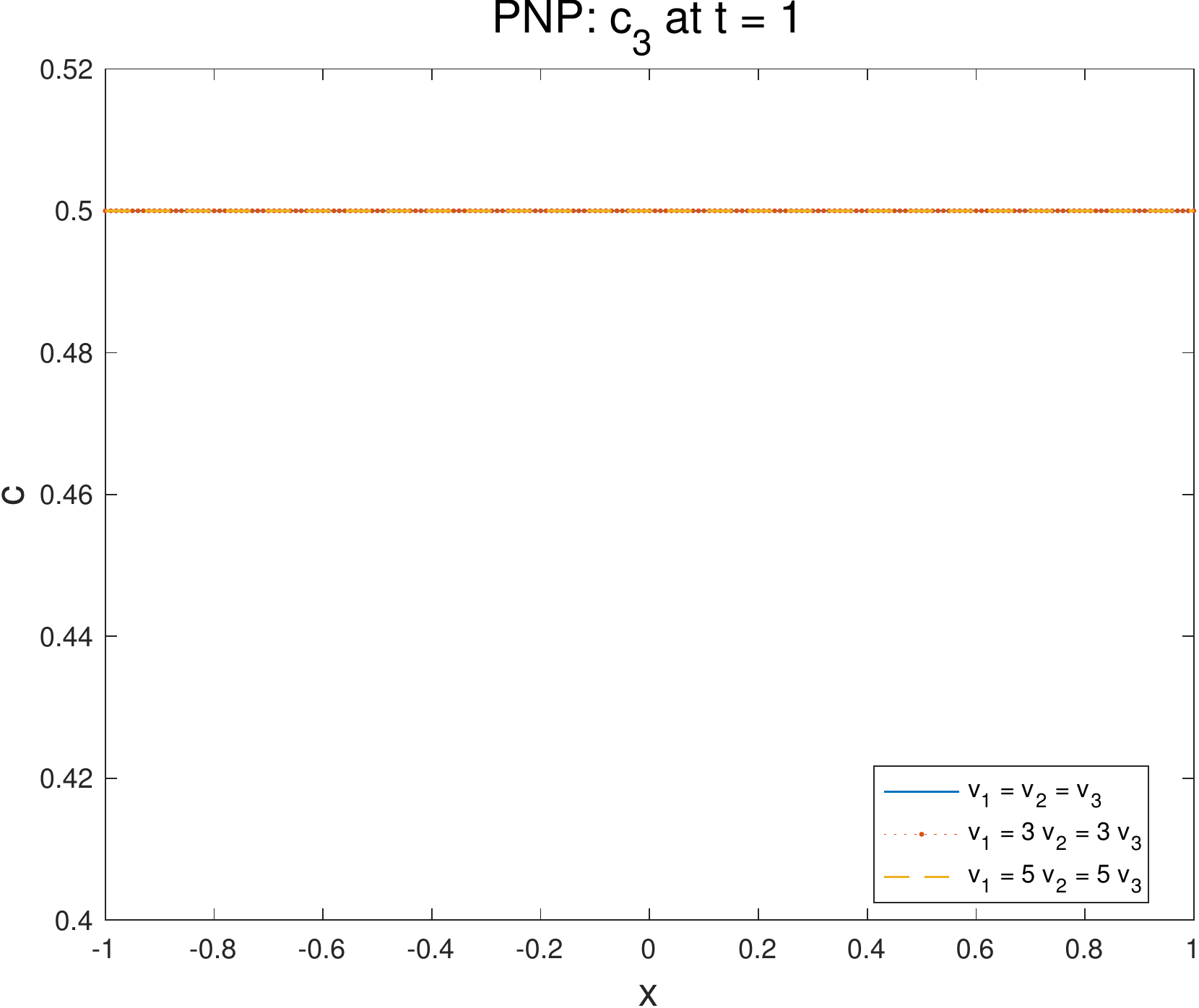}
    \caption{Comparison of the PNPB ($\eta = 1$) and PNP ($\eta = 0$)  model: $C_i(x, 1)$ for $i = 1, 2, 3$ with $\Delta t = 0.005, \Delta x = 0.01$ and $v_1 = 0.01, 0.03, 0.05$.}
	\label{plot_test3_v}
\end{figure}

\subsection{Two-dimensional case} 
It is seen that in 1D, parameters $\eta, \nu, \lambda, z_i, v_i$ have influenced the steady state of the dimensionless PNPB system \eqref{eq: no-dim dynamic}-\eqref{eq: no-dim 4pbik}. In summary, large $\eta$ results in larger steric repulsion effect, if $\eta=0$, the steric effect vanishes, if $\eta$ is bigger than some critical value, the positivity of $\Gamma$ can no longer be maintained. $\lambda$ and $\nu$ have opposite effects on the correlated electric field $\phi$, to be specific, smaller $\nu$ and larger $\lambda$ give rise to stronger electric field $\phi$. As for the intrinsic properties of ions, larger $|z_i|$ and $v_i$ would enhance the electrostatic effect and the steric effect of the $i^{\text{th}}$ species respectively. The same results can be found in the two-dimensional case. Here we focus only on the effect of the parameter $\eta$ on the steady state.

Consider the two-dimensional model, where the electrostatic potential $\phi = \mathcal{K}_{\nu} * \rho$, with $\mathcal{K}_{\nu}(x, y) = - \frac{1}{2 \pi \nu^2} \log(\sqrt{x^2+y^2})$. We choose $K = 2$, $V_0(x, y) = 10 x$, and the initial conditions in the following form,
$$
\begin{aligned}
C_1^0(x, y) &= \frac{40}{\pi} \exp\left(-10\left(\left(x-\frac{1}{5}\right)^2 + \left(y-\frac{1}{5}\right)^2\right)\right),\\
C_2^0(x, y) &= \frac{40}{\pi} \exp\left(-10\left(\left(x+\frac{1}{5}\right)^2 + \left(y+\frac{1}{5}\right)^2\right)\right), \\
C_3^0(x, y) &= \frac{40}{\pi} \exp\left(-10\left(x^2 + y^2\right)\right). 
\end{aligned}
$$
Furthermore, take the parameters $\nu = 1, (D_1, D_2, D_3) = (1, 1, 1), (z_1, z_2, z_3) = (1, -1, 0)$, \\
$(v_1, v_2, v_3) = (0.01, 0.01, 0.01)$, and hence $v_0 = 0.01$. Let $\eta = 1$, Fig \ref{plot_2d_new} shows the results of the concentrations $C_i$ and the chemical potential $\mu_i$ at sufficiently large time $t = 3$. It's seen that the chemical potential becomes constant at the steady state and the boundary layers exist on the left and right boundaries due to the external field $V_0$ in the $x$-direction. 

Next, consider the parameter $\eta = 0, 1, 3$, the concentrations in the $x$-direction $\int_{\Omega} C_1(x, y) \mathrm{d} y$ and $\int_{\Omega} C_2(x, y) \mathrm{d} y$ are shown in Fig \ref{plot_2d_eta_new}. Larger $\eta$ leads to larger steric repulsion effect and thus the the cations are less gathered on the corresponding left boundary while the anions are less gathered on the right. All the conclusions are the same as the one-dimensional case.

\begin{figure}[htp] 
	\centering
	\includegraphics[width=1\textwidth]{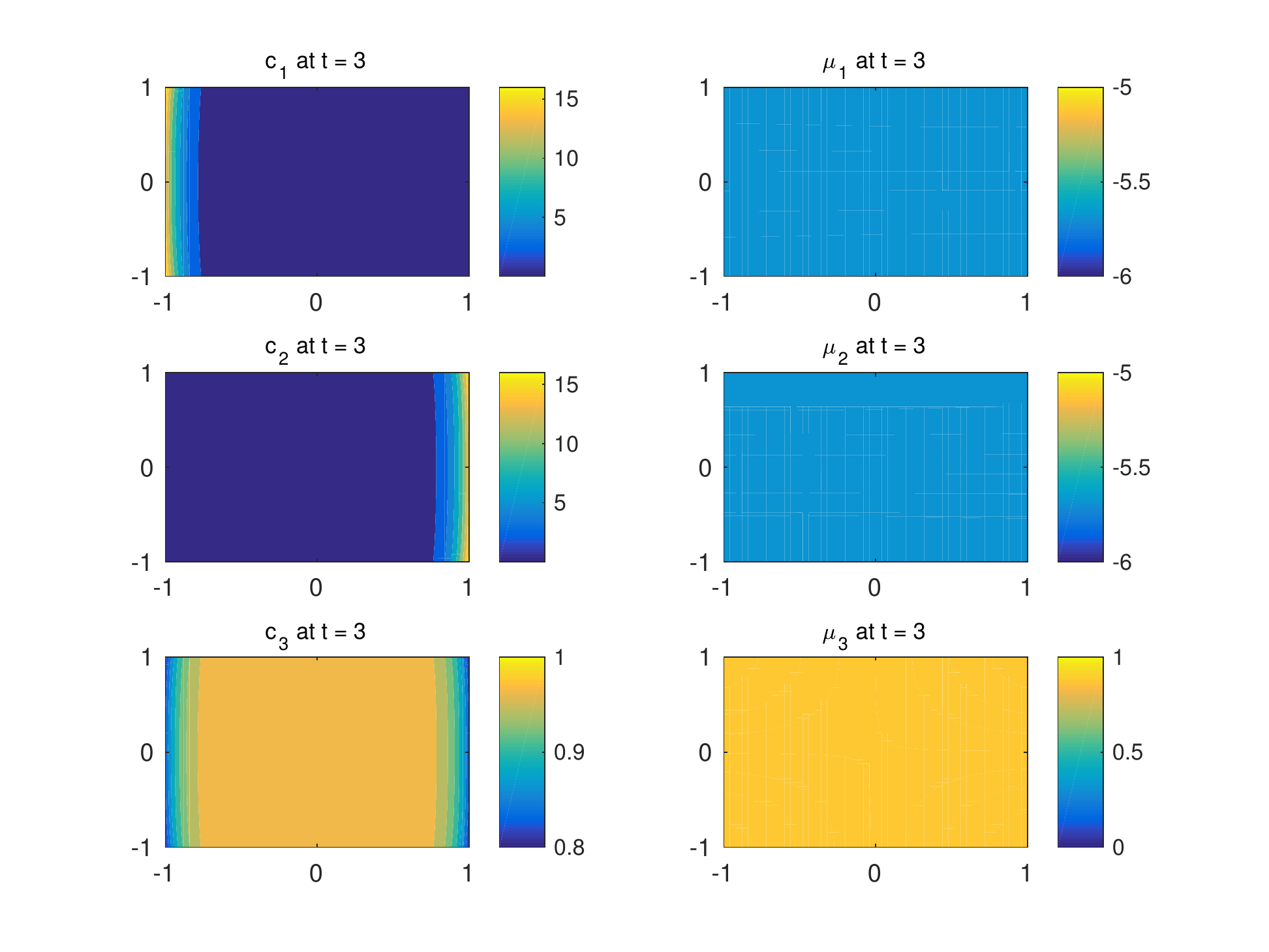}
	\caption{The concentration $C_i(x, y)$ and the chemical potential $\mu_i(x, y)$ for $i = 1, 2, 3$ with $\Delta t = 0.001, \Delta x = \Delta y = 0.04$ and $\eta = 3$.}
	\label{plot_2d_new}
\end{figure}

\begin{figure}[htp] 
	\centering
		\includegraphics[width=0.45\textwidth]{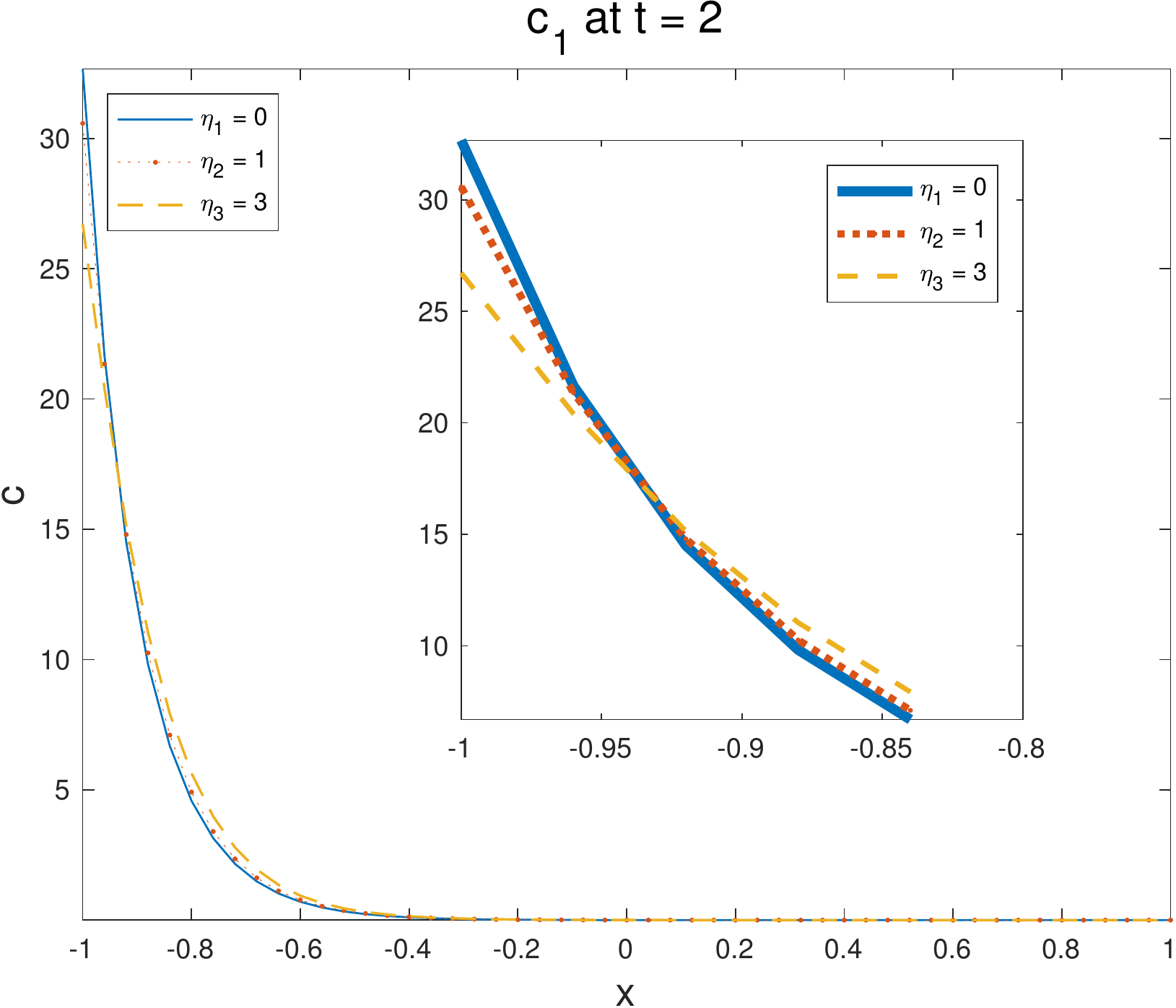}
		\includegraphics[width=0.45\textwidth]{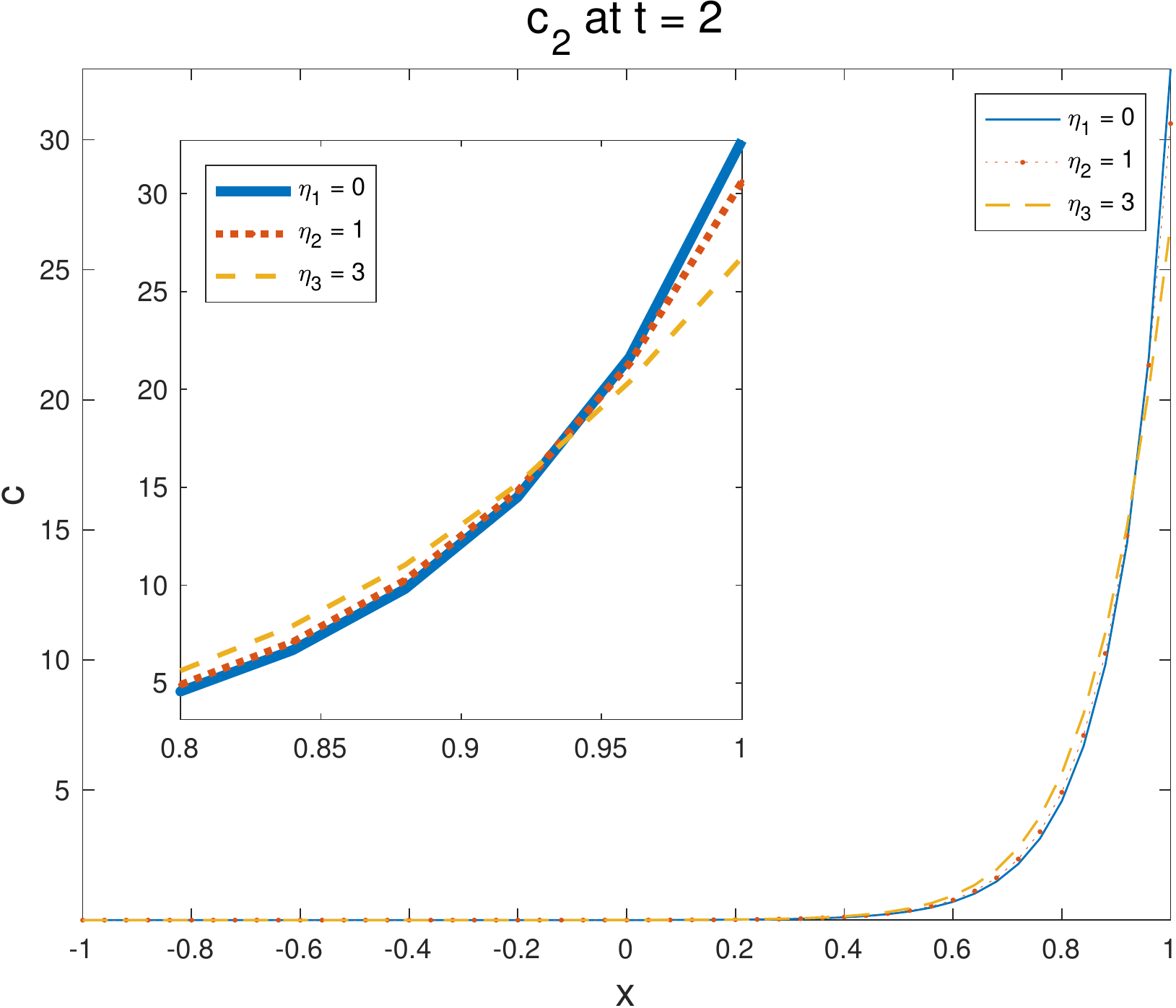}
	\caption{$\int_{\Omega} C_i(x, y) \mathrm{d} y$ for $i = 1, 2$ with $\Delta t = 0.001, \Delta x = \Delta y = 0.04$ and $\eta = 0, 1, 3$.}
	\label{plot_2d_eta_new}
\end{figure}

\section{Conclusions}
\label{sec: conclusion}

The PNPB model for ionic solution differs from the classical PNP and PB theories in many aspects. In the model setup, it treat all the ions with sizes and valences rather than volumeless. And it consider water molecules as another species with volume together with voids. In the sense of electric potential, the governing equation is the 4PBik equation instead of the Poisson equation so that it can describe the correlation effect. What's more, the steric potential is included which characterizes the steric effect quantitatively. 

In this work, we clarify some theoretical aspects of the steady state problem, i.e., the 4PBik equation with Fermi-like distribution. We give the self-adjointness and kernel of the fourth-order operator in the 4PBik equation. We also provide four equivalent statements for the steady state.
Besides, we prove the positivity of the void volume function in the steady state, otherwise, the steric potential is meaningless. Finally, we can obtain the well-posedness of the steady state owing to the convexity of the free energy functional.
Furthermore, we give some 1D and 2D tests to get intuitive impressions on the steric (finite size) effect. 

The understandings of the dynamical problem on the PNPB model is limited. The well-posedness and other analytical issues such as the positivity preserving property of the dynamical problem are left for our future work.

\section*{Acknowledgement}
We would like to thank Prof. Bob Eisenberg for his distinguished lectures on this topic at Duke Kunshan and helpful discussions with us and bringing the PNPB theory to us.  
We were grateful to Prof. Zhennan Zhou for fruitful discussions and comments and regular meetings on this work. Meanwhile, Y. Tang and Y. Zhao highly appreciated his guidance. The work of J.-G. Liu was supported by NSF grant DMS-2106988. 

\appendix
\section{Appendix}
\begin{thm}[Dunford-Pettis Theorem\cite{Haim_Brezis}] \label{thm: DP}
Assume that $\Omega$ is bounded. Let $\mathcal{F}$ be a bounded set in $L^1(\Omega)$. Then $\mathcal{F}$ has compact closure in the weak topology if and only if $\mathcal{F}$ is uniformly integrable: $\forall~ \epsilon>0$, $\exists~ \delta>0$ such that
$$
\int_{A} |f|<\epsilon, \qquad \forall A \subset \Omega,~\text{measurable with}~~ |A|<\delta, \quad \forall f\in \mathcal{F}.
$$
\end{thm}

\begin{lem}[Mazur lemma\cite{Yosida}] \label{lem: Mazur}
Let $(X,\|\cdot\|)$ be a Banach space and $\{u_k\}$ a sequence in $X$ that converges weakly to $u_0\in X$. Then there exist convex combinations $v_k$ of $u_{N+1}, ..., u_{N+k}$
for all $k \geq 1$ such that $\|v_k - u_0\|_X\rightarrow 0$.
\end{lem}

\bibliographystyle{amsplain}
\bibliography{rsRef.bib}

\end{document}